\newif\ifllncs

\ifllncs
    \documentclass{llncs}
\else
    \documentclass{article}
    \usepackage{fullpage}
\fi
\usepackage{diagbox}
\usepackage{graphicx} 
\usepackage{amsmath,amsfonts,amssymb}

\ifllncs
\else
    \usepackage{amsthm}
\fi

\usepackage{graphicx} 
\usepackage{hyperref}
\usepackage[nameinlink]{cleveref}
\usepackage[T1]{fontenc}
\usepackage{physics}
\usepackage{todonotes}
\usepackage{mdframed} 
\usepackage{lipsum}
\usepackage{multicol}
\usepackage{bbm}
\usepackage{longfbox}
\usepackage{enumitem}
\usepackage{float} 
\usepackage{nicefrac} 

\allowdisplaybreaks

\hypersetup{colorlinks={true},linkcolor={blue},citecolor=magenta}

\usepackage{comment}

\usepackage[style=alphabetic,minalphanames=3,maxalphanames=4,maxnames=99,backref=true]{biblatex}
\ifllncs

\spnewtheorem{myclaim}[theorem]{Claim}{\bfseries}{\itshape}
\Crefname{myclaim}{Claim}{Claims}

\renewenvironment{proof}[1][Proof]
{\par\noindent\textit{#1. }}
{\hfill$\square$\par}

\else
\newtheorem{theorem}{Theorem}

\newtheorem{definition}[theorem]{Definition}
\newtheorem{lemma}[theorem]{Lemma}

\newtheorem{myclaim}[theorem]{Claim}

\newtheorem{corollary}[theorem]{Corollary}

\Crefname{fact}{Fact}{Facts}
\fi

\newcommand{\ie}{i.e.,\ }

\newcommand{\eg}{e.g.,\ }

\newcommand{\N}{\mathbb{N}}

\renewcommand{\bra}[1]{\langle#1\rvert}
\renewcommand{\braket}[2]{\langle #1 \mid #2 \rangle}
\renewcommand{\ket}[1]{\lvert#1\rangle}
\newcommand{\set}[1]{\{ #1 \}}
\newcommand{\bit}{\{0,1\}}

\newcommand{\cC}{{\mathcal C}}

\newcommand{\cP}{{\mathcal P}}

\newcommand{\cR}{{\mathcal R}}

\newcommand{\bfa}{\mathbf{a}}
\newcommand{\bfb}{\mathbf{b}}

\newcommand{\eps}{\varepsilon}

\newcommand{\secp}{{\lambda}}
\newcommand{\poly}{\mathsf{poly}}

\newcommand{\negl}{\mathsf{negl}}

\newcommand{\Image}{\operatorname{Im}}

\newcommand{\TD}{\mathsf{TD}}

\newcommand{\Haar}{\mathcal{H}}

\renewcommand{\ketbra}[2]{\ket{#1}\bra{#2}}
\renewcommand{\braket}[2]{\langle #1 | #2 \rangle}

\newcommand{\good}{\mathsf{Good}}

\newcommand{\wt}{\widetilde}

\renewcommand{\bra}[1]{\langle#1\rvert}
\renewcommand{\braket}[2]{\langle #1 | #2 \rangle}
\renewcommand{\ket}[1]{\lvert#1\rangle}
\renewcommand{\ketbra}[2]{\ket{#1}\!\bra{#2}}
\newcommand{\proj}[1]{\ket{#1}\!\bra{#1}}
\newcommand{\reg}[1]{\mathsf{#1}}

\newcommand{\secparam}{\lambda}

\newcommand{\id}{\mathrm{id}}

\newcommand{\random}{\overset{{\scriptscriptstyle\$}}{\leftarrow}}

\newcommand{\hybrid}{\mathsf{Hybrid}}
\newcommand{\haarstates}{\Haar}
\newcommand{\haarunitaries}{\mu}

\newcommand{\expect}{\mathbb{E}}
\newcommand{\Adversary}{{\cal A}}
\newcommand{\setsofsets}{\mathfrak{S}}
\newcommand{\setsofrels}{\mathfrak{R}}
\newcommand{\rel}{\mathfrak{R}}
\newcommand{\inj}{\mathsf{inj}}
\newcommand{\cfree}{\mathsf{cf}}
\newcommand{\cfreeset}{\mathsf{CF}}
\newcommand{\pcfpr}[2]{\mathsf{pcf}_{{#1},{#2}}\mathsf{PR}}
\newcommand{\pr}{\mathsf{PR}}

\newcommand{\Vpart}{{V^{\mathsf{part}}}}
\newcommand{\Vpair}{{V^{\mathsf{pair}}}}
\newcommand{\Vfunc}[1]{{V^{\mathsf{func},{#1}}}}
\newcommand{\dist}{\mathrm{dist}}

\newcommand{\CorX}{\mathrm{CorX}}

\newcommand{\oracle}{\mathcal{O}}

\title{Pseudorandomness~in~the~(Inverseless)~Haar~Random~Oracle~Model}

\ifllncs
    \author{Prabhanjan Ananth\inst{1} \and John Bostanci\inst{2} \and Aditya Gulati\inst{1} \and Yao-Ting Lin\inst{1}}
    \institute{University of California, Santa Barbara \and Columbia University}
\else
    \author{Prabhanjan Ananth\thanks{\texttt{prabhanjan@cs.ucsb.edu}}\\ \small{UCSB} \and John Bostanci\thanks{\texttt{johnb@cs.columbia.edu}}\\ \small{Columbia} \and Aditya Gulati\thanks{\texttt{adityagulati@ucsb.edu}}\\ \small{UCSB} \and Yao-Ting Lin\thanks{\texttt{yao-ting\_lin@ucsb.edu}}\\ \small{UCSB}}
    \date{}
\fi

\begin{document}

\maketitle

\begin{abstract}
\noindent We study the (in)feasibility of quantum pseudorandom notions in a quantum analog of the random oracle model, where all the parties, including the adversary, have oracle access to the same Haar random unitary. In this model, we show the following: 
\begin{itemize}
    \item (Unbounded-query secure) pseudorandom unitaries (PRU) exist. Moreover, the PRU construction makes two calls to the Haar oracle. 
    \item We consider constructions of PRUs making a single call to the Haar oracle. In this setting, we show that unbounded-query security is impossible to achieve. We complement this result by showing that bounded-query secure PRUs do exist with a single query to the Haar oracle. 
    \item We show that multi-copy pseudorandom state generators and function-like state generators (with classical query access), making a single call to the Haar oracle, exist.  
\end{itemize}
\noindent Our results have two consequences: (a) when the Haar random unitary is instantiated suitably, our results present viable approaches for building quantum pseudorandom objects without relying upon one-way functions and, (b) for the first time, we show that the key length in pseudorandom unitaries can be generically shrunk (relative to the output length). Our results are also some of the first usecases of the new ``path recording'' formalism for Haar random unitaries, introduced in the recent breakthrough work of Ma and Huang.  
\end{abstract}

\ifllncs
\else
\newpage 
    \tableofcontents
\newpage 
\fi

\section{Introduction}
Pseudorandomness is a powerful concept that is integral to not only cryptography but to the broader area of theoretical computer science. In  the recent years, there have been an exciting line of works on designing pseudorandom primitives in the quantum world. Quantum pseudorandom primitives already have had a major impact with applications in areas including quantum gravity theory~\cite{BFV19,ABFGVZZ24}, quantum machine learning~\cite{huang2022quantum}, quantum complexity theory~\cite{Kretschmer21,chia2024quantum} and more importantly, in quantum cryptography~\cite{AQY21,MY21}. From a cryptographic standpoint, thanks to Kretschmer's result~\cite{Kretschmer21} (see also~\cite{KQST23}), there is some evidence to believe that many of the recently introduced quantum pseudorandom primitives could be a weaker assumption than one-way functions. This has led to a plethora of new results that suggest that commitments~\cite{AQY21,MY21,AGQY22,Yan22,BCQ23,HMY23,BEMPQY23,ananth2023pseudorandom,KT24,BJ24}, encryption schemes~\cite{AQY21,HMY23}, digital signatures~\cite{MY21,BBOSS24} could be based on assumptions plausibly weaker than one-way functions. 
\par One major criticism on this line of work is the fact that all the quantum pseudorandom primitives proposed so far~\cite{JLS18,BS19,BrakerskiS20,AQY21,AGQY22,BBSS23,LQSYZ23,ABFGVZZ24,AGKL,MPSY24,brakerski2024real} rely upon the existence of one-way functions. In order for us to gain more confidence that the quantum pseudorandom primitives are weaker than one-way functions, it is imperative we need to look for candidate constructions that do not rely upon the existence of one-way functions. 

\paragraph{On Pseudorandomness from Random Quantum Circuits.} Indeed,~\cite{AQY21} suggest using local random circuits, that are quantum circuits with local Haar random gates, to design pseudorandom state generators (PRSGs). A pseudorandom state generator~\cite{JLS18}, one of the first quantum pseudorandom primitives, is an efficient quantum circuit $G$ that takes as input a key $k \in \{0,1\}^{\secparam}$ and produces an $n$-qubit quantum state such that the output distribution of $G(k)^{\otimes t}$, where $t$ is a polynomial in $\secparam$, is computationally indistinguishable from $t$ copies of an $n$-qubit Haar random state.
\par It is natural to consider random quantum circuits to build pseudorandom state generators. In fact,~\cite{AQY21} was not the only one to suggest using random quantum circuits to instantiate PRSGs. Couple of other works~\cite{BCQ23,KT24b} also suggested using random circuits in the design of quantum cryptographic primitives. Random quantum circuits are extensively studied, notably in quantum supremacy experiments~\cite{arute2019quantum} and in the unitary design constructions~\cite{BHH16}. In some ways, random quantum circuits share similar properties with Haar random unitaries. The seminal work of~\cite{BHH16} (see also~\cite{haferkamp2022random}) show that random circuits with polynomial (in $t$) depth are unitary $t$-designs; in other words, random circuits with sufficient depth agree with Haar random unitaries upto the $t^{th}$ moment. Recently,~\cite{schuster2024random} showed that local random quantum circuits, where the local gates act on sufficiently many qubits, are close to Haar random unitaries\footnote{Concretely, in order to have the diamond norm between the random circuits and Haar random unitaries to be negligible in $\secparam$, the local gates need to act upon $\omega(\log(\secparam))$ qubits.}, and~\cite{bostanci2024efficient} showed that states and unitaries generated by local random diagonal circuits, where local diagonal gates act alternate with Hadamard gates, also have similar properties to Haar random states and unitaries. 
\par The PRSG candidate posited by~\cite{AQY21} roughly states that the output of polynomial-sized random quantum circuits is pseudorandom. Interestingly, this candidate does not seem to rely upon one-way functions at all. Unfortunately, they do not provide any evidence on the security of their candidate. Concretely identifying cryptographic assumptions underlying this candidate is quite challenging. 
\par This is reminiscent of many cryptographic constructions that use real-world hash functions, which work well in practice although formally proving security of these constructions has remained elusive. To gain more confidence in such constructions, we have often resorted to proving security in idealized models, such as the random oracle model~\cite{BR93}. Although proving the security of constructions in the random oracle model does not outright
guarantee the security of their implementations in the real world (indeed, there are counterexamples~\cite{CG24}), so far, random oracles have been proven to be a useful heuristic and widely adopted in theoretical and practical cryptography~\cite{Green20}. Studying similar idealized models in the quantum world could prove to be impactful in quantum cryptography. 

\paragraph{Our Work: Pseudorandomness in the Quantum Haar Random Oracle Model.} We consider the quantum Haar random oracle model (QHROM) introduced by Chen and Movassagh~\cite{CM24} as a quantum analog of the random oracle model. In this model, all the parties have access to a Haar random unitary $U$ and its inverse. This model is especially useful to consider for cryptographic applications that use random circuits as a building block. Since it is challenging to base certain properties of random quantum circuits on cryptographic assumptions, we can instead consider an idealized model, where the adversary has access to a Haar random oracle. Arguing security in this model would then provide insight into the security of the construction where the Haar random oracle is instantiated with random quantum circuits. Beyond random quantum circuits, other phenomena could also be modeled in this framework. As an example, Bouland, Fefferman and Vazirani~\cite{BFV19} presented a candidate construction of pseudorandom state generators in the conformal field theory. Towards proving security of this candidate, they modeled the time evolution operators as a Haar random unitary and analyzed its behavior.~\cite{CM24} also proposed a construction of succinct quantum commitments in the QHROM. Unfortunately, they were unable to prove its security and they attribute the difficulty of proving security to the lack of tools for analyzing QHROM. 
\par Indeed, proving security in the QHROM is much more challenging than its classical counterpart which could also partially explain the reason why it took so long to establish the feasibility of pseudorandom unitaries, which are efficiently computable unitaries that are computationally indistinguishable from Haar random unitaries.  
\par Towards making progress in this direction, we consider a relaxation of the QHROM, that we refer to as the {\em inverseless} QHROM (${\sf iQHROM}$).  In this relaxation, all the parties have access to a Haar random unitary $U$ but not its inverse. The focus of our work is to study quantum pseudorandom primitives in the ${\sf iQHROM}$. As we will see later, proving security in the ${\sf iQHROM}$ is already quite challenging and it involves developing new techniques. Another reason to study the ${\sf iQHROM}$ is that showing feasibility results in the ${\sf iQHROM}$ would serve as a stepping stone towards investigating the feasibility in the QHROM (with inverses).

\subsection{Our Contributions}
We initiate a research direction on understanding the feasibility of quantum pseudorandomness using Haar random oracles. We focus on two pseudorandom primitives: namely pseudorandom state generators and pseudorandom unitaries, both first defined in~\cite{JLS18}. We briefly discuss their definitions in the {\sf iQHROM} before stating our results: 
\begin{itemize}
    \item Pseudorandom state generators (PRSG) in the {\sf iQHROM}: a PRSG is an efficient quantum circuit $G$, with oracle access to a Haar random unitary $U$, that takes as input a key $k \in \{0,1\}^{\secparam}$\footnote{$\secp$ is the security parameter.} and produces an $n$-qubit quantum state. In terms of security, we require that any query-bounded adversary\footnote{Typically, PRSGs guarantee security only against quantum polynomial-time adversaries. In this work, similar to the setting of state designs, we allow the adversary to be computationally unbounded. However, we restrict the number of queries made by $\Adversary$ to the Haar random oracle to be any polynomial in $\secp$.} $\Adversary$, with oracle access to $U$, should not be able to distinguish $G(k)^{\otimes t}$ from $\ket{\psi}^{\otimes t}$, where $\ket{\psi}$ is an $n$-qubit Haar random state. If $t$ is an arbitrary polynomial then we call this {\em multi-copy} PRSGs and if $t$ is fixed ahead of time (similar to state designs), we call this {\em bounded-copy} PRSGs.

    \item Pseudorandom function-like state generators (PRFS) in the $\mathsf{iQHROM}$: a PRFS is a keyed polynomial-sized circuit $G^{U}(k, \cdot)$ that produces $2^{m(\secp)}$ many $n$-qubit states $\ket{\psi_x} = G^{U}(k, w)$.  For PRFS security, we require that any $t$ query-bounded adversary $\Adversary^U$ that can adaptively request copies of $\ket{\psi_{w}}$ can not distinguish between the outputs of the PRFS generator and a family of $2^{m(\secp)}$ many i.i.d states sampled from the Haar measure. If $t$ is an arbitrary polynomial then we call this {\em multi-copy} PRFSs and if $t$ is fixed ahead of time (similar to state designs), we call this {\em bounded-copy} PRFSs.
    
    \item Pseudorandom unitaries (PRU) in the {\sf iQHROM}: a PRU is a keyed polynomial-sized quantum circuit $G_k^U$ that is functionally equivalent to an $m$-qubit unitary. In terms of security, we require that any query-bounded $\Adversary$, with oracle access to ${\cal O}$ and $U$, should not be able to distinguish whether ${\cal O}=G_k^U$ or whether ${\cal O}=V$, where $V$ is a freshly sampled $m$-qubit Haar random unitary. If the number of adversarial calls to $G_k^U$ is an arbitrary polynomial then we call $G_k^U$ an {\em unbounded-query} secure PRU and if the number of calls is fixed ahead of time, we call it a {\em bounded-query} secure PRUs.

\end{itemize}
It should be emphasized at this point that in both the definitions, the adversary does have oracle access to $U$.\footnote{We assume that the number of queries to $U$ is an arbitrary polynomial.} This is what makes the design of both these primitives challenging. This is akin to the random oracle model, where the adversary also has access to the random function. 

\paragraph{Unbounded-query secure PRUs.} Our main result is showing that pseudorandom unitaries exist in the ${\sf iQHROM}$. 

\begin{theorem}[Informal]
\label{thm:intro:prus:iqhrom}
PRUs exist in the inverseless quantum Haar random oracle model.     
\end{theorem}
We remark that our construction is quite simple: if $U$ is the $n$-qubit Haar unitary and if $k \in \{0,1\}^{\secp}$, where $\secp \leq n$, is the PRU key then $G_k^U=U(X^{k} \otimes \id_{n-\secp})U$. It is important to note that $G_k$ makes two sequential calls to $U$.  

\paragraph{Implications to Pseudorandom Unitaries in the plain model.} In the plain model, we can substitute the Haar random oracle with sampling a single PRU key.  Instantiating the PRU in the iQHROM with this single sample of a pseudorandom unitary yields a new unitary that looks pseudorandom relative to the first sample, at the cost of only sampling $\omega(\log(\secp))$ more bits.  Plugging these independent looking unitaries into the construction of \cite{schuster2024random} yields a PRU with a much larger output size, at the cost of only a small amount of additional randomness. And as a corollary of our result, we show how to stretch the output length of \emph{any} PRU that exists in the plain model. 
\begin{theorem}[Informal]
\label{thm:intro:stretching_pru}
    If any pseudorandom unitary family exists, there is a construction of pseudorandom unitaries with keys of size $O(\secp)$ and output size $O(\secp^c)$ for all constants $c$.
\end{theorem}

\noindent The work of~\cite{GJMZ23} proved a similar result (although technically incomparable) where the resulting pseudorandom unitary was only secure against a single adversarial query, but the stretching algorithm does not require sampling any additional randomness.  
\par We highlight that our stretched PRU only stretches \emph{output length}. In particular, it does not shrink the key length; though the increase in output length is much larger than the increase in key length.  Hence, for some output size $n$, one could start from PRUs for a much smaller security parameter (say, $n^\delta$), and build a pseudorandom unitary with output length $n$. The question of taking a pseudorandom unitary with a fixed output length to another pseudorandom unitary with the \emph{same} output length, but smaller key, is still an open question.

\paragraph{Unbounded-query secure PRUs: On the number of calls needed.} It is interesting to explore if it is inherent that $G_k$ needs to make at least two calls to $U$. We show that it is necessary. Informally, we show that any PRU construction making a single call to $U$ is insecure as long as the adversary is allowed to make $\Omega\left( \frac{\secp}{\log(\secp)} \right)$ queries to the PRU.

\begin{theorem}[Informal]
\label{thm:intro:negresult:one}
Any PRU construction that only makes a single query to the Haar random oracle is insecure against adversaries making $\Omega\left( \frac{\secp}{\log(\secp)} \right)$ non-adaptive queries to the PRU.
\end{theorem}

\paragraph{Bounded-query secure PRUs.} The above negative result leads us to an intriguing question: does there exist a PRU construction that only makes a single call to the Haar random oracle and satisfies security as long as the adversary only makes an {\em a priori} bounded number of queries to the PRU? We answer this question below. 

\begin{theorem}[Informal]
\label{thm:intro:bqprus}
PRUs exist in the inverseless quantum Haar random oracle model with the following properties: (a) the construction makes a single call to the Haar oracle, (b) the adversary makes at most $O\left(\frac{\secparam}{(\log(\secparam))^{1+\varepsilon}}\right)$ queries, for $\varepsilon > 0$. 
\end{theorem}

\noindent From our negative result (\Cref{thm:intro:negresult:one}), we have that condition (b) in the above theorem is tight. 

\paragraph{Negative Result: Generalization to the Parallel Query case.} We already saw that at least two calls to the Haar oracle are necessary if we were to design PRUs with unbounded query security. However, our construction in~\Cref{thm:intro:prus:iqhrom} makes two sequential calls to the Haar random oracle. A natural question is: {\em are two sequential calls necessary?} We answer this question in the affirmative. In fact, we show that any number of parallel calls to the Haar random oracle is not sufficient. Concretely, we show that any PRU making arbitrary number of {\em parallel} calls to the Haar random oracle $U$ (i.e. of the form, $W \cdot U^{\otimes t} \cdot V$, for some unitaries $W,V$ and for some polynomial $t$) is insecure as long as the adversary is allowed to make $\Omega\left( \secp \right)$ queries to the PRU and $U$. 

\begin{theorem}[Informal]
\label{thm:intro:neglresult:two}
Any PRU construction that only makes parallel queries to the Haar oracle is insecure against adversaries making $\Omega(\secp)$ non-adaptive queries to the PRU.  
\end{theorem}

\paragraph{Pseudorandom Quantum States and Function-like States in the ${\sf iQHROM}$.} We next look at pseudorandom state generators in the ${\sf iQHROM}$. Since PRUs imply PRSGs, we immediately get the implication that PRSGs exist in the ${\sf iQHROM}$. The question, then, is if there are even simpler constructions of PRSGs in the ${\sf iQHROM}$. We show the following. 

\begin{theorem}[Informal]
\label{thm:intro:prsgs}
PRSGs exist in the inverseless quantum Haar random oracle model with the following properties: (a) the construction makes a single call to the Haar oracle and, (b) the adversary is given an arbitrary polynomial number of copies. 
\end{theorem}

\noindent In particular, our result shows that the negative result~\Cref{thm:intro:negresult:one} only holds for PRUs and not PRSGs and PRFSs. We note that~\cite{BFV19} proposed a much more involved construction of PRSGs in the stronger ${\sf QHROM}$ model although they did not formally prove the security of their candidate. We also extend our construction of pseudorandom states to a construction of pseudorandom \emph{function-like} states.
\begin{theorem}[Informal]
    \label{thm:intro:prfs}
    PRFSs exist in the inverseless quantum Haar random oracle model, with the following properties: (a) the construction makes a single call to the Haar random oracle and, (b) the adversary is given an arbitrary polynomial number of classical queries. 
\end{theorem}

\vspace{-10 pt}
\begin{figure}[H]
\begin{tabular}{|p{\textwidth}|}
\hline 
\vspace{-10 pt}
\begin{center}
\underline{\textsc{Results Summary (Informal)
}}\\
\begin{enumerate}
    \item Unbounded query secure PRUs, with \underline{two calls} to the Haar random oracle, exist in ${\sf iQHROM}$ (\Cref{thm:intro:prus:iqhrom})
    \item Unbounded query secure PRUs exist with keys of size $O(\lambda^{1/c})$ for any constant $c$, if any PRU exists in the plain model (\Cref{thm:intro:stretching_pru}).
    \item Unbounded query secure PRUs, with \underline{one call} to the Haar random oracle, does {\bf not} exist in ${\sf iQHROM}$ (\Cref{thm:intro:negresult:one})
    \item Bounded query secure PRUs, with \underline{one call} to the Haar random oracle, exists in ${\sf iQHROM}$ (\Cref{thm:intro:bqprus})
    \item Multi-copy PRSGs, with \underline{one call} to the Haar random oracle, exists in ${\sf iQHROM}$ (\Cref{thm:intro:prsgs})
    \item Adaptively secure PRFSs, with \underline{one call} to the Haar random oracle, exists in the $\mathsf{iQHROM}$ (\Cref{thm:intro:prfs})
    \item The negative result in bullet 3 can be generalized further: unbounded query secure PRUs, with any number of parallel calls to the Haar random oracle, does {\bf not} exist in ${\sf iQHROM}$ (\Cref{thm:intro:neglresult:two})
\end{enumerate}
\end{center}
\vspace{-10pt}
\ \\
\hline 
\end{tabular} 
\label{fig:resultssummary}
\end{figure}

\ifllncs

\else
    \begin{figure}[H]
    \input{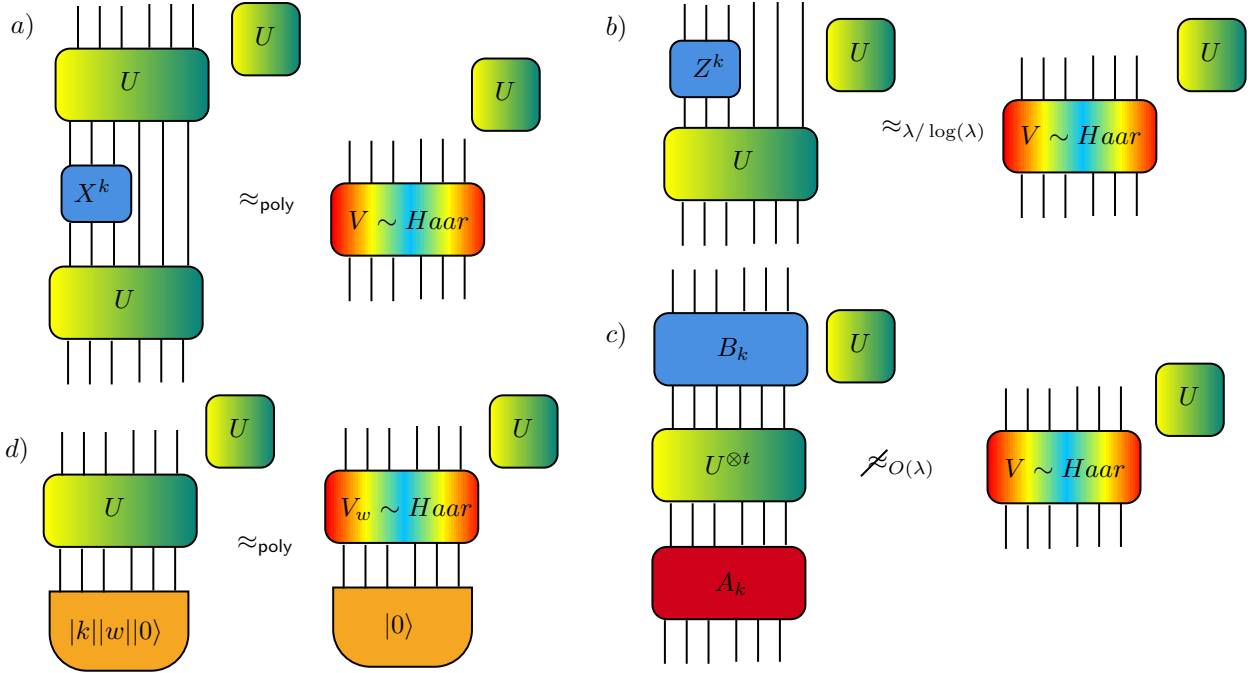}
    \caption{A summary of our results, time goes up in all diagrams. \textbf{(a)} We show that the simple $U X U$ is indistinguishable from an independently sampled Haar random unitary for adversaries who have query access to $U$.  \textbf{(b)} We also show that up to $\secp / \log(\secp)$ queries, the even simpler unitary $Z U$ is indistinguishable from a Haar random unitary to adversaries that have query access to $U$.  \textbf{(c)} We also show that there is no construction of $O(\secp)$-secure pseudo-random unitaries that only make a single parallel query to the common Haar random unitary, if the adversary is given polynomial-space computation.  \textbf{(d)} Finally, we show that simply calling the Haar random unitary on a uniformly random classical basis state is indistinguishable from a Haar random state to adversaries that get polynomially many queries to $U$, yielding both PRSGs and PRFSs}
\end{figure}
    \section{Related Work}

\paragraph{Quantum Pseudorandomness.}
Study into quantum pseudorandomness began with the work of \cite{JLS18}, who first defined pseudorandom states and unitaries.\cite{JLS18} presented the first construction of pseudorandom states from one-way functions. Since then, a few works~\cite{BS19,BrakerskiS20,AGQY22} presented improved and simpler constructions of pseudorandom states. However, until very recently, the construction of pseudorandom unitaries have remained elusive. Recently, a few works~\cite{LQSYZ23,AGKL,brakerski2024real} made progress on building pseudorandom unitaries. Specifically, they consider the security of pseudorandom unitaries on specific sets of queries. 




Building on these results, \cite{MPSY24} provided the first constructions of non-adaptively secure pseudorandom unitaries, which used the so-called $\mathsf{PFC}$ ensemble (which stands for ``Permutation-Random Function-Clifford'').  \cite{chen2024efficient} simultaneously constructed pseudorandom unitaries from random permutations, but both papers rely heavily on Schur-Weyl duality and the properties of the symmetric group.  Recently, \cite{MH24} extended the compressed oracle techniques of \cite{zhandry2019record} to the path recording formalism for Haar random unitaries.  Specifically, they show that the $\mathsf{PFC}$ ensemble is a pseudorandom unitary and that $\mathsf{C}^{\dagger}\mathsf{PFC}$ is additionally inverse secure. \cite{schuster2024random} showed that low-depth pseudorandom unitaries can be instantiated from concrete assumptions such as LWE.
\ 

\paragraph{Common Haar State Model.} The common Haar random state model (CHRS)~\cite{AGL24,chen2024power} is an idealized model of computation where all parties have access to a joint common state, and can be viewed as the quantum equivalent of the classical common reference string model.  Both works show that in the CHRS model, some form of bounded copy pseudorandom states with short keys (and therefore quantum bit commitments) exist, while ruling out a wide range of other primitives.  \cite{AGL24} rule out quantum cryptography primitives with classical communication, and \cite{chen2024power} rule out unbounded copy pseudorandom states.  While this idealized model is interesting to study and provides both efficient constructions of cryptographic primitives and black box separations, the model can be problematic to instantiate in a realistic setting.  For example, instantiating the model in the real world might require a complicated multi-party computation to compute a shared quantum state, or a trusted third party who can distribute copies of the state.  

\paragraph{Quantum Haar Random Oracle Model.} The quantum Haar random oracle model was first introduced in \cite{CM24}, who provided a construction of succinct commitments.  However, \cite{CM24} was not able to analyze the security of their mode in the QHROM.  \cite{BFV19} separately consider the QHROM, using it as an idealized model of the scrambling behavior of black holes.  They provide a construction of pseudorandom states, with a proof sketch. The common denominator in both the works is the lack of techniques to formally analyze security in the quantum Haar random oracle model. Finally, \cite{Kretschmer21} considers an idealized version of pseudorandom states, consisting of $2^{\secp}$ many Haar random unitaries.  They show that relative to this oracle and a $\mathsf{PSPACE}$ oracle, one-way functions do not exist while pseudorandom unitaries do.  

\fi
\section{Technical Overview}

Here we provide proof sketches for the our construction of unbounded-copy secure pseudorandom unitaries in the $\mathsf{iQHROM}$ (\Cref{sec:tech:pru}), our negative result on an unbounded pseudorandom unitary from a single parallel query (\Cref{sec:tech:negative}), our construction of bounded-copy secure pseudorandom unitaries from a single query in the $\mathsf{iQHROM}$ (\Cref{sec:tech:bounded_pru}),
 and our simple construction of pseudorandom states in the $\mathsf{iQHROM}$ (\Cref{sec:tech:prs}).

\subsection{Path Recording Formalism of Ma and Huang}

We begin by recalling the (forward secure) characterization of Haar random unitaries, known colloquially as the path recording framework.  If $\reg{A}$, $\reg{X}$ and $\reg{Y}$ be three quantum registers, with $\reg{A}$ being the adversary's register and $\reg{XY}$ being purifying registers, then the path recording oracle $\mathsf{PR}: \reg{AXY} \mapsto \reg{AXY}$ is the following linear map:
    \begin{equation*}
        \mathsf{PR}: \ket{x}_{\reg{A}} \otimes \ket{R}_{\reg{XY}} \mapsto \frac{1}{\sqrt{N - |R|}}\sum_{\substack{y \in [N] \setminus \Im(R)}} \ket{y}_{\reg{A}} \otimes  \ket{R \cup \{(x, y)\}}\,,
    \end{equation*}
where $R$ is an injective relation state, which is a set of input/output pairs with the condition that the same output never appears twice in the set, and $\Im(R)$ is the set of outputs in the relation state. In the rest of this technical overview, we drop the normalizing factor of $\frac{1}{\sqrt{N - |R|}}$.  \cite{MH24} show that the path recording oracle is right invariant and thus indistinguishable from oracle access to a Haar random unitary.  In the paper of \cite{MH24}, the authors go on to argue that the action of a uniformly random sampled permutation and binary phase function (i.e. $\mathsf{PF}$) is identical to the path recording oracle \emph{if} the inputs to the oracle are distinct.  Using the properties of a $2$-design, they show that any adversary (except with negligible probability) will query the $\mathsf{PF}$ part of the $\mathsf{PFC}$ oracle on distinct strings, allowing them to claim the following are all (approximately) indistinguishable from each other
\begin{equation*}
    \mathsf{PFC} \approx \mathsf{PR}\cdot \mathsf{C} \approx \mathsf{PR} \approx \mathsf{PR} \cdot \mathcal{U} \approx \mathsf{PF}\cdot \mathcal{U} = \mathcal{U}\,.  
\end{equation*}

As the path recording framework is so important to the result of this paper, we review some of the proof.  Let $\mathcal{A}$ be a $t$-query adversary, and let $\ket{\mathcal{A}^{\mathsf{PR}}}$ be the state of the adversary with access to $\mathsf{PR}$.  We can write this state as follows:
\begin{equation*}
    \ket{\mathcal{A}^{\mathsf{PR}}} = \prod_{i = 1}^{t} \left(\mathsf{PR}_{\reg{AXY}} \cdot  A_{i}\right) \ket{0}_{\reg{A}} \otimes \ket{\{\}}_{\reg{XY}}\,.
\end{equation*}
Expanding the definition of the path recording oracle, we get the following state:
\begin{equation*}
    \ket{\mathcal{A}^{\mathsf{PR}}} = \sum_{\substack{\Vec{x} \in [N]^t\\ \Vec{y} \in [N]^t_{\mathrm{dist}}}} \prod_{i = 1}^{t} \left(\ket{y_i}\!\bra{x_i} \cdot A_i\right)\ket{0}_{\reg{A}} \otimes \ket{\{(x_i, y_i)\}_{i \in [t]}}_{\reg{XY}}\,.
\end{equation*}
Now consider another oracle, where instead of only being given access to the path recording oracle, the adversary is given access to an oracle that first applies a unitary $U = \sum_{x, y \in [N]} \alpha_{xy} \ket{y}\!\bra{x}$, and then applies the path recording oracle.
\begin{align*}
    \ket{\mathcal{A}^{\mathsf{PR} \cdot U}} &= \sum_{\substack{\Vec{x} \in [N]^t\\ \Vec{y} \in [N]^t_{\mathrm{dist}}}} \prod_{i = 1}^{t} \left(\ket{y_i}\!\bra{x_i} \cdot U \cdot A_i\right)\ket{0}_{\reg{A}} \otimes \ket{\{(x_i, y_i)\}_{i \in [t]}}_{\reg{XY}}\\
    &= \sum_{\substack{\Vec{x} \in [N]^t\\ \Vec{y} \in [N]^t_{\mathrm{dist}}}} \prod_{i = 1}^{t} \left(\sum_{z_i} \alpha_{z_i, x_i}\ket{y_i}\!\bra{z_i} \cdot A_i\right)\ket{0}_{\reg{A}} \otimes \ket{\{(x_i, y_i)\}_{i \in [t]}}_{\reg{XY}}\\
    &= \sum_{\substack{\Vec{x}, \Vec{z} \in [N]^t\\ \Vec{y} \in [N]^t_{\mathrm{dist}}}} \prod_{i = 1}^{t} \left(\ket{y_i}\!\bra{z_i} \cdot A_i\right)\ket{0}_{\reg{A}} \otimes \alpha_{z_i, x_i}\ket{\{(x_i, y_i)\}_{i \in [t]}}_{\reg{XY}}\\
    &= \sum_{\substack{\Vec{x}, \Vec{z} \in [N]^t\\ \Vec{y} \in [N]^t_{\mathrm{dist}}}} \prod_{i = 1}^{t} \left(\ket{y_i}\!\bra{x_i} \cdot A_i\right)\ket{0}_{\reg{A}} \otimes \alpha_{x_i, z_i}\ket{\{(z_i, y_i)\}_{i \in [t]}}_{\reg{XY}}\\
    &= \sum_{\substack{\Vec{x} \in [N]^t\\ \Vec{y} \in [N]^t_{\mathrm{dist}}}} \prod_{i = 1}^{t} \left(\ket{y_i}\!\bra{x_i} \cdot A_i\right)\ket{0}_{\reg{A}} \otimes U_{\reg{X}}^{\otimes t}\ket{\{(x_i, y_i)\}_{i \in [t]}}_{\reg{XY}}\,.
\end{align*}
Here the second line is the result of expanding $U$, the third line is the result of aggregating all of the sums and moving the coefficient $\alpha_{z_i, x_i}$ to the purifying register, and the next lines the result of re-labelling $x_i \leftrightarrow z_i$ and subsequently applying the definition of $U$ again.  Thus, the state in the $\reg{A}$ register (after tracing out $\reg{XY}$) remains un-changed after applying $U$.  Taking $U$ to be sampled from the Haar measure, we see that the state of an adversary with access to the path recording oracle is identical to their state if they first applied a Haar random unitary and then had the path recording oracle act, which is itself indistinguishable from a Haar random unitary (without the path recording). 

\subsection{(Unbounded) PRUs with Two Queries}
\label{sec:tech:pru}

The construction of unbounded-query pseudorandom unitaries in the iQHROM (with $U$ being the Haar random oracle) is the following
\begin{equation*}
    \mathsf{PRU}_{k}^{U} = U (X^{k} \otimes \id) U\,.
\end{equation*}
To show that this is a pseudorandom unitary, our goal is to show that $(U, \mathsf{PRU}_k^{U})$ is computationally indistinguishable from $(U, V)$, where $V$ is sampled independently from the Haar measure, to an adversary that makes a polynomial number of queries.  At a high level, we want to show that (in the purifying register), $X^{k}$ allows the path recording oracle to dis-entangle calls to $U$ and $\mathsf{PRU}_k^{U}$, which will allow the oracle (except with negligible probability) to record the calls to the pseudorandom unitary in a separate register than calls to $U$, resembling an entirely distinct path recording oracle. 

In order to prove this, we can extend the path recording oracle to the case when an adversary $\mathcal{A}$ gets access to two Haar random unitaries $U$ and $V$ and makes $t$ queries.  In this case, we can replace both Haar random unitaries with \emph{distinct} path recording oracles $\mathsf{PR}_1$ and $\mathsf{PR}_2$, each with their own purifying register, $\reg{X}_1\reg{Y}_1$, and $\reg{X}_2\reg{Y}_2$.  If we say $\mathcal{A}^{\mathsf{PR}_1, \mathsf{PR}_2}$ makes alternating queries to each of their oracles (with unitaries $A_i$ and $B_i$ before the respective oracle calls), we can express their state as follows:
\begin{equation*}
    \ket{\mathcal{A}^{\mathsf{PR}_1, \mathsf{PR}_2}} = \sum_{\substack{\Vec{x}, \Vec{a} \in [N]^t\\ \Vec{y}, \Vec{b} \in [N]^t_{\mathrm{dist}}}} \prod_{i = 1}^{t} \left(\ket{b_i}\!\bra{a_i} \cdot B_{i}\ket{y_i}\!\bra{x_i} \cdot A_{i} \right)\ket{0}_{\reg{A}} \otimes \ket{\{(x_i, y_i)\}_{i \in [t]}}_{\reg{X}_1\reg{Y}_1}\ket{\{(a_i, b_i)\}_{i \in [t]}}_{\reg{X}_2\reg{Y}_2}\,.
\end{equation*}
Similarly, we can write the state of the $t$-query adversary after querying $U$ and $\mathsf{PRU}_{k}^{U}$ for a uniformly random $k$ using the path recording framework as follows:
\begin{align*}
    &\ket{\mathcal{A}^{\mathsf{PR}, \mathsf{PRU}(\mathsf{PR})}}\\
    &\hspace{5mm}=\sum_{k \in \{0, 1\}^{\lambda}}\sum_{\substack{\Vec{x}, \Vec{a}, \Vec{c} \in [N]^t\\ (\Vec{y}, \Vec{b}, \Vec{d}) \in [N]^{3t}_{\mathrm{dist}}}} \prod_{i = 1}^{t} \left(\ket{d_i} \!\braket{c_i}{b_i \oplus k}\!\bra{a_i} \cdot B_{i} \ket{y_i}\!\bra{x_i} \cdot A_{i} \right)\ket{0}_{\reg{A}} \otimes \ket{\{(x_i, y_i), (a_i, b_i), (c_i, d_i)\}_{i \in [t]}}_{\reg{X}\reg{Y}} \ket{k}_{\reg{K}}\\
    &\hspace{5mm}=
    \sum_{k \in \{0, 1\}^{\lambda}}\sum_{\substack{\Vec{x}, \Vec{a}, \in [N]^t\\ (\Vec{y}, \Vec{b}, \Vec{d}) \in [N]^{3t}_{\mathrm{dist}}}} \prod_{i = 1}^{t} \left(\ket{d_i}\!\bra{a_i} \cdot B_{i} \ket{y_i}\!\bra{x_i} \cdot A_{i} \right)\ket{0}_{\reg{A}} \otimes \ket{\{(x_i, y_i), (a_i, b_i), (b_i \oplus k, d_i)\}_{i \in [t]}}_{\reg{X}\reg{Y}} \ket{k}_{\reg{K}}\,.    
\end{align*}
Here the first line uses the fact that $X^{k} \ket{y} = \ket{y \oplus k}$, together with the expansion of the path recording oracle $\mathsf{PR}$.  The second line removes terms in the sum for which $\braket{c_i}{b_i \oplus k} = 0$.  Note that we also purify the key register $\reg{K}$ so that we can write the adversary's purified view as a pure state.  

To show the closeness between $\ket{\mathcal{A}^{\mathsf{PR}_1, \mathsf{PR}_2}}$ and $\ket{\mathcal{A}^{\mathsf{PR}, \mathsf{PRU}(\mathsf{PR})}}$ after partially tracing out their purified registers, our approach is to define two isometries $\mathsf{Split}:\reg{XYK} \to \reg{X_1Y_1X_2Z_2Y_2K}$ and $\mathsf{Augment}:\reg{X_1Y_1X_2Y_2} \to \reg{X_1Y_1X_2Z_2Y_2K}$, where $\reg{Z}_2$ is an ancilla register, such that 
\begin{equation*}
    \mathsf{Split}\ket{\mathcal{A}^{\mathsf{PR}, \mathsf{PRU}(\mathsf{PR})}} 
    \approx \mathsf{Augment} \ket{\mathcal{A}^{\mathsf{PR}_1, \mathsf{PR}_2}}.
\end{equation*}
For intuition, we have the following \emph{classical} interpretation of path recording oracles. 

\paragraph{Splitting $\ket{\mathcal{A}^{\mathsf{PR}, \mathsf{PRU}(\mathsf{PR})}}$.} Suppose a classical algorithm $\Adversary$ is given oracle access to \emph{randomized} oracles $f:\bit^n \to \bit^n$ and $G^f_k:\bit^n \to \bit^n$. Without loss of generality, we assume that $\Adversary$ alternatively asks $t = \poly(\secp)$ queries to each of $f$ and $G^f_k$.\footnote{That is the adversary makes all odd-indexed queries to $f$ and all even-indexed queries to $G^f_k$.} The distribution of $(f,G^f_k)$ is defined via the following interactive experiment involving a challenger $\cC$ and an adversary $\Adversary$. First, $\cC$ initializes an empty relation $R = \emptyset$ and samples $k \random \bit^n$. For odd $i$'s, upon receiving query $x_i$ to $f$, $\cC$ samples $y_i \random \bit^n \setminus \Im(R)$, adds $(x_i,y_i)$ to $R$, and returns $y_i$ to $\Adversary$. For even $i$'s, upon receiving query $x_i$ to $G^f_k$, $\cC$ samples $z_{\nicefrac{i}{2}} \random \bit^n \setminus \Im(R)$, adds $(x_i,z_{\nicefrac{i}{2}})$ to $R$, samples $y_i \random \bit^n \setminus \Im(R)$, adds $(z_{\nicefrac{i}{2}} \oplus k,y_i)$ to $R$, and returns $y_i$ to $\Adversary$. Note that even if some $x_i$'s are equal to each other, the corresponding $y_i$'s are pairwise distinct. At the end, $\Adversary$ has learned the query-answer pairs $\set{(x_i,y_i)}_{i \in [2t]}$. On the other hand, $R$ is of size $3t$ and becomes 
\begin{equation*}
    \set{(x_1,y_1),(x_2,z_1),(z_1 \oplus k, y_2), \dots, 
    (x_{2t-1},y_{2t-1}),(x_{2t},z_t),(z_t \oplus k, y_{2t})}.
\end{equation*}
We crucially rely on the notion of \emph{correlated pairs} defined as follows. Given $(R,k)$, a pair $((u,v),(u',v')) \in R \times R$ is \emph{$k$-correlated} if $v \oplus u' = k$. Note that, as long as there are only $t$ many $k$-correlated pairs, we can always split odd and even queries. We say that $\Adversary$ wins if there are more than $t$ many $k$-correlated pairs in $R$.

We first show that $\Adversary$'s winning probability is negligible. The idea is to defer the sampling of $k$ until $\Adversary$ is done with querying. Consider the following identically distributed experiment. First, $\cC$ initializes an empty relation $R = \emptyset$. For odd $i$'s, upon receiving query $x_i$ to $f$, $\cC$ samples $y_i \random \bit^n \setminus \Im(R)$, adds $(x_i,y_i)$ to $R$, and returns $y_i$ to $\Adversary$. For even $i$'s, upon receiving query $x_i$ to $G^f_k$, $\cC$ samples $y_i \random \bit^n \setminus \Im(R)$, adds $(\bot,y_i)$ to $R$, and returns $y_i$ to $\Adversary$. At the end, for $i\in[t]$, $\cC$ samples $z_i \random \bit^n \setminus \Im(R)$ and adds $(x_{2i},z_i)$ to $R$. Finally, $\cC$ samples  $k \random \bit^n$ and updates $(\bot,y_{2i}) \mapsto (z_i \oplus k,y_{2i})$ for $i\in[t]$. 

By a careful case analysis, we show that for \emph{any} $(\Vec{x},\Vec{y},\Vec{z})$\footnote{Here $\Vec{x} := (x_1,x_2,\dots,x_{2t})$, $\Vec{y} := (y_1,y_2,\dots,y_{2t})$ and $\Vec{z} := (z_1,x_2,\dots,z_t)$.} in the support, $R_k := \set{(x_1,y_1),(x_2,z_1),(z_1 \oplus k, y_2), \dots, (x_{2t-1},y_{2t-1}),(x_{2t},z_t),(z_t \oplus k, y_{2t})}$ has more than $t$ many $k$-correlated pairs \emph{if and only if} $k$ is of the form $x_i \oplus y_j$ or $x_i \oplus z_j$ for some $i,j$. That is to say, right before $\cC$ samples $k$, there are \emph{always} at most $O(t^2) = \poly(\secp)$ many ``bad'' keys $k$ such that $(R_k,k)$ has more than $t$ many $k$-correlated pairs. Therefore, the winning probability of $\Adversary$ is at most $O(t^2/2^\secp) = \negl(\secp)$. Suppose $\Adversary$ does not win, observe that $R_k$ has exactly $t$ \emph{mutually disjoint} $k$-correlated pairs. With the information of $k$, one can uniquely map $R_k$ into $(R_k^{\mathsf{iso}}, R_k^{\mathsf{cor}})$ where $R_k^{\mathsf{iso}} := \set{(x_{2i-1},y_{2i-1})}_{i\in[t]}$ and $R_k^{\mathsf{cor}} := \set{(x_{2i},z_i,y_{2i})}_{i\in[t]}$.

\paragraph{Augmenting $\ket{\mathcal{A}^{\mathsf{PR}_1, \mathsf{PR}_2}}$.} Suppose a classical algorithm $\Adversary$ is given oracle access to randomized oracles $f_1,f_2:\bit^n \to \bit^n$. Similarly, we assume that $\Adversary$ alternatively asks $t$ queries to each oracle. The distribution of $(f_1,f_2)$ is defined via the following interactive experiment involving a challenger $\cC$ and an adversary $\Adversary$. First, $\cC$ initializes two empty relations $R_1 = R_2 = \emptyset$. For odd $i$'s, upon receiving query $x_i$ to $f_1$, $\cC$ samples $y_i \random \bit^n \setminus \Im(R_1 \cup R_2)$, adds $(x_i,y_i)$ to $R_1$, and returns $y_i$ to $\Adversary$. For even $i$'s, upon receiving query $x_i$ to $f_2$, $\cC$ samples $y_i \random \bit^n \setminus \Im(R_1 \cup R_2)$, adds $(x_i,y_i)$ to $R_2$, and returns $y_i$ to $\Adversary$.\footnote{Notice that by definition, it is always the case that $\Im(R_1) \cap \Im(R_2) = \emptyset$ whereas it is not necessarily true for $\ket{\mathcal{A}^{\mathsf{PR}_1, \mathsf{PR}_2}}$. However, in~\Cref{sec:PR_New} we prove that making such an approximation only introduces negligible error.} At the end, we have $R_1 = \set{(x_{2i-1},y_{2i-1})}_{i\in[t]}$ and $R_2 = \set{(x_{2i},y_{2i})}_{i\in[t]}$.

Now, our goal is to sample $(\set{z_i}_{i\in[t]}, k)$ conditioned on $(R_1,R_2)$ so that the joint distribution of $(\set{(x_{2i-1},y_{2i-1})}_{i\in[q]}$ and $\set{(x_{2i}, z_i, y_{2i})}_{i\in[t]}, k)$ is negligibly close to that of $(R_k^{\mathsf{iso}}, R_k^{\mathsf{cor}}, k)$ in the previous experiment. The following natural way turns out to work: for $i\in[t]$, sample $z_i \random \bit^n \setminus (\Im(R_1 \cup R_2) \cup \bigcup_{j<i} \set{z_j})$; sample a uniformly random $k$ conditioned on $k$ is not of the form $x_i \oplus y_j$ or $x_i \oplus z_j$ for some $i,j$. \\

It turns out that the above classical reasoning offers a method for deriving the quantum proof in~\Cref{sec:pru:shortkeys}, though it involves technical subtleties. To work through these technical subtleties, we introduce a framework of working with relation states in~\Cref{sec:multiset}. We employ the mentioned framework to construct $\mathsf{Split}$ and $\mathsf{Augment}$.

\subsection{(Unbounded) PRUs: One Parallel Query is Insufficient}
\label{sec:tech:negative}
We also show that sequential queries to the Haar random oracle are required to get unbounded query secure pseudorandom unitaries.  Formally, we show that for every PRU construction in the iQHROM that only makes a single parallel calls to $U$, there is a polynomial space adversary that breaks PRU security with $O(\lambda)$ many non-adaptive calls to the PRU and common Haar random unitary.  In order to prove this, we use the quantum OR attack from \cite{chen2024power}.  In particular, we show how, using the ricochet property of EPR pairs, an adversary can prepare the Choi state $\ket{\Phi_{\mathsf{PRU}_k}}$ from many copies of $\ket{\Phi_{U}}$.  

Then, an adversary given oracle access to $\mathcal{O}$ can prepare many copies of $\ket{\Phi_{\mathcal{O}}}$ and perform swap tests with $\ket{\Phi_{\mathsf{PRU}_k}}$ to determine if they have access to one of the pseudorandom unitaries or an independently sampled Haar random unitary.  Since there is a unitary that prepares $\ket{\Phi_{\mathsf{PRU}_k}}$ from many copies of $\ket{\Phi_{U}}$, the adversary can recover and re-use their copies of $\ket{\Phi_{U}}$.  The proof of correctness follows a similar line as \cite{chen2024power}. We also present a tighter analysis using techniques from~\cite{AGL24} when the PRU construction queries the Haar random oracle exactly once.

\subsection{(Bounded) PRUs with One Query: Feasibility} 
\label{sec:tech:bounded_pru}

In the case where the construction only makes a single query to the Haar random oracle, we present an even simpler construction of $o(\lambda / \log(\lambda))$-query secure pseudorandom unitaries.  Specifically, the pseudorandom unitary is the following
$$\mathsf{PRU}_{k}^{U} = (Z^{k}\otimes \id) U\,.$$

\noindent We prove that all adversaries making $o(\lambda / \log(\lambda))$ (potentially adaptive) calls to $\mathsf{PRU}$ and arbitrary polynomial calls to $U$ cannot distinguish between $\mathsf{PRU}$ and an independently sampled Haar random unitary. To show this, begin by writing out the construction using the path recording framework to represent $U$. We will write the state assuming that the adversary $t$ total queries of with $\ell$ queries are made to $\mathsf{PRU}_{k}^{U}$ (on indices $\bfa = \set{a_1,\ldots,a_\ell}$).
\begin{align*}
    \ket{\mathcal{A}^{\mathsf{PR}, \mathsf{PRU}(\mathsf{PR})}} 
    = \begin{split}
    \sum_{k \in \{0, 1\}^{\lambda}}\sum_{\substack{\Vec{x}\in [N]^t\\ \Vec{y}\in [N]^t_{\mathrm{dist}}}} (-1)^{\langle k || 0^{n - \lambda},\bigoplus_{i\in\bfa} y_{i}\rangle}\prod_{i = 1}^{t} &\left( \ketbra{y_i}{x_i} \cdot A_{i} \right)\ket{0}_{\reg{A}} \\& \otimes \ket{\{(x_i, y_i)\}_{i \in [t]}}_{\reg{X}\reg{Y}} \ket{k}_{\reg{K}}\,.
    \end{split}
\end{align*}
Here the phase comes from the fact that $Z^{k}\ket{y} = (-1)^{\langle k,y\rangle} \ket{y}$.  Then we can push both the phase, and sum over keys into the $\reg{K}$ register to the following state
$$\sum_{\substack{\Vec{x}\in [N]^t\\ \Vec{y}\in [N]^t_{\mathrm{dist}}}} \prod_{i = 1}^{t} \left( \ketbra{y_i}{x_i} \cdot A_{i} \right)\ket{0}_{\reg{A}}\ket{\{(x_i, y_i)\}_{i \in [t]}}_{\reg{X}\reg{Y}}\sum_{k \in \{0, 1\}^{\lambda}}(-1)^{\langle k || 0^{n - \lambda},\bigoplus_{i\in\bfa} y_{i}\rangle} \ket{k}_{\reg{K}}$$
Then if we apply an isometry that appends $n - \lambda$ many $0$'s to the key register and then performs an $n$-qubit Hadamard, we get the following.
$$\sum_{\substack{\Vec{x}\in [N]^t\\ \Vec{y}\in [N]^t_{\mathrm{dist}}}} \prod_{i = 1}^{t} \left( \ketbra{y_i}{x_i} \cdot A_{i} \right)\ket{0}_{\reg{A}}\ket{\{(x_i, y_i)\}_{i \in [t]}}_{\reg{X}\reg{Y}} \ket{\bigoplus_{i\in\bfa} y_{i}}_{\reg{K}}$$
Using the idea of $\ell$-fold collision-freeness from \cite{AGL24}, we are able to show that if $\Vec{y}$ is an $\ell$-fold collision-free set, the XOR of all the $\set{y_i}_{i\in\bfa}$ will suffice to determine all $\set{y_i}_{i\in\bfa}$ from $\Vec{y}$. We show that as long as $\ell = o(\secp / \log(\secp))$, the weight on $\ell$-fold collision-free $\Vec{y}$'s is overwhelming.
Thus, there is an isometry that for most $\Vec{y}$'s identifies $\set{y_i}_{i\in\bfa}$ from the XOR value in the key register, and extracts the elements of the set into a new relation state. Hence, outputting a state close to the following
$$\sum_{\substack{\Vec{x}\in [N]^t\\ \Vec{y}\in [N]^t_{\mathrm{dist}}}} \prod_{i = 1}^{t} \left( \ketbra{y_i}{x_i} \cdot A_{i} \right)\ket{0}_{\reg{A}}\ket{\{(x_i, y_i)\}_{i \in [t]\setminus\bfa}}\ket{\{(x_i, y_i)\}_{i \in\bfa}}_{\reg{X}\reg{Y}}.$$
This is exactly the state that the adversary would have after querying two independent Haar random unitaries, as desired.

\subsection{(Unbounded) Pseudorandom States with One Query}
\label{sec:tech:prs}
We also present an extremely simple construction of pseudorandom states in the iQHROM.  The pseudorandom state for key $k$ of size $\secp$ is simply
\begin{equation*}
    \ket{\psi_k} = U\ket{k|| 0^n}\,.
\end{equation*}
In order to show that this is a pseudorandom state even against adversaries who can query $U$, we write out the state of the adversary who recieves $t$ copies of the pseudorandom state as input. The state will be proportional to
\begin{equation*}
    \sum_{k = 0}^{2^{\lambda} - 1}\left(\mathsf{PR}\ket{k||0}\right)^{\otimes t}\ket{k||0^{n}} = \sum_{y_1, \ldots, y_t \in [N]^{t}_{\mathrm{dist}}} \ket{y_1, \ldots, y_t} \otimes \sum_{k = 0}^{2^{\lambda} - 1}\ket{\{(k||0, y_i)\}_{i = 1}^{t}} \ket{k||0}\,.
\end{equation*}
Then when the adversary makes $s$ additional calls to the Haar random unitary, they will have the following state
\begin{equation*}
    \sum_{\vec{x} \in [N]^s}\sum_{(\vec{y}, \vec{z}) \in [N]^{t+s}_{\mathrm{dist}}} \left(\prod_{i = 1}^{s} \ket{z_i}\!\bra{x_i} A_i \right) \ket{y_1, \ldots, y_t} \ket{0} \otimes \sum_{k = 0}^{2^{\lambda} - 1}\ket{\{(k||0, y_i)\}_{i = 1}^{t} \cup \{(x_j, z_j)\}_{j = 1}^{s}} \ket{k||0}
\end{equation*}
Then, we can imagine the state that results from projecting the key register onto keys not in the support of $\vec{x}$.  Since there are at most $t$ keys, this will lead to an error of $O(t/\sqrt{2^{\secp}})$.  
For those keys that are distinct from $\{x_{j}\}$, we can apply an isometry on the purifying register that extracts out the elements of the relation that have input $k || 0$, to get the following state
\begin{equation*}
    \sum_{\vec{x} \in [N]^s}\sum_{(\vec{y}, \vec{z}) \in [N]^{t+s}_{\mathrm{dist}}} \left(\prod_{i = 1}^{s} \ket{z_i}\!\bra{x_i} A_i \right) \ket{y_1, \ldots, y_t} \ket{0} \otimes \ket{\{(0, y_i)\}_{i = 1}^{t}} \ket{\{(x_j, z_j)\}_{j = 1}^{s}}\,.
\end{equation*}
Notice that for the Haar random case, $y$ and $z$ are sampled independently instead of being sampled to be distinct from each other.  However, the probability that uniformly random $\vec{z}$ overlaps with $\vec{y}$ is on the order of $(t+s)^2 / 2^{n+\lambda}$.
Thus, the two states are close in trace distance, and our original construction is a pseudorandom state.
\section{Preliminaries} \label{sec:prelim}

We denote the security parameter by $\secp$. We assume that the reader is familiar with fundamentals of quantum computing, otherwise readers can refer to \cite{nielsen_chuang_2010}.

\newcommand{\Symgp}{\mathsf{Sym}}
\subsection{Notation}

\paragraph{Indexing and sets} We use the notation $[n]$ to refer to the set $\{1, \ldots, n\}$.  For a string $x \in \{0, 1\}^{n+m}$, let $x_{[1:n]}$ to denote the first $n$ bits of $x$.  For a finite set $T$, we use the binomial notation $\binom{T}{k}$ to refer to the set of all size-$k$ subsets of $T$.  We also use the notation $x \random T$ to indicate that $x$ is sampled uniformly at random from $T$.  For $N, \ell \in \mathbb{N}$, we let $N^{\downarrow \ell} = \prod_{i = 0}^{\ell-1} (N - i)$. We use $\uplus$ to denote the disjoint union of two sets.

\paragraph{Set products and the symmetric group}
We use $\Symgp_t$ to refer to the symmetric group over $t$ elements (i.e. the group of all permutations of $t$ elements).  Given a set $A$ and $t \in \mathbb{N}$, we use the notation $A^t$ to denote the $t$-fold Cartesian product of $A$, and the notation $A^t_{\mathrm{dist}}$ to denote distinct subspace of $A^t$, i.e. the vectors in $A^t$, $\vec{y} = (y_1, \ldots, y_t)$, such that for all $i \neq j$, $y_i \neq y_j$. We also define $\{\vec{x}\} := \bigcup_{i\in[t]} \{x_i\}$. 

\paragraph{Quantum states and distances}
A register $\reg{R}$ is a named finite-dimensional Hilbert space.  If $\reg{A}$ and $\reg{B}$ are registers, then $\reg{AB}$ denotes the tensor product of the two associated Hilbert spaces.  We denote by $\mathcal{D}(\reg{R})$ the density matrices over register $\reg{R}$.  For $\rho_{\reg{AB}} \in \mathcal{D}(\reg{AB})$, we let $\Tr_\reg{B}(\rho_{\reg{AB}}) \in \mathcal{D}(\reg{A})$ denote the reduced density matrix that results from taking the partial trace over $\reg{B}$.  We denote by $\TD(\rho, \rho') = \frac{1}{2} \norm{\rho - \rho'}_1$ the trace distance between $\rho$ and $\rho'$, where $\norm{X}_1 = \Tr(\sqrt{X^{\dagger} X})$ is the trace norm. For two pure (and possibly subnormalized) states $\ket{\psi}$ and $\ket{\phi}$, we use $\TD\qty(\ket{\psi},\ket{\phi})$ as a shorthand for $\TD\qty(\proj{\psi},\proj{\phi})$. We also say that $A \preceq B$ if $B - A$ is a positive semi-definite matrix.  For a permutation $\sigma \in \Symgp_t$, we let $(S_\sigma)_{\reg{R}_1\ldots\reg{R}_t}$ be the $nt$-qubit unitary that acts on registers $\reg{R}_1, \ldots, \reg{R}_t$ by permuting the registers according to $\sigma$.  That is,
\begin{equation*}
    S_{\sigma} \ket{x_1, \ldots, x_t} = \ket{x_{\sigma(1)}, \ldots, x_{\sigma(t)}}\,.
\end{equation*}
We denote by $\haarstates_n$ the Haar distribution over $n$-qubit states, and $\haarunitaries_n$ the Haar measure over $n$-qubit unitaries (i.e. the unique left and right invariant measure).

\subsection{Cryptographic Primitives}

In this section, we define the quantum cryptographic primitives that we reference in the rest of the paper, beginning with pseudorandom states~\cite{JLS18}.

\begin{definition}[Pseudorandom states]
    \label{def:prs}
    We say that a quantum polynomial-time algorithm $G$ is a \emph{pseudorandom state (PRS) generator} if the following holds:
    \begin{itemize}
        \item (Pure output) For all $\secp$ and $k \in \{0, 1\}^{\secp}$, the algorithm outputs
        \begin{equation*}
            G_{\secp}(k) = \proj{\psi_k}\,,
        \end{equation*}
        for some $n(\secp)$-qubit pure state $\ket{\psi_k}$.
        \item (Pseudorandomness) For all polynomials $t$ and quantum polynomial-time adversaries $\mathcal{A}$, there exists a negligible function $\epsilon$ such that for all $\secp$, 
        \begin{equation*}
            \left|\Pr_{k \leftarrow\{0, 1\}^{\secp}} \left[1 \leftarrow \mathcal{A}_{\secp}(\ket{\psi_{k}}^{\otimes t(\secp)})\right] - \Pr_{\ket{\psi} \leftarrow \Haar_{n(\secp)}} \left[1 \leftarrow \mathcal{A}_{\secp}(\ket{\psi}^{\otimes t(\secp)})\right]\right|\leq \epsilon(\secp)\,.    
        \end{equation*}
    \end{itemize}
    In the $\mathsf{iQHROM}$, both $G_\secp$ and $\mathcal{A}_{\secp}$ have oracle access to a family of unitaries $\{U_{\secp}\}_{\secp \in \mathbb{N}}$ sampled from the Haar measure on $\secp$ qubits.
\end{definition}

Pseudorandom state generators can be generalized into pseudorandom function-like states~\cite{AGQY22}, which are a family of states (indexed by a parameter $w$) that look as if they were all sampled independently from the Haar measure.

\begin{definition}[Pseudorandom function-like states]
    \label{def:prfs}
    A quantum polynomial-time algorithm $G$ is an {\em adaptively secure pseudorandom function-like state (APRFS)} generator if for all quantum polynomial-time adversaries $\mathcal{A}$, there exists a negligible function $\epsilon$ such that for all $\secp$,
     \[
    \left|\Pr_{k \leftarrow \{0,1\}^{\secparam} }\left[1\gets \Adversary_\lambda^{{\cal O}_{\sf PRFS}(k,\cdot)}\right] - \Pr_{{\cal O}_{\sf Haar}}\left[1\gets \Adversary_\lambda^{{\cal O}_{\sf Haar}(\cdot)}\right]\right| \le \eps(\lambda),
  \]
  where:
  \begin{itemize}
      \item ${\cal O}_{\sf PRFS}(k,\cdot)$, 
      on input $w \in \{0,1\}^{m(\secparam)}$, outputs $G_{\secparam}(k,w)$.
      \item ${\cal O}_{\sf Haar}(\cdot)$, 
      on input $w \in \{0,1\}^{m(\secparam)}$, outputs $\ket{\vartheta_w}$, where, for every $y \in \{0,1\}^{m(\secparam)}$,  $\ket{\vartheta_y} \leftarrow \Haar_{n(\lambda)}$. 
  \end{itemize}
    Moreover, the adversary $A_{\secparam}$ has \emph{classical} access to ${\cal O}_{\sf PRFS}(k,\cdot)$ and ${\cal O}_{\sf Haar}(\cdot)$. That is, we can assume without loss of generality that any query made to either oracle is measured in the computational basis.\footnote{In~\cite{AGQY22}, the authors further study a stronger security notation called \emph{quantum-accessible adaptively secure pseudorandom function-like states (QAPRFS)}, where the adversary has \emph{superposition} oracle access to ${\cal O}_{\sf PRFS}(k,\cdot)$ and ${\cal O}_{\sf Haar}(\cdot)$.}

    \par In the $\mathsf{iQHROM}$, both $G_\secp$ and $\mathcal{A}_{\secp}$ have oracle access to a family of unitaries $\{U_{\secp}\}_{\secp \in \mathbb{N}}$ sampled from the Haar measure on $\secp$ qubits.

    \par We say that $G$ is a $(
    \secp, m(\secp),n(\secp))$-APRFS generator to succinctly indicate that its input length is $m(\secp)$ and its output length is $n(\secp)$.
\end{definition}

Pseudorandom unitaries~\cite{JLS18} are the quantum equivalent of a pseudorandom function, in that an adversary can not distinguish the PRU from a truly Haar random unitary.

\begin{definition}[Pseudorandom unitaries]
    \label{def:pru}
    We say that a quantum polynomial-time algorithm $G$ is a pseudorandom unitary if for all quantum polynomial-time adversaries $\mathcal{A}$, there exists a negligible function $\epsilon$ such that for all $\secp$,
    \begin{equation*}
        \left|\mathop{\mathrm{Pr}}_{k \gets \{0, 1\}^{\secp}} \left[1 \gets \mathcal{A}_{\lambda}^{G_\secp(k)} \right]- \mathop{\mathrm{Pr}}_{\mathcal{U} \gets \mu_{n(\secp)}} \left[1 \gets \mathcal{A}_{\lambda}^{\mathcal{U}} \right]\right|  \leq \epsilon(\secp)\,.
    \end{equation*} 
    In the $\mathsf{iQHROM}$, both $G_\secp$ and $\mathcal{A}_{\secp}$ have oracle access to an additional family of unitaries $\{U_{\secp}\}_{\secp \in \mathbb{N}}$ sampled from the Haar measure on $\secp$ qubits.
\end{definition}

\subsection{Useful Lemmas}
Here we present useful quantum lemmas that should be familiar to a reader well versed in quantum computation.

\begin{lemma}[Gentle operator lemma, a special case of~{\cite[Lemma 9.4.2 \& Exercise 9.4.1]{WildeBook}}]   \label{lem:norm:proj}
Let $\ket{\psi}$ be a subnormalized state and $\Pi$ a projector. Then 
$\TD( \ket{\psi}, \Pi\ket{\psi} )
\leq
\sqrt{ 1 - \norm{\Pi\ket{\psi}}^2 }\,.$
\end{lemma}
\ifllncs
The proof of the following lemmas can be found in~\Cref{app:prelim}.
\fi
\begin{lemma}   \label{lem:norm:cs}
    Let $\ket{\phi}$ be a state and $\ket{\psi}$ a subnormalized state, $\|\ket{\psi}\| \geq |\braket{\psi}{\phi}|$.
\end{lemma}
\ifllncs
\else
\begin{proof}
    By Cauchy-Schwarz inequality, $|\braket{\psi}{\phi}|^2 \leq \|\ket{\psi}\|^2\|\ket{\phi}\|^2$. Since $\|\ket{\phi}\|=1$, we have $|\braket{\psi}{\phi}| \leq \|\ket{\psi}\|$.
\end{proof}
\fi

\begin{lemma}   \label{lem:norm:sub-state}
    Let $\set{\ket{i}}_{i\in [N]}$ be some set of orthonormal vectors. Let $\ket{\psi} = \sum_{i} \ket{i}\ket{\psi_i}$ be a normalized state. Let $\set{\alpha_i}_{i\in [N]}$, be a set of non-negative real numbers with $\alpha_i\geq \beta$. Define the vector $\ket{\phi} := \sum_{i} \alpha_i\ket{i}\ket{\psi_i}$. Then $\braket{\psi}{\phi} \geq \beta$.
\end{lemma}
\ifllncs
\else
\begin{proof}
    Notice that $\braket{\psi}{\phi} = \sum_i \alpha_i\braket{\psi_i}{\psi_i}$. Since $\alpha_i\geq \beta$ and $\braket{\psi_i}{\psi_i} \geq 0$, we have $\braket{\psi}{\phi} \geq \beta \cdot \sum_i\braket{\psi_i}{\psi_i}$. Moreover, since $\ket{\psi}$ is normalized, we have $\braket{\psi}{\psi} = \sum_i \braket{\psi_i}{\psi_i} = 1$. Hence, $\braket{\psi}{\phi} \geq \beta$.
\end{proof}
\fi

\begin{lemma}[Ricochet property]
    Let $\ket{\Omega}$ be an EPR pair on registers $\reg{AB}$ and $U$ be a unitary acting on $\reg{A}$, then the following holds:
    \begin{equation*}
        (U \otimes \id)\ket{\Omega} = (\id \otimes U^{\intercal}) \ket{\Omega}\,.
    \end{equation*}
\end{lemma}
\section{Ma-Huang's Path-Recording Framework}   \label{sec:PR_framework}
We first recall the path-recording framework by Ma and Huang~\cite{MH24}. Most of the text is copied verbatim from~\cite{MH24}, but is contained here for the sake of completeness. 

\subsection{Oracle Adversary} 
We present the following definition of an oracle adversary. 

\begin{definition}[Oracle Adversary]
An oracle adversary $\Adversary$ is a quantum algorithm that makes queries to an oracle ${\cal O}$ that acts on the first $n$ qubits of the adversary's space, which we call the $\reg{A}$ register. The adversary also has an $m$-qubit ancillary space, which we call the $\reg{B}$ register. A $t$-query adversary $\Adversary$ specified by a $t$-tuple of unitaries $(A_{\reg{A}\reg{B}}^{(1)},\ldots,A_{\reg{A}\reg{B}}^{(t)})$.
\end{definition}

\begin{definition}[Oracle Adversary's view after $t$ queries]
Given a $t$-query adversary $\Adversary$ specified by a $t$-tuple of unitaries $(A_{\reg{A}\reg{B}}^{(1)},\ldots,A_{\reg{A}\reg{B}}^{(t)})$, we define the adversary's view after $t$ queries as: 
$$\ket{\Adversary_t^{{\cal  O}}}_{\reg{A} \reg{B}} = \prod_{i=1}^t \left( {\cal O}_{\reg{A}} \cdot A_{\reg{A} \reg{B}}^{(i)} \right) \ket{0}_{\reg{A}\reg{B}}  $$
Here, ${\cal O}$ represents the $n$-qubit oracle, and $A_{\reg{A} \reg{B}}^{(i)}$ is the unitary operation applied by the adversary between the $(i-1)^{th}$ and $i^{th}$ oracle queries. For an arbitrary $t$, we denote as $\ket{{\cal A}^{{\cal O}}}_{\reg{A}\reg{B}}$. 
\end{definition} 

\paragraph{Generalizations.} We consider generalized oracle adversaries in this work, where the adversary has access to two or more oracles. We also consider a restricted adversary who only makes selective calls. 

\begin{definition}[Multi-Oracle Adversary] 
An oracle adversary $\Adversary$ is a quantum algorithm that makes queries to $\ell$ oracles ${\cal O}_1,\ldots,{\cal O}_{\ell}$ each of which acts on the first $n$ qubits of the adversary's space, which we call the $\reg{A}$ register. The adversary also has an $m$-qubit ancillary space, which we call the $\reg{B}$ register. A $t$-query adversary $\Adversary$ specified by a $t$-tuple of unitaries $(A_{\reg{A}\reg{B}}^{(1)},\ldots,A_{\reg{A}\reg{B}}^{(t)})$. Here, $t$ denotes the total number of queries to all of the oracles ${\cal O}_1,\ldots,{\cal O}_{\ell}$. 
\end{definition} 

\begin{definition}[Multi-Oracle Adversary's view after $t$ queries]
Given a $t$-query multi-oracle adversary $\Adversary$ specified by a $t$-tuple of unitaries $(A_{\reg{A}\reg{B}}^{(1)},\ldots,A_{\reg{A}\reg{B}}^{(t)})$, we define the adversary's view after $t$ queries as: 
$$\ket{\Adversary_t^{{\cal  O}_1,\ldots,{\cal O}_t}}_{\reg{A} \reg{B}} = \prod_{i=1}^t \left( {\cal O}_{\reg{A}}^{(i)} \cdot A_{\reg{A} \reg{B}}^{(i)} \right) \ket{0}_{\reg{A}\reg{B}}  $$
Here, ${\cal O}^{(i)} \in \{{\cal O}_1,\ldots,{\cal O}_{\ell}\}$ represents one of the $\ell$ $n$-qubit oracles, and $A_{\reg{A} \reg{B}}^{(i)}$ is the unitary operation applied by the adversary between the $(i-1)^{th}$ and $i^{th}$ oracle queries. For an arbitrary $t$, we denote as $\ket{{\cal A}^{{\cal O}_1,\ldots,{\cal O}_t}}_{\reg{A}\reg{B}}$. 
\par When the number of queries to each of the oracles needs to be made explicit, we use the notation $(t_1,\ldots,t_{\ell})$-query multi-oracle adversary. In this case, $\Adversary$ makes $t_i$ queries to the oracle ${\cal O}_i$
\end{definition} 

\subsection{Relation States}

\noindent Before we recall the definition of relation states, we first define size-$t$ relations. A size-$t$ relation $R$ is represented by a multiset of $t$ tuples $R=\{(x_1,y_1),\ldots,(x_t,y_t)\}$, where $x_i \in [N],y_i \in [N]$ for all $i \in [t]$. We define ${\rm Im}(R) := \{y_1,\ldots,y_t\}$ and ${\rm Dom}(R) := \set{x_1,\ldots,x_t}$.
\par We define relation states below. 

\begin{definition}[Relation States]
\label{def:relationstates}
For $0 \leq t \leq N$ and a size-$t$ relation $R=\{(x_1,y_1),\ldots,(x_t,y_t)\}$, define the corresponding relation state to be the unit vector: 
$$\ket{R}_{\reg{A}\reg{B}} = \alpha_R \cdot \sum_{\pi \in \Symgp_t} S_{\pi} \ket{x_1,\ldots,x_t} \otimes S_{\pi} \ket{y_1,\ldots,y_t},$$
where: 
$$\alpha_R =  \sqrt{\frac{\prod_{x,y \in [N]}\left( \sum_{i=1}^t \delta_{(x_i,y_i) = (x,y)} \right)!}{t!}}\,.$$
\end{definition}

\begin{definition}[Injective Relation]
Let $t,N \in \mathbb{N}$. A relation $R=\{(x_1,y_1),\ldots,(x_t,y_t)\}$ is an injective relation if $(y_1,\ldots,y_t) \in [N]^{t}_{\dist}$. The set of all injective relations consisting of exactly $t$ pairs is denoted by $\rel^{\inj}_t$. Let $\rel^{\inj} = \cup_{j=0}^N \rel_j^{\inj}$. 
\end{definition}

\noindent For an injective relation $R$, note that $\alpha_R = \frac{1}{\sqrt{t!}}$, where $\alpha_R$ is defined in~\Cref{def:relationstates}.   

\subsection{Path-Recording Isometry ($\mathsf{PR}$)} Consider the following: for some $N \in \mathbb{N}$,
\begin{itemize}
    \item $x$ is an element of $[N]$, 
    \item $R \in \rel^{\inj}$ is an injective relation, over pairs in $[N] \times [N]$, of size $|R| < N$. 
\end{itemize}
The linear map ${\sf PR}$, on registers $\reg{A}$ and $\reg{E}$, is defined as follows: 
$${\sf PR}_{\reg{A}\reg{E}}: \ket{x}_{\reg{A}}\ket{R}_{\reg{E}} \mapsto \frac{1}{\sqrt{N - |R|}} \sum_{\substack{y \in [N],\\ y \notin \mathrm{Im}(R)}} \ket{y}_{\reg{A}} \ket{R \cup \{(x,y)\}}_{\reg{E}}\,.$$

\paragraph{Partial Isometry}~\cite{MH24} showed that ${\sf PR}$ is an isometry on some subspaces. Formally, they prove the following lemma.

\begin{lemma}[Lemma 4.1 of \cite{MH24}]
The path-recording linear map ${\sf PR}$, on the registers $(\reg{A},\reg{E})$, is an isometry on the subspace spanned by the states $\ket{x}_{\reg{A}}\ket{R}_{\reg{E}}$ for $x \in [N]$ and $R \in \rel^{\inj}$ such that $|R| < N$. 
\end{lemma}
 
\paragraph{Indistinguishability Theorem.} We import the following theorem from~\cite{MH24}. Informally, it states that a $t$-query oracle adversary cannot distinguish a Haar unitary versus a path-recording isometry. 

\begin{theorem}[Theorem 5 of \cite{MH24}]
\label{thm:MH24}
Let $\Adversary$ be a $t$-oracle adversary. Then: 
$$\TD\left( \underset{{{\cal O} \leftarrow \haarunitaries_n}}{\expect} \ketbra{\Adversary_t^{{\cal O}}}{\Adversary_t^{{\cal O}}},\ \Tr_{\reg{E}}\left( \ketbra{\Adversary_t^{{\sf PR}_{\reg{A}\reg{E}}}}{\Adversary_t^{{\sf PR}_{\reg{A}\reg{E}}}}_{\reg{A}\reg{B}\reg{E}}  \right) \right) \leq \frac{2t(t-1)}{N+1} $$
\end{theorem}

\noindent We have the following simple corollary. 

\begin{corollary}
\label{cor:ind:twopr}
Let $\Adversary$ be a $t$-oracle adversary. Then: 
$$\TD\left( \underset{{{\cal O}_1,{\cal O}_2 \leftarrow \haarunitaries_n}}{\expect} \ketbra{\Adversary_t^{{\cal O}_1,{\cal O}_2}}{\Adversary_t^{{\cal O}_1,{\cal O}_2}},\ \Tr_{\reg{E}_1\reg{E}_2}\left( \ketbra{\Adversary_t^{{\sf PR}_{\reg{A}\reg{E}_1},{{\sf PR}_{\reg{A}\reg{E}_2}} }}{\Adversary_t^{{\sf PR}_{\reg{A}\reg{E}_1},{{\sf PR}_{\reg{A}\reg{E}_2}} }}_{\reg{A}\reg{B}\reg{E}_1\reg{E}_2}  \right) \right) \leq \frac{4t(t-1)}{N+1} $$
\end{corollary}

\section{Path-Recording Framework: New Observations}    \label{sec:PR_New}
We discuss new observations about the path-recording framework in this section. Before that, we state some definitions related to sets whose prefixes are all distinct.

\subsection{Definitions}

\begin{definition}[Strong $\ell$-fold $\secp$-prefix collision-free sets]
\label{def:s_l_fold_collisions}
Let $n,\secp \in \mathbb{N}$ with $\secp\leq n$. Let ${\cal S}$ be a set with elements from $\{0,1\}^{n}$. We say that ${\cal S}$ is {\em strong $\ell$-fold $\secp$-prefix collision-free} if the following holds: for all two $i$-sized subsets $S_1,S_2$, for any $i \leq \ell$, it holds that $\bigoplus_{x\in S_1} x_{[1:\secp]} = \bigoplus_{y\in S_2} y_{[1:\secp]} $ if and only if $S_1=S_2$. We denote $\setsofsets^{\cfree(\ell,\secp)}_{n}$ to be the set of all strong $\ell$-fold $\secp$-prefix collision-free sets. 
\end{definition}

\begin{lemma}\label{lem:cf_set_size}
    Let $S \in \setsofsets^{\cfree(\ell,\secp)}_{n}$, where $\setsofsets^{\cfree(\ell,\secp)}_{n}$ is as defined in~\Cref{def:s_l_fold_collisions}. Define the following: 
    $$\cfreeset_{\ell,\secp}(S) = \set{y\in\bit^{n}| S\cup\set{y}\text{ is strong }\ell\text{-fold }\secp\text{-prefix collision-free}} $$
    then $$|\cfreeset_{\ell,\secp}(S)|\geq 2^{n} - \ell|S|^{2\ell}2^{n-\secp}.$$ 
\end{lemma}

\noindent The proof of~\Cref{lem:cf_set_size} can be found in~\Cref{app:PR_New}.

\begin{definition}[Prefix Collision-Free Relations]
Let $t,n,\secp\in\mathbb{N}$ with $\secp\leq n$. A relation $R=\{(x_1,y_1),\ldots,(x_t,y_t)\}$ is strong $\ell$-fold $\secp$-prefix collision-free if $\{y_1,\ldots,y_t\} \in \setsofsets^{\cfree(\ell,\secp)}_{n}$. We denote $\setsofrels^{\cfree(\ell,\secp)}$ to be the set of all strong $\ell$-fold $\secp$-prefix collision-free relations.  
\end{definition}

\noindent Note that $\setsofrels^{\cfree(\ell,\secp)}$ is a subset of $\rel^{\inj}$.

\subsection{Recording $\ell$-Fold Collision-Free Paths}  \label{sec:cfPR}
We define a variant of the path-recording linear map below. Later we show the indistinguishability of this variant from the original path-recording map. 

\paragraph{Prefix $\ell$-fold Collision-Free Path Recording Linear Maps.} Consider the following: for some $n,\secp,\ell\in\mathbb{N}$ such that $\secp \leq n$,
\begin{itemize}
    \item $x$ is an element of $\{0,1\}^{n}$, 
    \item $R_1,R_2 \in \setsofrels^{\cfree(\ell,\secp)}$ are strong $\ell$-fold $\secp$-prefix collision-free relation, over pairs in $\{0,1\}^{n} \times \{0,1\}^{n}$, such that of $R_1\cup R_2\in \setsofrels^{\cfree(\ell,\secp)}$, $R_1\cap R_2 = \emptyset$ and $|R_1\cup R_2| < 2^{n}$. 
\end{itemize}
We define two prefix collision-free path linear maps, on registers $\reg{A}$ and $\reg{E} :=(\reg{E_1},\reg{E_2})$, as follows: 
\ifllncs
    \begin{equation*}
    \begin{split}
        \pcfpr{\ell}{\secp}^{(\reg{E_1})}_{\reg{A}\reg{E}}:\ & \ket{x}_{\reg{A}}\ket{R_1}_{\reg{E_1}}\ket{R_2}_{\reg{E_2}} \\
        \mapsto & \frac{1}{\sqrt{|\cfreeset_{\ell,\secp}({\rm Im}(R_1 \cup R_2)) }|} \sum_{\substack{y \in \cfreeset_{\ell,\secp}({\rm Im}(R_1 \cup R_2))}} \ket{y}_{\reg{A}} \ket{R_1 \cup \{(x,y)\}}_{\reg{E_1}}\ket{R_2}_{\reg{E_2}},   
    \end{split}
    \end{equation*}
    \begin{equation*}
    \begin{split}
        \pcfpr{\ell}{\secp}^{(\reg{E_2})}_{\reg{A}\reg{E}}:\ & \ket{x}_{\reg{A}}\ket{R_1}_{\reg{E_1}}\ket{R_2}_{\reg{E_2}} \\
        \mapsto & \frac{1}{\sqrt{|\cfreeset_{\ell,\secp}({\rm Im}(R_1 \cup R_2)) }|} \sum_{\substack{y \in \cfreeset_{\ell,\secp}({\rm Im}(R_1 \cup R_2))}} \ket{y}_{\reg{A}} \ket{R_1}_{\reg{E_1}}\ket{R_2  \cup \{(x,y)\}}_{\reg{E_2}}
    \end{split}
    \end{equation*}
\else
    $$\pcfpr{\ell}{\secp}^{(\reg{E_1})}_{\reg{A}\reg{E}}: \ket{x}_{\reg{A}}\ket{R_1}_{\reg{E_1}}\ket{R_2}_{\reg{E_2}} \mapsto \frac{1}{\sqrt{|\cfreeset_{\ell,\secp}({\rm Im}(R_1 \cup R_2)) }|} \sum_{\substack{y \in \cfreeset_{\ell,\secp}({\rm Im}(R_1 \cup R_2))}} \ket{y}_{\reg{A}} \ket{R_1 \cup \{(x,y)\}}_{\reg{E_1}}\ket{R_2}_{\reg{E_2}},$$
    $$\pcfpr{\ell}{\secp}^{(\reg{E_2})}_{\reg{A}\reg{E}}: \ket{x}_{\reg{A}}\ket{R_1}_{\reg{E_1}}\ket{R_2}_{\reg{E_2}} \mapsto \frac{1}{\sqrt{|\cfreeset_{\ell,\secp}({\rm Im}(R_1 \cup R_2)) }|} \sum_{\substack{y \in \cfreeset_{\ell,\secp}({\rm Im}(R_1 \cup R_2))}} \ket{y}_{\reg{A}} \ket{R_1}_{\reg{E_1}}\ket{R_2  \cup \{(x,y)\}}_{\reg{E_2}}$$
\fi

Intuitively, these prefix collision-free path recording oracles are similar to the original path recording oracle, except that they only output $y$ that are prefix collision-free with respect to the relation, instead of outputting $y$ that are distinct from the image of the relation.  We will show that most $y$ that are distinct are also prefix collision-free (for a suitably large prefix), which will imply that the prefix collision-free path recording oracle acts similarly to the original path recording oracle.

\paragraph{Indistinguishability.} We prove the following theorem: 

\begin{theorem} \label{thm:ind:pcfpr:pr}
Let $\Adversary$ be a $t$-oracle adversary. Then:
$$\TD(\rho,\sigma) \leq O\left(\frac{\sqrt{\ell}t^{\ell+1}}{2^{\secp/2}}\right),$$
where: 
$$\rho = \Tr_{\reg{E}_1\reg{E}_2}\left( \ketbra{\Adversary_t^{{\pcfpr{\ell}{\secp}}_{\reg{A}\reg{E}}^{(\reg{E}_1)},{{\pcfpr{\ell}{\secp}}_{\reg{A}\reg{E}}^{(\reg{E}_2)}} }}{\Adversary_t^{{\pcfpr{\ell}{\secp}}_{\reg{A}\reg{E}}^{(\reg{E}_1)},{{\pcfpr{\ell}{\secp}}_{\reg{A}\reg{E}}^{(\reg{E}_2)}} }}_{\reg{A}\reg{B}\reg{E}_1\reg{E}_2}  \right)$$
$$\sigma = \Tr_{\reg{E}_1\reg{E}_2}\left( \ketbra{\Adversary_t^{{\sf PR}_{\reg{A}\reg{E}_1},{{\sf PR}_{\reg{A}\reg{E}_2}} }}{\Adversary_t^{{\sf PR}_{\reg{A}\reg{E}_1},{{\sf PR}_{\reg{A}\reg{E}_2}} }}_{\reg{A}\reg{B}\reg{E}_1\reg{E}_2}  \right) $$
\end{theorem}

\ifllncs
\noindent The proof of~\Cref{thm:ind:pcfpr:pr} is deferred to~\Cref{app:PR_New}. Applying the triangle inequality on~\Cref{thm:ind:pcfpr:pr} and~\Cref{cor:ind:twopr}, we have the following corollary.
\else

\begin{proof}
Without loss of generality, assume $\Adversary = (A_{\reg{A}\reg{B}}^{(1)},\ldots,A_{\reg{A}\reg{B}}^{(t)})$. 
We start by defining the following hybrids.

\noindent $\hybrid_1$: Output $$\Tr_{\reg{E}_1\reg{E}_2}\left( \ketbra{\Adversary_t^{{\sf PR}_{\reg{A}\reg{E}_1},{{\sf PR}_{\reg{A}\reg{E}_2}} }}{\Adversary_t^{{\sf PR}_{\reg{A}\reg{E}_1},{{\sf PR}_{\reg{A}\reg{E}_2}} }}_{\reg{A}\reg{B}\reg{E}_1\reg{E}_2}  \right) .$$

\noindent $\hybrid_{2.j}$, for $j\in \set{0,\ldots,t}$: Let $$\ket{\psi_j} = \prod_{i=j+1}^t \left( {\cal O}_{\reg{A}}^{(i)} \cdot A_{\reg{A} \reg{B}}^{(i)} \right)\prod_{i=1}^j \left( \Tilde{{\cal O}}_{\reg{A}}^{(i)} \cdot A_{\reg{A} \reg{B}}^{(i)} \right) \ket{0}_{\reg{A}\reg{B}},$$ with ${\cal O}_{\reg{A}}^{(i)}\in\set{{\sf PR}_{\reg{A}\reg{E}_1},{{\sf PR}_{\reg{A}\reg{E}_2}}}$ and $\Tilde{{\cal O}}_{\reg{A}}^{(i)}\in\set{{\pcfpr{\ell}{\secp}}_{\reg{A}\reg{E}}^{(\reg{E}_1)},{{\pcfpr{\ell}{\secp}}_{\reg{A}\reg{E}}^{(\reg{E}_2)}}}$. 
Output $$\Tr_{\reg{E}_1\reg{E}_2}\left( \ketbra{\psi_j}{\psi_j}\right).$$
    
\noindent $\hybrid_3$: Output $$\Tr_{\reg{E}_1\reg{E}_2}\left( \ketbra{\Adversary_t^{{\pcfpr{\ell}{\secp}}_{\reg{A}\reg{E}}^{(\reg{E}_1)},{{\pcfpr{\ell}{\secp}}_{\reg{A}\reg{E}}^{(\reg{E}_2)}} }}{\Adversary_t^{{\pcfpr{\ell}{\secp}}_{\reg{A}\reg{E}}^{(\reg{E}_1)},{{\pcfpr{\ell}{\secp}}_{\reg{A}\reg{E}}^{(\reg{E}_2)}} }}_{\reg{A}\reg{B}\reg{E}_1\reg{E}_2}  \right).$$

\noindent Note that for $\ket{\psi_0} = \ket{\Adversary_t^{{\sf PR}_{\reg{A}\reg{E}_1},{{\sf PR}_{\reg{A}\reg{E}_2}} }}$ and $\ket{\psi_t} = \ket{\Adversary_t^{{\pcfpr{\ell}{\secp}}_{\reg{A}\reg{E}}^{(\reg{E}_1)},{{\pcfpr{\ell}{\secp}}_{\reg{A}\reg{E}}^{(\reg{E}_2)}} }}$. Hence, $\hybrid_1$ is identical to $\hybrid_{2.0}$ and $\hybrid_3$ is identical to $\hybrid_{2.t}$. 

\begin{lemma} \label{lem:pcfpr:hybrids}
    The trace distance between the output of $\hybrid_{2.(j-1)}$ and $\hybrid_{2.(j)}$, for $j\in [t]$, is $\frac{\sqrt{\ell}(j-1)^{\ell}}{2^{\secp/2}}$.
\end{lemma}
\begin{proof}[Proof of~\Cref{lem:pcfpr:hybrids}]
    Since taking partial trace cannot increase the trace distance, we instead find the trace distance between $\ket{\psi_{j-1}}$ and $\ket{\psi_{j}}$. Recall that $$\ket{\psi_{j-1}} = \prod_{i=j}^t \left( {\cal O}_{\reg{A}}^{(i)} \cdot A_{\reg{A} \reg{B}}^{(i)} \right)\prod_{i=1}^{j-1} \left( \Tilde{{\cal O}}_{\reg{A}}^{(i)} \cdot A_{\reg{A} \reg{B}}^{(i)} \right) \ket{0}_{\reg{A}\reg{B}},$$ and $$\ket{\psi_j} = \prod_{i=j+1}^t \left( {\cal O}_{\reg{A}}^{(i)} \cdot A_{\reg{A} \reg{B}}^{(i)} \right)\prod_{i=1}^j \left( \Tilde{{\cal O}}_{\reg{A}}^{(i)} \cdot A_{\reg{A} \reg{B}}^{(i)} \right) \ket{0}_{\reg{A}\reg{B}}.$$ Since applying the same channel cannot increase the trace distance, we only need to calculate the trace distance between $$\ket{\phi_{j-1}} := {\cal O}_{\reg{A}}^{(j)} \cdot A_{\reg{A} \reg{B}}^{(j)} \prod_{i=1}^{j-1} \left( \Tilde{{\cal O}}_{\reg{A}}^{(i)} \cdot A_{\reg{A} \reg{B}}^{(i)} \right) \ket{0}_{\reg{A}\reg{B}},$$ and $$\ket{\phi_j} := \prod_{i=1}^j \left( \Tilde{{\cal O}}_{\reg{A}}^{(i)} \cdot A_{\reg{A} \reg{B}}^{(i)} \right) \ket{0}_{\reg{A}\reg{B}}.$$
    Note that the trace distance between pure states $\ket{\phi_{j-1}}$ and $\ket{\phi_{j}}$ is $\sqrt{1-|\braket{\phi_{j-1}}{\phi_{j}}|^2}$. 
    Notice that $$\ket{\phi_{j-1}} = {\cal O}_{\reg{A}}^{(j)} \ket{\theta_j},$$ and $$\ket{\phi_j} := \Tilde{{\cal O}}_{\reg{A}}^{(j)} \ket{\theta_j},$$ for $$\ket{\theta_j} := A_{\reg{A} \reg{B}}^{(j)} \prod_{i=1}^{j-1} \left( \Tilde{{\cal O}}_{\reg{A}}^{(i)} \cdot A_{\reg{A} \reg{B}}^{(i)} \right) \ket{0}_{\reg{A}\reg{B}}.$$

    \noindent Hence, $\braket{\phi_{j-1}}{\phi_{j}} = \bra{\theta_j}\left({\cal O}_{\reg{A}}^{(j)}\right)^{\dagger}\Tilde{{\cal O}}_{\reg{A}}^{(j)}\ket{\theta_j}$ and we will look at the behavior of $\left({\cal O}_{\reg{A}}^{(j)}\right)^{\dagger}\Tilde{{\cal O}}_{\reg{A}}^{(j)}$. We assume ${\cal O}_{\reg{A}}^{(j)} = {\sf PR}_{\reg{A}\reg{E}_1}$ and $\Tilde{{\cal O}}_{\reg{A}}^{(j)} = {\pcfpr{\ell}{\secp}}_{\reg{A}\reg{E}}^{(\reg{E}_1)}$, the other case works by symmetry. 
    
    \noindent Notice that if 
    \begin{itemize}
        \item $x$ is an element of $\{0,1\}^{n}$, 
        \item $R_1,R_2 \in \setsofrels^{\cfree(\ell,\secp)}$ are strong $\ell$-fold $\secp$-prefix collision-free relations, over pairs in $\{0,1\}^{n} \times \{0,1\}^{n}$, such that of $R_1\cup R_2\in \setsofrels^{\cfree(\ell,\secp)}$, $R_1\cap R_2 = \emptyset$ and $|R_1\cup R_2| < 2^{n}$. 
    \end{itemize}
    \begin{equation*}
    \begin{split}
        & \pcfpr{\ell}{\secp}^{(\reg{E_1})}_{\reg{A}\reg{E}}:\ \ket{x}_{\reg{A}}\ket{R_1}_{\reg{E_1}}\ket{R_2}_{\reg{E_2}} \\
        & \mapsto\ \frac{1}{\sqrt{|\cfreeset_{\ell,\secp}({\rm Im}(R_1 \cup R_2)) }|} \sum_{\substack{y \in \cfreeset_{\ell,\secp}({\rm Im}(R_1 \cup R_2))}} \ket{y}_{\reg{A}} \ket{R_1 \cup \{(x,y)\}}_{\reg{E_1}}\ket{R_2}_{\reg{E_2}},
    \end{split}
    \end{equation*}
    and 
    $${\sf PR}_{\reg{A}\reg{E_1}}: \ket{x}_{\reg{A}}\ket{R_1}_{\reg{E_1}}\ket{R_2}_{\reg{E_2}} \mapsto \frac{1}{\sqrt{N - |R_1|}} \sum_{\substack{y \in [N],\\ y \notin \mathrm{Im}(R_1)}} \ket{y}_{\reg{A}} \ket{R_1 \cup \{(x,y)\}}_{\reg{E_1}}\ket{R_2}_{\reg{E_2}},$$

    Since $\cfreeset_{\ell,\secp}({\rm Im}(R_1 \cup R_2))\subseteq [N]\setminus\mathrm{Im}(R_1)$, $$\left({\sf PR}_{\reg{A}\reg{E_1}}\right)^{\dagger}\pcfpr{\ell}{\secp}^{(\reg{E_1})}_{\reg{A}\reg{E}} \left(\ket{x}_{\reg{A}}\ket{R_1}_{\reg{E_1}}\ket{R_2}_{\reg{E_2}}\right) = \sqrt{\frac{|\cfreeset_{\ell,\secp}({\rm Im}(R_1 \cup R_2))|}{N-|R_1|}}\ket{x}_{\reg{A}}\ket{R_1}_{\reg{E_1}}\ket{R_2}_{\reg{E_2}}.$$
    Hence, $$\left({\sf PR}_{\reg{A}\reg{E_1}}\right)^{\dagger}\pcfpr{\ell}{\secp}^{(\reg{E_1})}_{\reg{A}\reg{E}} = \sum_{\substack{x,R_1,R_2\\ R_1,R_2 \in \setsofrels^{\cfree(\ell,\secp)}\\ R_1\cap R_2 = \set{}\\ R_1\cup R_2\in \setsofrels^{\cfree(\ell,\secp)}}} \sqrt{\frac{|\cfreeset_{\ell,\secp}({\rm Im}(R_1 \cup R_2))|}{N-|R_1|}} 
    \ketbra{x}{x}_{\reg{A}}\otimes\ketbra{R_1}{R_1}_{_{\reg{E_1}}}\otimes \ketbra{R_2}{R_2}_{\reg{E_2}}.$$

    \noindent By~\Cref{lem:cf_set_size}, $$\sqrt{\frac{|\cfreeset_{\ell,\secp}({\rm Im}(R_1 \cup R_2))|}{N-|R_1|}}\geq\sqrt{\frac{2^n-\ell|{\rm Im}(R_1 \cup R_2)|^{2\ell}2^{n-\secp}}{2^{n}}}\geq\sqrt{1-\frac{\ell|{\rm Im}(R_1 \cup R_2)|^{2\ell}}{2^{\secp}}}.$$

    \noindent Hence, we can write $$\left({\sf PR}_{\reg{A}\reg{E_1}}\right)^{\dagger}\pcfpr{\ell}{\secp}^{(\reg{E_1})}_{\reg{A}\reg{E}}\succeq I_{\reg{A}}\otimes \left(\bigoplus_{i=1}^{2^{n/2}}\sum_{\substack{R_1,R_2 \in \setsofrels^{\cfree(\ell,\secp)}\\ R_1\cap R_2 = \set{}\\ R_1\cup R_2\in \setsofrels^{\cfree(\ell,\secp)}\\ |R_1\cup R_2|=i}}\sqrt{1-\frac{\ell i^{2\ell}}{2^{\secp}}} \ketbra{R_1}{R_1}_{\reg{E_1}}\otimes \ketbra{R_2}{R_2}_{\reg{E_2}}\right).$$

    \noindent Note that since $\ket{\theta_j}$ is formed by $j-1$ total queries to $\Tilde{{\cal O}}_{\reg{A}}$, 
    \begin{multline*}
        \ket{\theta_j}\in\mathsf{span}\set{\ket{x}_{\reg{AB}}\ket{R_1}_{\reg{E_1}}\ket{R_2}_{\reg{E_2}}:x\in\bit^{n+m},R_1,R_2 \in \setsofrels^{\cfree(\ell,\secp)},\\ R_1\cap R_2 = \emptyset,R_1\cup R_2\in \setsofrels^{\cfree(\ell,\secp)},|R_1\cup R_2|=j-1}.
    \end{multline*}
    Hence, $$\left|\braket{\phi_{j-1}}{\phi_{j}}\right|^2 = \left|\bra{\theta_j}\left({\cal O}_{\reg{A}}^{(j)}\right)^{\dagger}\Tilde{{\cal O}}_{\reg{A}}^{(j)}\ket{\theta_j}\right|^2\geq 1-\frac{\ell (j-1)^{2\ell}}{2^{\secp}}.$$

    \noindent The trace distance between $\ket{\phi_{j-1}}$ and $\ket{\phi_{j}}$ is $\sqrt{1-|\braket{\phi_{j-1}}{\phi_{j}}|^2}\leq\frac{\sqrt{\ell}(j-1)^{\ell}}{2^{\secp/2}}.$
    This completes the proof of~\Cref{lem:pcfpr:hybrids}.
\end{proof}
By~\Cref{lem:pcfpr:hybrids}, the total trace distance between $\hybrid_1$ and $\hybrid_3$ is $\sum_{j=1}^{t} \frac{\sqrt{\ell}(j-1)^{\ell}}{2^{\secp/2}}\leq \frac{\sqrt{\ell}t^{\ell+1}}{2^{\secp/2}}$. This completes the proof of~\Cref{thm:ind:pcfpr:pr}.
\end{proof}

\noindent Applying triangle inequality on~\Cref{thm:ind:pcfpr:pr} and~\Cref{cor:ind:twopr}, we have the following corollary. 
\fi

\begin{corollary}
\label{cor:ind:twopcfpr}
Let $\Adversary$ be a $t$-oracle adversary. Then:
$$\TD(\rho,\sigma) \leq O\left(\frac{\sqrt{\ell}t^{\ell+1}}{2^{\secp/2}}\right) + \frac{4t(t-1)}{N+1},$$
where: 
$$\rho = \Tr_{\reg{E}_1\reg{E}_2}\left( \ketbra{\Adversary_t^{{\pcfpr{\ell}{\secp}}_{\reg{A}\reg{E}}^{(\reg{E}_1)},{{\pcfpr{\ell}{\secp}}_{\reg{A}\reg{E}}^{(\reg{E}_2)}} }}{\Adversary_t^{{\pcfpr{\ell}{\secp}}_{\reg{A}\reg{E}}^{(\reg{E}_1)},{{\pcfpr{\ell}{\secp}}_{\reg{A}\reg{E}}^{(\reg{E}_2)}} }}_{\reg{A}\reg{B}\reg{E}_1\reg{E}_2}  \right)$$
$$\sigma = \underset{{{\cal O}_1,{\cal O}_2 \leftarrow \haarunitaries_n}}{\expect} \ketbra{\Adversary_t^{{\cal O}_1,{\cal O}_2}}{\Adversary_t^{{\cal O}_1,{\cal O}_2}} $$
\end{corollary}

\noindent We also look at a special case of $\ell$-fold collision-free path recording where $\ell = 1$ and $\secp = n$, which we call the collision-free path recording oracle. We define these as follows:
\paragraph{Collision-Free Path Recording Linear Maps.} Consider the following: for some $n\in\mathbb{N}$,
\begin{itemize}
    \item $x$ is an element of $\{0,1\}^{n}$, 
    \item $R_1,R_2 \in \setsofrels^{\inj}$ are injective relations, over pairs in $\{0,1\}^{n} \times \{0,1\}^{n}$, such that of $R_1\cup R_2\in \setsofrels^{\inj}$, $R_1\cap R_2 = \emptyset$ and $|R_1\cup R_2| < 2^{n}$. 
\end{itemize}
We define two collision-free path linear maps, on registers $\reg{A}$ and $\reg{E}=(\reg{E_1},\reg{E_2})$, as follows: 
$$\pr^{(\reg{E_1})}_{\reg{A}\reg{E}}: \ket{x}_{\reg{A}}\ket{R_1}_{\reg{E_1}}\ket{R_2}_{\reg{E_2}} \mapsto \frac{1}{\sqrt{N-|{\rm Im}(R_1 \cup R_2)}|} \sum_{\substack{y \in [N]\setminus {\rm Im}(R_1 \cup R_2))}} \ket{y}_{\reg{A}} \ket{R_1 \cup \{(x,y)\}}_{\reg{E_1}}\ket{R_2}_{\reg{E_2}},$$
$$\pr^{(\reg{E_2})}_{\reg{A}\reg{E}}: \ket{x}_{\reg{A}}\ket{R_1}_{\reg{E_1}}\ket{R_2}_{\reg{E_2}} \mapsto \frac{1}{\sqrt{N-|{\rm Im}(R_1 \cup R_2) }|} \sum_{\substack{y \in [N]\setminus {\rm Im}(R_1 \cup R_2)}} \ket{y}_{\reg{A}} \ket{R_1}_{\reg{E_1}}\ket{R_2  \cup \{(x,y)\}}_{\reg{E_2}}$$

\noindent We highlight that the difference between this oracle and two independent path recording oracles is that the two collision-free path recording oracles are aware not only of their own relation, but also the image of the other relation.  Setting $\ell = 1$, $\secp = n$ for~\Cref{cor:ind:twopcfpr}, we get the following corollary.
\begin{corollary}
\label{cor:ind:twopr:shared}
Let $\Adversary$ be a $t$-oracle adversary. Then:
$$\TD(\rho,\sigma) \leq O\left(\frac{t^{2}}{\sqrt{N}}\right),$$
where: 
$$\rho = \Tr_{\reg{E}_1\reg{E}_2}\left( \ketbra{\Adversary_t^{{\pr}_{\reg{A}\reg{E}}^{(\reg{E}_1)},{{\pr}_{\reg{A}\reg{E}}^{(\reg{E}_2)}} }}{\Adversary_t^{{\pr}_{\reg{A}\reg{E}}^{(\reg{E}_1)},{{\pr}_{\reg{A}\reg{E}}^{(\reg{E}_2)}} }}_{\reg{A}\reg{B}\reg{E}_1\reg{E}_2}  \right)$$
$$\sigma = \underset{{{\cal O}_1,{\cal O}_2 \leftarrow \haarunitaries_n}}{\expect} \ketbra{\Adversary_t^{{\cal O}_1,{\cal O}_2}}{\Adversary_t^{{\cal O}_1,{\cal O}_2}} $$
\end{corollary}

\subsection{Multiset States and Operations.}
\label{sec:multiset}
We define a generalization of relation states as multiset states below and define some general isometric operations on these states:

\begin{definition}[Multiset States]
\label{def:multisetstates}
For $0 \leq t \leq N$ and a size-$t$ multiset $S=\{a_1,\ldots,a_t\}$, define the corresponding multiset state to be the unit vector: 
$$\ket{S}_{\reg{A}} = \alpha_S \cdot \sum_{\pi \in \Symgp_t} S_{\pi} \ket{a_1,\ldots,a_t}$$
where: 
$$\alpha_S =  \sqrt{\frac{\prod_{a \in [N]}\left( \sum_{i=1}^t \delta_{a_i = a} \right)!}{t!}}$$
\end{definition}

\noindent Note that when the elements of the multiset are tuples of size $2$, it becomes a relation state. \\

\noindent Next, we define some isometric operations on these multiset states that give us a framework to work with these states. We define the three operations as:
\begin{itemize}
    \item $\mathsf{Partition}$: For any key $k\in\bit^*$, there exists a unique partition of any multiset $S$ into $S_1^{k},S_2^{k}$ with $S_1^k\uplus S_2^k = S$,\footnote{Without loss of generality, we can always assume an order between $S_1^{k},S_2^{k}$.} then we define $$\Vpart:\ket{S}\ket{k}\mapsto\ket{S_1^{k}}\ket{S_2^{k}}\ket{k}.$$
    
    \item $\mathsf{Apply}$: For any key $k\in\bit^*$, there exists an injective function $f(\cdot,\cdot)$ on the elements of any multiset $S$. Define $f(S,k):= \set{f(a,k):a\in S}$. Then we define $$\Vfunc{f}:\ket{S}\ket{k}\mapsto\ket{f(S,k)}\ket{k}.$$
    
    \item $\mathsf{Pair}$: For any key $k\in\bit^*$, and any two multisets $S_1$ and $S_2$ with $|S_1|=|S_2|$, there exists a unique relation $R^k_{S_1,S_2}$ such that ${\rm Dom}(R^k_{S_1,S_2}) = S_1$ and ${\rm Im}(R^k_{S_1,S_2}) = S_2$. Then we define $$\Vpair:\ket{S_1}\ket{S_2}\ket{k}\mapsto\ket{R^k_{S_1,S_2}}\ket{k}.$$
\end{itemize}

\noindent Note that since all the above operations are reversible on the range of these operations, they are all isometries.

\section{Pseudorandom Unitaries with Short Keys}
\label{sec:pru:shortkeys}

In this section, we present a construction of pseudorandom unitaries with keys of length $n$, which is secure against adversaries making any $\poly(n)$ many queries in the iQHROM. The construction simply involves (1) applying $U$, (2) applying a short random Pauli $X$ string, and (3) applying $U$ again.

\begin{theorem}
    \label{thm:stretch_pru}
   For $k\in\bit^{n}$, define $G_k^U := U X^k U$, where $U$ is an $n$-qubit unitary. Then $\set{G_k^U}_{k\in\bit^{n}}$ is a PRU in ${\sf iQHROM}$, where $U$ is the Haar oracle. Formally, for any $t$-query two-oracle adversary $\Adversary$, for any polynomial $t$ in $n$, we have: 
   $$\TD\left( \underset{\substack{U \leftarrow \haarunitaries_n\\ k \leftarrow \{0,1\}^{n}}}{\expect}\left[\ketbra{\Adversary_{t}^{G_k^U,U}}{\Adversary_{t}^{G_k^U,U}}\right],\ \underset{\substack{U \leftarrow \haarunitaries_n\\ V \leftarrow \haarunitaries_n}}{\expect}\left[\ketbra{\Adversary_{t}^{V,U}}{\Adversary_{t}^{V,U}}\right]\right) \leq \negl(n).$$
\end{theorem}
\begin{proof}
Without loss of generality, assume that $\Adversary$ queries the first oracle on indices $\bfa = \set{a_1,\ldots,a_{\ell}} \subseteq [t]$ and the second oracle on indices $\bfb = [t]\setminus\set{a_1,\ldots,a_{\ell}}$. Consider the following hybrids. \\

\noindent $\hybrid_1$: Output $\rho_1=\underset{\substack{U \leftarrow \haarunitaries_n\\ k \leftarrow \{0,1\}^{n}}}{\expect}\left[\ketbra{\Adversary_{t}^{G_k^U,U}}{\Adversary_{t}^{G_k^U,U}}\right]$.

\noindent $\hybrid_2$: Output $\rho_2=\underset{\substack{k \leftarrow \{0,1\}^{n}}}{\expect}\left[\Tr_{\reg{E}_1}\left( \ketbra{\Adversary_t^{{\pr}_{\reg{A}\reg{E}_1}X^k{\pr}_{\reg{A}\reg{E}_1},{{\pr}_{\reg{A}\reg{E}_1}} }}{\Adversary_t^{{\pr}_{\reg{A}\reg{E}_1}X^k{{\pr}_{\reg{A}\reg{E}_1}},{{\pr}_{\reg{A}\reg{E}_1}} }}_{\reg{A}\reg{B}\reg{E}_1}  \right)\right]$.

\noindent $\hybrid_3$: Output $\rho_3= \Tr_{\reg{E}_1\reg{E}_2}\left( \ketbra{\Adversary_t^{{\pr}_{\reg{A}\reg{E}}^{(\reg{E}_1)},{{\pr}_{\reg{A}\reg{E}}^{(\reg{E}_2)}} }}{\Adversary_t^{{\pr}_{\reg{A}\reg{E}}^{(\reg{E}_1)},{{\pr}_{\reg{A}\reg{E}}^{(\reg{E}_2)}} }}_{\reg{A}\reg{B}\reg{E}_1\reg{E}_2}  \right)$. \\
\ \\

\noindent $\hybrid_4$: Output $\rho_4= \underset{\substack{U \leftarrow \haarunitaries_n\\ V \leftarrow \haarunitaries_n}}{\expect}\left[\ketbra{\Adversary_{t}^{V,U}}{\Adversary_{t}^{V,U}}\right]$.

\begin{myclaim}
$\TD(\rho_1,\rho_2) \leq \frac{4(t+\ell)(t+\ell-1)}{N+1}$.
\end{myclaim}
\begin{proof}
Follows from~\Cref{thm:MH24}. 
\end{proof}

\begin{myclaim}
$\TD(\rho_3,\rho_4) \leq O\left(\frac{(t+\ell)^2}{\sqrt{N}}\right)$.
\end{myclaim}
\begin{proof}
Follows from~\Cref{cor:ind:twopr:shared}. 
\end{proof}

\noindent Hence, the only thing left to prove is that $\rho_2$ is close to $\rho_3$.
\begin{myclaim}
    $\TD(\rho_2,\rho_3)\leq 2\sqrt{\frac{t^2+t\ell}{N}}$.
\end{myclaim}
\begin{proof}
    We start by noticing that $\rho_2 = \Tr_{\reg{E}_1\reg{K}}\left( \ketbra{\psi_2}{\psi_2}_{\reg{A}\reg{B}\reg{E}_1\reg{K}}\right)$, where $$\ket{\psi_{2}}_{\reg{A}\reg{B}\reg{E}_1\reg{K}} =  \sum_{k\in\bit^n} \frac{1}{2^{n/2}} \ket{\Adversary_t^{{\pr}_{\reg{A}\reg{E}_1}X^k{\pr}_{\reg{A}\reg{E}_1},{{\pr}_{\reg{A}\reg{E}_1}}}}_{\reg{A}\reg{B}\reg{E}_1}\ket{k}_{\reg{K}}.$$
    Similarly, $\rho_3 = \Tr_{\reg{E}_1\reg{E}_2}\left( \ketbra{\psi_3}{\psi_3}_{\reg{A}\reg{B}\reg{E}_1\reg{E}_2}  \right)$, where 
    $$\ket{\psi_3}_{\reg{A}\reg{B}\reg{E}_1\reg{E}_2} = \ket{\Adversary_t^{{\pr}_{\reg{A}\reg{E}}^{(\reg{E}_1)},{{\pr}_{\reg{A}\reg{E}}^{(\reg{E}_2)}}}}_{\reg{A}\reg{B}\reg{E}_1\reg{E}_2}.$$

    \noindent We start by expanding $\ket{\psi_3}$ using the definition of $\pr_{\reg{A}\reg{E}}^{(\reg{E}_1)}$ and $\pr_{\reg{A}\reg{E}}^{(\reg{E}_2)}$, 
    \ifllncs
        \begin{multline*}
            \ket{\psi_3}_{\reg{A}\reg{B}\reg{E}_1\reg{E}_2} = \frac{1}{\sqrt{N^{\downarrow t}}}\sum_{\substack{(x_1,\ldots,x_t)\in [N]^t\\ (y_1,\ldots,y_t)\in[N]^t_{\dist}}}\prod_{i=1}^{t}\left( \ketbra{y_i}{x_i}_{\reg{A}} \cdot A_{\reg{A} \reg{B}}^{(i)} \right) \ket{0}_{\reg{A}\reg{B}} \\ \ket{\set{(x_i,y_i)}_{i\in\bfa}}_{\reg{E}_1}\ket{\set{(x_i,y_i)}_{i\in\bfb}}_{\reg{E}_2}.
        \end{multline*}
    \else
        $$\ket{\psi_3}_{\reg{A}\reg{B}\reg{E}_1\reg{E}_2} = \frac{1}{\sqrt{N^{\downarrow t}}}\sum_{\substack{(x_1,\ldots,x_t)\in [N]^t\\ (y_1,\ldots,y_t)\in[N]^t_{\dist}}}\prod_{i=1}^{t}\left( \ketbra{y_i}{x_i}_{\reg{A}} \cdot A_{\reg{A} \reg{B}}^{(i)} \right) \ket{0}_{\reg{A}\reg{B}} \ket{\set{(x_i,y_i)}_{i\in\bfa}}_{\reg{E}_1}\ket{\set{(x_i,y_i)}_{i\in\bfb}}_{\reg{E}_2}.$$
    \fi

    \noindent We use $\Vec{x} = (x_1,\ldots,x_t)$, $\Vec{y} = (y_1,\ldots,y_t)$, and $\ket{\phi_{\Vec{x},\Vec{y}}}_{\reg{A}\reg{B}} = \frac{1}{\sqrt{N^{\downarrow t}}}\prod_{i=1}^{t}\left( \ketbra{y_i}{x_i}_{\reg{A}} \cdot A_{\reg{A} \reg{B}}^{(i)} \right) \ket{0}_{\reg{A}\reg{B}}$. Hence,
    $$\ket{\psi_3}_{\reg{A}\reg{B}\reg{E}_1\reg{E}_2} = \sum_{\substack{\Vec{x}\in [N]^t\\ \Vec{y}\in[N]^t_{\dist}}}\ket{\phi_{\Vec{x},\Vec{y}}}_{\reg{A}\reg{B}} \ket{\set{(x_i,y_i)}_{i\in\bfa}}_{\reg{E}_1}\ket{\set{(x_i,y_i)}_{i\in\bfb}}_{\reg{E}_2}.$$

    \noindent Similarly, to expand $\ket{\psi_2}$, we notice that for any $k\in\bit^n$, ${\pr}_{\reg{A}\reg{E}_1}X^k{\pr}_{\reg{A}\reg{E}_1}$ behaves as follows:
    \ifllncs
        \begin{multline*}
            {\pr}_{\reg{A}\reg{E}_1}X^k{\pr}_{\reg{A}\reg{E}_1} \ket{x_i}_{\reg{A}}\ket{R}_{\reg{E}_1} = \\ 
            \frac{1}{\sqrt{(N-|R|)^{\downarrow 2}}}\sum_{\substack{z_i\in [N]\setminus\Image(R)\\ y_i\in [N]\setminus (\Image(R)\cup\set{z_i})}}\ket{y_i}_{\reg{A}}\ket{R\uplus\set{(x_i,z_i),(z_i\oplus k, y_i)}}_{\reg{E}_1}.
        \end{multline*}
    \else
        $${\pr}_{\reg{A}\reg{E}_1}X^k{\pr}_{\reg{A}\reg{E}_1} \ket{x_i}_{\reg{A}}\ket{R}_{\reg{E}_1} = \frac{1}{\sqrt{(N-|R|)^{\downarrow 2}}}\sum_{\substack{z_i\in [N]\setminus\Image(R)\\ y_i\in [N]\setminus (\Image(R)\cup\set{z_i})}}\ket{y_i}_{\reg{A}}\ket{R\uplus\set{(x_i,z_i),(z_i\oplus k, y_i)}}_{\reg{E}_1}.$$
    \fi
    
    \noindent Using identity functions, we can write the above as:
    \ifllncs
        \begin{multline*}
            {\pr}_{\reg{A}\reg{E}_1}X^k{\pr}_{\reg{A}\reg{E}_1} \ket{x_i}_{\reg{A}}\ket{R}_{\reg{E}_1} = \frac{1}{\sqrt{(N-|R|)^{\downarrow 2}}}\sum_{\substack{z_i,y_i\in [N]}}\\ \prod_{w\in\Image(R)}\left(\delta_{z_i\neq w}\delta_{y_i\neq w}\right)\delta_{z_i\neq y_i}
            \ket{y_i}_{\reg{A}}\ket{R\uplus\set{(x_i,z_i),(z_i\oplus k, y_i)}}_{\reg{E}_1}.
        \end{multline*}
    \else
        $${\pr}_{\reg{A}\reg{E}_1}X^k{\pr}_{\reg{A}\reg{E}_1} \ket{x_i}_{\reg{A}}\ket{R}_{\reg{E}_1} = \frac{1}{\sqrt{(N-|R|)^{\downarrow 2}}}\sum_{\substack{z_i,y_i\in [N]}}\prod_{w\in\Image(R)}\left(\delta_{z_i\neq w}\delta_{y_i\neq w}\right)\delta_{z_i\neq y_i}\ket{y_i}_{\reg{A}}\ket{R\uplus\set{(x_i,z_i),(z_i\oplus k, y_i)}}_{\reg{E}_1}.$$
    \fi

    \noindent Hence, we expand $\ket{\psi_2}$ as 

    \ifllncs
        \begin{equation*}
        \begin{split}
        & \ket{\psi_2}_{\reg{A}\reg{B}\reg{E}_1\reg{K}} = \frac{1}{\sqrt{N\cdot N^{\downarrow (t+\ell)}}}\sum_{\substack{k\in\bit^n\\ x_1,\ldots,x_t\in [N]\\ y_1,\ldots,y_t\in[N]\\ z_{a_1},\ldots,z_{a_\ell}\in [N]}}
        \prod_{i=1}^{t} \left(\prod_{j=i}^t\delta_{y_i\neq y_j} \prod_{j=1}^\ell\delta_{y_i\neq z_{a_j}}\right) \times \\
        & \prod_{i=1}^{t}\left( \ketbra{y_i}{x_i}_{\reg{A}} \cdot A_{\reg{A} \reg{B}}^{(i)} \right) \ket{0}_{\reg{A}\reg{B}} 
        \ket{\set{(x_i,z_i)}_{i\in\bfa}\uplus \set{(z_i\oplus k,y_i)}_{i\in\bfa}\uplus\set{(x_i,y_i)}_{i\in\bfb}}_{\reg{E}_1}\ket{k}_{\reg{K}}.
        \end{split}
        \end{equation*}
    \else
        \begin{multline*}
        \ket{\psi_2}_{\reg{A}\reg{B}\reg{E}_1\reg{K}} = \frac{1}{\sqrt{N\cdot N^{\downarrow (t+\ell)}}}\sum_{\substack{k\in\bit^n\\ x_1,\ldots,x_t\in [N]\\ y_1,\ldots,y_t\in[N]\\ z_{a_1},\ldots,z_{a_\ell}\in [N]}}\prod_{i=1}^{t}\left(\prod_{j=i}^t\delta_{y_i\neq y_j}\prod_{j=1}^\ell\delta_{y_i\neq z_{a_j}}\right)\prod_{i=1}^{t}\left( \ketbra{y_i}{x_i}_{\reg{A}} \cdot A_{\reg{A} \reg{B}}^{(i)} \right) \ket{0}_{\reg{A}\reg{B}} \\
        \otimes \ket{\set{(x_i,z_i)}_{i\in\bfa}\uplus \set{(z_i\oplus k,y_i)}_{i\in\bfa}\uplus\set{(x_i,y_i)}_{i\in\bfb}}_{\reg{E}_1}\ket{k}_{\reg{K}}.
        \end{multline*}
    \fi

    \noindent We use $\Vec{x} = (x_1,\ldots,x_t)$, $\Vec{y} = (y_1,\ldots,y_t)$ and $\Vec{z} = (z_{a_1},\ldots,z_{a_\ell})$ and expand $\ket{\psi_2}$ as 
    \begin{multline*}
    \ket{\psi_2}_{\reg{A}\reg{B}\reg{E}_1\reg{K}} = \\
    \sum_{\substack{\Vec{x}\in [N]^t
    \Vec{y}\in[N]^t_\dist}} \ket{\phi_{\Vec{x},\Vec{y}}}_{\reg{A}\reg{B}} \otimes  \frac{1}{\sqrt{N\cdot (N-t)^{\downarrow \ell}}} 
    \sum_{\substack{k\in\bit^n\\ \Vec{z}\in ([N]\setminus\set{\Vec{y}})^\ell_{\dist}}}\ket{\set{(x_i,z_i)}_{i\in\bfa}\uplus \set{(z_i\oplus k,y_i)}_{i\in\bfa}\uplus\set{(x_i,y_i)}_{i\in\bfb}}_{\reg{E}_1}\ket{k}_{\reg{K}}.
    \end{multline*}

    \noindent We prove $\rho_2$ is close to $\rho_3$ by showing that there exists a projector $\Pi^{\mathrm{good}}_{\reg{E}_1\reg{K}}$, and two isometries $V^2_{\reg{E}_1\reg{K}}$ and $V^{3}_{\reg{E}_1\reg{E}_2}$ such that $V^2_{\reg{E}_1\reg{K}}\Pi^{\mathrm{good}}_{\reg{E}_1\reg{K}}\ket{\psi_{2}}_{\reg{A}\reg{B}\reg{E}_1\reg{K}}$ is close to $V^{3}_{\reg{E}_1\reg{E}_2}\ket{\psi_3}_{\reg{A}\reg{B}\reg{E}_1\reg{E}_2}$. To define $\Pi^{\mathrm{good}}_{\reg{E}_1\reg{K}}$, we start by defining $\CorX(R,k)$ as a set of pairs in $R$ that are correlated by $k$. Formally,
    \begin{equation*}
        \CorX(R,k) := \set{((u,v),(u',v'))\in R \times R: v\oplus u'=k}.
    \end{equation*} 
    Next, we define the set of good keys for a given relation $R$ as 
    $$\good(R,\ell) := \set{k:|\CorX(R,k)|=\ell}.$$ Hence, we define $\Pi^{\mathrm{good}}_{\reg{E}_1\reg{K}}$ as $$\Pi^{\mathrm{good}}_{\reg{E}_1\reg{K}} := \sum_{R}\left(\ketbra{R}{R}_{\reg{E}_1}\otimes\sum_{k\in\good(R,\ell)}\ketbra{k}{k}_{\reg{K}}\right).$$

    \noindent Next, we study the effect of $\Pi^{\mathrm{good}}_{\reg{E}_1\reg{K}}$ on relations of the form $\set{(x_i,z_i)}_{i\in\bfa}\uplus \set{(z_i\oplus k,y_i)}_{i\in\bfa}\uplus\set{(x_i,y_i)}_{i\in\bfb}$.
    
    \noindent We use the following notation: $\left(\set{\Vec{x}}\oplus\set{\Vec{y}}\right)\cup\left(\set{\Vec{x}}\oplus\set{\Vec{z}}\right) = \left(\set{x_i\oplus y_j}_{\substack{i\in [t]\\ j\in [t]}}\cup \set{x_i\oplus z_{j}}_{\substack{i\in [t]\\ j\in \bfa}}\right).$

    \begin{myclaim}     \label{cor:proj:key} 
        Let $\Vec{x}\in [N]^t$, $\Vec{y}\in [N]^t_{\dist}$ and $\Vec{z}\in ([N]\setminus\set{\Vec{y}})^\ell_{\dist}$, then 
        \begin{multline*}
            \Pi^{\mathrm{good}}_{\reg{E}_1\reg{K}} \sum_{\substack{k\in\bit^n}}\ket{\set{(x_i,z_i)}_{i\in\bfa}\uplus \set{(z_i\oplus k,y_i)}_{i\in\bfa}\uplus\set{(x_i,y_i)}_{i\in\bfb}}_{\reg{E}_1}\ket{k}_{\reg{K}} = \\
            \sum_{\substack{k\in\bit^n\setminus \left(\set{\Vec{x}}\oplus\set{\Vec{y}}\right)\cup\left(\set{\Vec{x}}\oplus\set{\Vec{z}}\right)}}\ket{\set{(x_i,z_i)}_{i\in\bfa}\uplus \set{(z_i\oplus k,y_i)}_{i\in\bfa}\uplus\set{(x_i,y_i)}_{i\in\bfb}}_{\reg{E}_1}\ket{k}_{\reg{K}}.
        \end{multline*}
    \end{myclaim}
    
    \noindent The proof of~\Cref{cor:proj:key} is deferred to~\Cref{app:pru:shortkeys}. We define $\ket{\wt{\psi}_2} := \Pi^{\mathrm{good}}_{\reg{E}_1\reg{K}}\ket{\psi_2}$. Then by~\Cref{cor:proj:key}, 
    \begin{multline*}
        \ket{\wt{\psi}_2} = \sum_{\substack{\Vec{x}\in [N]^t\\ \Vec{y}\in[N]^t_\dist}} \ket{\phi_{\Vec{x},\Vec{y}}}_{\reg{A}\reg{B}} \otimes  \frac{1}{\sqrt{N\cdot (N-t)^{\downarrow \ell}}}\sum_{\substack{\Vec{z}\in ([N]\setminus\set{\Vec{y}})^\ell_{\dist}}}\\ \sum_{k\in [N]\setminus\left(\set{\Vec{x}}\oplus\set{\Vec{y}}\right)\cup\left(\set{\Vec{x}}\oplus\set{\Vec{z}}\right)}\ket{\set{(x_i,z_i)}_{i\in\bfa}\uplus \set{(z_i\oplus k,y_i)}_{i\in\bfa}\uplus\set{(x_i,y_i)}_{i\in\bfb}}_{\reg{E}_1}\ket{k}_{\reg{K}}.
    \end{multline*}

    \noindent Next we define the isometry $V^{2}_{\reg{E_1}\reg{K}}$ on $\ket{\wt{\psi}_2}$ as the following procedure:
    \begin{enumerate}
        \item We define the following partition (indexed by a key $k$) of any relation $R = \set{(u_i,v_i)}_{i}$ as $R_2^k = \set{(u,v): \exists (u',v')\in R, v\oplus u' = k}$, and $R_1^k = R\setminus R_2^k$. Then applying $\Vpart_{\reg{E}_1\reg{K}}$ isometry for the above partition $\ket{\wt{\psi}_2}$ gives us 
    \begin{multline*}
    \sum_{\substack{\Vec{x}\in [N]^t\\ \Vec{y}\in[N]^t_\dist}} \ket{\phi_{\Vec{x},\Vec{y}}}_{\reg{A}\reg{B}} \otimes  \frac{1}{\sqrt{N\cdot (N-t)^{\downarrow \ell}}}\sum_{\substack{\Vec{z}\in ([N]\setminus\set{\Vec{y}})^\ell_{\dist}}}\\ \sum_{k\in [N]\setminus\left(\set{\Vec{x}}\oplus\set{\Vec{y}}\right)\cup\left(\set{\Vec{x}}\oplus\set{\Vec{z}}\right)}\ket{ \set{(z_i\oplus k,y_i)}_{i\in\bfa}\uplus\set{(x_i,y_i)}_{i\in\bfb}}_{\reg{E}_1}\ket{\set{(x_i,z_i)}_{i\in\bfa}}_{\reg{E}_2}\ket{k}_{\reg{K}}.
    \end{multline*}
    
    \item We define the following partition (indexed by a key $k$ and a relation $R'$) of any relation $R = \set{(u_i,v_i)}_{i}$ as $R_2^k = \set{(u,v): \exists (u',v')\in R', u\oplus v' = k}$, and $R_1^k = R\setminus R_2^k$. Then applying $\Vpart_{\reg{E}_1\reg{E}_2\reg{K}}$ isometry gives us 
    \begin{multline*}
    \sum_{\substack{\Vec{x}\in [N]^t\\ \Vec{y}\in[N]^t_\dist}} \ket{\phi_{\Vec{x},\Vec{y}}}_{\reg{A}\reg{B}} \otimes  \frac{1}{\sqrt{N\cdot (N-t)^{\downarrow \ell}}}\sum_{\substack{\Vec{z}\in ([N]\setminus\set{\Vec{y}})^\ell_{\dist}}}\\ \sum_{k\in [N]\setminus\left(\set{\Vec{x}}\oplus\set{\Vec{y}}\right)\cup\left(\set{\Vec{x}}\oplus\set{\Vec{z}}\right)}\ket{ \set{(x_i,y_i)}_{i\in\bfb}}_{\reg{E}_1}\ket{\set{(x_i,z_i)}_{i\in\bfa}}_{\reg{E}_2}\ket{\set{(z_i\oplus k,y_i)}_{i\in\bfa}}_{\reg{E}_3}\ket{k}_{\reg{K}}.
    \end{multline*}
    
    \item We define the following pairing (indexed by a key $k$) of two relations $R_1 = \set{(u^1_i,v^1_i)}_{i}$ and  $R_2 = \set{(u^2_i,v^2_i)}_{i}$ as $R_{R_1,R_2}^k = \set{(u^1_i,v^1_i,u^2_{i'},v^2_{i'}):  (u^1_i,v^1_i)\in R_1, (u^2_{i'},v^2_{i'})\in R_2, v^1_i\oplus u^2_{i'} = k}$. Then applying $\Vpair_{\reg{E}_2\reg{E}_3\reg{K}}$ isometry gives us 
    \begin{multline*}
    \sum_{\substack{\Vec{x}\in [N]^t\\ \Vec{y}\in[N]^t_\dist}} \ket{\phi_{\Vec{x},\Vec{y}}}_{\reg{A}\reg{B}} \otimes  \frac{1}{\sqrt{N\cdot (N-t)^{\downarrow \ell}}}\sum_{\substack{\Vec{z}\in ([N]\setminus\set{\Vec{y}})^\ell_{\dist}}}\\ \sum_{k\in [N]\setminus\left(\set{\Vec{x}}\oplus\set{\Vec{y}}\right)\cup\left(\set{\Vec{x}}\oplus\set{\Vec{z}}\right)}\ket{ \set{(x_i,y_i)}_{i\in\bfb}}_{\reg{E}_1}\ket{\set{(x_i,z_i,z_i\oplus k,y_i)}_{i\in\bfa}}_{\reg{E}_2}\ket{k}_{\reg{K}}.
    \end{multline*}
    
    \item We define the following injection (indexed by a key $k$) $f_k:(x_i,z_i,z_i\oplus k,y_i)\mapsto (x_i,z_i,y_i)$. Then applying $\Vfunc{f_k}_{\reg{E}_2\reg{K}}$ isometry gives us 
    \begin{multline*}
    \sum_{\substack{\Vec{x}\in [N]^t\\ \Vec{y}\in[N]^t_\dist}} \ket{\phi_{\Vec{x},\Vec{y}}}_{\reg{A}\reg{B}} \otimes  \frac{1}{\sqrt{N\cdot (N-t)^{\downarrow \ell}}}\sum_{\substack{\Vec{z}\in ([N]\setminus\set{\Vec{y}})^\ell_{\dist}}}\\ \ket{ \set{(x_i,y_i)}_{i\in\bfb}}_{\reg{E}_1}\ket{\set{(x_i,z_i,y_i)}_{i\in\bfa}}_{\reg{E}_2}
    \otimes 
    \sum_{k\in [N]\setminus\left(\set{\Vec{x}}\oplus\set{\Vec{y}}\right)\cup\left(\set{\Vec{x}}\oplus\set{\Vec{z}}\right)}\ket{k}_{\reg{K}}.
    \end{multline*}
    \end{enumerate}

    \noindent Now, we further expand the register $\reg{E}_2$ according to the definition of multiset states: 
    \begin{multline*}
        \sum_{\substack{\Vec{x}\in [N]^t\\ \Vec{y}\in[N]^t_\dist}} \ket{\phi_{\Vec{x},\Vec{y}}}_{\reg{A}\reg{B}} \otimes  \frac{1}{\sqrt{N\cdot (N-t)^{\downarrow \ell}}}\sum_{\substack{\Vec{z}\in ([N]\setminus\set{\Vec{y}})^\ell_{\dist}}}\ket{ \set{(x_i,y_i)}_{i\in\bfb}}_{\reg{E}_1}\otimes\\ 
        \left(\frac{1}{\sqrt{\ell!}}\sum_{\pi\in\Symgp_{\ell}}S_{\pi}\ket{x_{a_1},\ldots,x_{a_\ell}} \otimes S_{\pi}\ket{z_{a_1},\ldots,z_{a_\ell}} \otimes S_{\pi}\ket{y_{a_1},\ldots,y_{a_\ell}}\right)_{\reg{E}_2}\otimes
        \sum_{k\in [N]\setminus\left(\set{\Vec{x}}\oplus\set{\Vec{y}}\right)\cup\left(\set{\Vec{x}}\oplus\set{\Vec{z}}\right)}\ket{k}_{\reg{K}}.
    \end{multline*}

    \noindent Notice that we can sample $\set{\Vec{z}}$ first and then assign an order\footnote{Consider $\set{\Vec{z}}$ to have a canonical order and we look at all possible permutations of this order.}. Hence, the state looks like 
    \begin{multline*}
        \sum_{\substack{\Vec{x}\in [N]^t\\ \Vec{y}\in[N]^t_\dist}} \ket{\phi_{\Vec{x},\Vec{y}}}_{\reg{A}\reg{B}} \otimes  \frac{1}{\sqrt{N\cdot (N-t)^{\downarrow \ell}}}\sum_{\set{\Vec{z}}\in {[N]\setminus\set{\Vec{y}} \choose \ell}}\sum_{\sigma \in \Symgp_\ell}\ket{ \set{(x_i,y_i)}_{i\in\bfb}}_{\reg{E}_1} \otimes \\ \left(\frac{1}{\sqrt{\ell!}}\sum_{\pi\in\Symgp_{\ell}} S_{\pi}\ket{x_{a_1},\ldots,x_{a_\ell}} \otimes S_{\pi}S_{\sigma}\ket{z_{a_1},\ldots,z_{a_\ell}} \otimes S_{\pi}\ket{y_{a_1},\ldots,y_{a_\ell}}\right)_{\reg{E}_2} \otimes 
        \sum_{k\in [N]\setminus\left(\set{\Vec{x}}\oplus\set{\Vec{y}}\right)\cup\left(\set{\Vec{x}}\oplus\set{\Vec{z}}\right)}\ket{k}_{\reg{K}}.
    \end{multline*}
    By switching the order of $\sigma$ and $\pi$ and setting $\sigma\gets\pi^{-1}\sigma$, we can disentangle $\Vec{z}$, 
    \begin{multline*}
        \sum_{\substack{\Vec{x}\in [N]^t\\ \Vec{y}\in[N]^t_\dist}} \ket{\phi_{\Vec{x},\Vec{y}}}_{\reg{A}\reg{B}} \otimes  \frac{1}{\sqrt{N\cdot (N-t)^{\downarrow \ell}}}\sum_{\set{\Vec{z}}\in {[N]\setminus\set{\Vec{y}} \choose \ell}}\ket{ \set{(x_i,y_i)}_{i\in\bfb}}_{\reg{E}_1}\otimes\\
        \left(\frac{1}{\sqrt{\ell!}}\sum_{\pi\in\Symgp_{\ell}}S_{\pi}\ket{x_{a_1},\ldots,x_{a_\ell}}\otimes S_{\pi}\ket{y_{a_1},\ldots,y_{a_\ell}}\right)_{\reg{E}_2}\otimes\left(\sum_{\sigma \in \Symgp_\ell}S_{\sigma}\ket{z_{a_1},\ldots,z_{a_\ell}}\right)_{\reg{E}_3}\otimes \\
        \sum_{k\in [N]\setminus\left(\set{\Vec{x}}\oplus\set{\Vec{y}}\right)\cup\left(\set{\Vec{x}}\oplus\set{\Vec{z}}\right)}\ket{k}_{\reg{K}}.
    \end{multline*}
    Using the multiset state notation:
    \begin{multline*}
        \sum_{\substack{\Vec{x}\in [N]^t\\ \Vec{y}\in[N]^t_\dist}} \ket{\phi_{\Vec{x},\Vec{y}}}_{\reg{A}\reg{B}} \otimes  \frac{\sqrt{\ell!}}{\sqrt{N\cdot (N-t)^{\downarrow \ell}}}\sum_{\set{\Vec{z}}\in {[N]\setminus\set{\Vec{y}} \choose \ell}}\ket{ \set{(x_i,y_i)}_{i\in\bfb}}_{\reg{E}_1} \otimes \\ \ket{\set{(x_i,y_i)}_{i\in\bfa}}_{\reg{E}_2}\ket{\set{z_i}_{i\in\bfa}}_{\reg{E}_3}\sum_{k\in [N]\setminus\left(\set{\Vec{x}}\oplus\set{\Vec{y}}\right)\cup\left(\set{\Vec{x}}\oplus\set{\Vec{z}}\right)}\ket{k}_{\reg{K}}.
    \end{multline*}

    \noindent Re-arranging and setting $\ket{\Psi_2} := V^2_{\reg{E_1}\reg{K}}\Pi^{\mathrm{good}}_{\reg{E}_1\reg{K}}\ket{\psi_2}_{\reg{A}\reg{B}\reg{E}_1\reg{K}}$, then 
    \begin{multline*}
    \ket{\Psi_2}  = \sum_{\substack{\Vec{x}\in [N]^t\\ \Vec{y}\in[N]^t_\dist}} \ket{\phi_{\Vec{x},\Vec{y}}}_{\reg{A}\reg{B}} \otimes \ket{ \set{(x_i,y_i)}_{i\in\bfa}}_{\reg{E}_1}\ket{\set{(x_i,y_i)}_{i\in\bfb}}_{\reg{E}_2}\otimes\\
    \sqrt{\frac{\ell!}{(N-t)^{\downarrow \ell}}}\sum_{\set{\Vec{z}}\in {[N]\setminus\set{\Vec{y}} \choose \ell}}\ket{\set{z_i}_{i\in\bfa}}_{\reg{E}_3}\frac{1}{\sqrt{N}} \otimes \sum_{k\in [N]\setminus\left(\set{\Vec{x}}\oplus\set{\Vec{y}}\right)\cup\left(\set{\Vec{x}}\oplus\set{\Vec{z}}\right)}\ket{k}_{\reg{K}}.
    \end{multline*}

    \noindent Recall that $\ket{\psi_3}$ is defined as 
    $$\ket{\psi_3}_{\reg{A}\reg{B}\reg{E}_1\reg{E}_2} = \sum_{\substack{\Vec{x}\in [N]^t\\ \Vec{y}\in[N]^t_{\dist}}}\ket{\phi_{\Vec{x},\Vec{y}}}_{\reg{A}\reg{B}} \ket{\set{(x_i,y_i)}_{i\in\bfa}}_{\reg{E}_1}\ket{\set{(x_i,y_i)}_{i\in\bfb}}_{\reg{E}_2}.$$

    \noindent Next we define the isometry $V^{3}_{\reg{E}_1\reg{E}_2}$ on $\ket{\psi_3}$ as the following procedure:
    \begin{enumerate}
        \item Controlled on the $y_i$'s in registers $\reg{E}_1\reg{E}_2$, we create a superposition over all $\ell$-sized sets disjoint from $y_i$'s to get 
        \begin{multline*}
        \sum_{\substack{\Vec{x}\in [N]^t\\ \Vec{y}\in[N]^t_\dist}} \ket{\phi_{\Vec{x},\Vec{y}}}_{\reg{A}\reg{B}} \otimes \ket{ \set{(x_i,y_i)}_{i\in\bfa}}_{\reg{E}_1}\ket{\set{(x_i,y_i)}_{i\in\bfb}}_{\reg{E}_2} \otimes \sqrt{\frac{\ell!}{(N-t)^{\downarrow \ell}}}\sum_{\set{\Vec{z}}\in {[N]\setminus\set{\Vec{y}} \choose \ell}}\ket{\set{z_i}_{i\in\bfa}}_{\reg{E}_3}.
        \end{multline*}

        \item Finally, controlled on $\reg{E}_1\reg{E}_2\reg{E}_3$, we create a superposition over all keys not in $\left(\set{\Vec{x}}\oplus\set{\Vec{y}}\right)\cup\left(\set{\Vec{x}}\oplus\set{\Vec{z}}\right)$, 
        \begin{multline*}
        \sum_{\substack{\Vec{x}\in [N]^t\\ \Vec{y}\in[N]^t_\dist}} \ket{\phi_{\Vec{x},\Vec{y}}}_{\reg{A}\reg{B}} \otimes \ket{ \set{(x_i,y_i)}_{i\in\bfa}}_{\reg{E}_1}\ket{\set{(x_i,y_i)}_{i\in\bfb}}_{\reg{E}_2}\otimes 
        \sqrt{\frac{\ell!}{(N-t)^{\downarrow \ell}}}\sum_{\set{\Vec{z}}\in {[N]\setminus\set{\Vec{y}} \choose \ell}}\ket{\set{z_i}_{i\in\bfa}}_{\reg{E}_3} \\
        \otimes
        \frac{1}{\sqrt{N-|\left(\set{\Vec{x}}\oplus\set{\Vec{y}}\right)\cup\left(\set{\Vec{x}}\oplus\set{\Vec{z}}\right)|}}\sum_{k\in[N]\setminus\left(\set{\Vec{x}}\oplus\set{\Vec{y}}\right)\cup\left(\set{\Vec{x}}\oplus\set{\Vec{z}}\right)}\ket{k}_{\reg{K}}.
        \end{multline*}
    \end{enumerate}
    Hence, let $\ket{\Psi_3} := V^3_{\reg{E}_1\reg{E}_2}\ket{\psi_3}_{\reg{A}\reg{B}\reg{E}_1\reg{E}_2}$. Then 
    \begin{multline*}
    \ket{\Psi_3} = \sum_{\substack{\Vec{x}\in [N]^t\\ \Vec{y}\in[N]^t_\dist}} \ket{\phi_{\Vec{x},\Vec{y}}}_{\reg{A}\reg{B}} \otimes \ket{ \set{(x_i,y_i)}_{i\in\bfa}}_{\reg{E}_1}\ket{\set{(x_i,y_i)}_{i\in\bfb}}_{\reg{E}_2}\otimes
    \sqrt{\frac{\ell!}{(N-t)^{\downarrow \ell}}}\sum_{\set{\Vec{z}}\in {[N]\setminus\set{\Vec{y}} \choose \ell}}\ket{\set{z_i}_{i\in\bfa}}_{\reg{E}_3} \\
    \otimes
    \frac{1}{\sqrt{N-|\left(\set{\Vec{x}}\oplus\set{\Vec{y}}\right)\cup\left(\set{\Vec{x}}\oplus\set{\Vec{z}}\right)|}}\sum_{k\in[N]\setminus\left(\set{\Vec{x}}\oplus\set{\Vec{y}}\right)\cup\left(\set{\Vec{x}}\oplus\set{\Vec{z}}\right)}\ket{k}_{\reg{K}}.
    \end{multline*}

    \noindent Note that by a simple counting argument, for all $(\Vec{x},\Vec{y},\Vec{z})$, 
    \begin{equation}
    \label{eq:pru:small_key}
        \sqrt{\frac{N-|\left(\set{\Vec{x}}\oplus\set{\Vec{y}}\right)\cup\left(\set{\Vec{x}}\oplus\set{\Vec{z}}\right)|}{N}}\geq\sqrt{\frac{N-(t^2+t\ell)}{N}}.
    \end{equation} 
    Applying~\Cref{lem:norm:sub-state}, $$\braket{\Psi_2}{\Psi_3}\geq\sqrt{\frac{N-(t^2+t\ell)}{N}}.$$ Hence, by~\Cref{lem:norm:cs}, we get $$\|\ket{\wt{\psi}_2}\|^2 = \|\ket{\Psi_2}\|^2\geq 1-\frac{t^2+t\ell}{N}.$$
    Finally, by~\Cref{lem:norm:proj}, we get $$\TD(\ket{\psi_2},\ket{\wt{\psi}_2})\leq \sqrt{\frac{t^2+t\ell}{N}}.$$ 
    Hence, combining the above bounds, $\TD(\rho_2,\rho_3)\leq 2\sqrt{\frac{t^2+t\ell}{N}}$. 
\end{proof}
\noindent Combining the above three claims, we complete the proof of~\Cref{thm:stretch_pru}.
\end{proof}
\ \newline
\noindent In the above proof, we assume the size of the Pauli $X$ to be equal to the size of $U$, but we can generalize the above proof to hold as long as the size of the Pauli $X$ is $\omega(\log{\secp})$.
\begin{corollary}
\label{cor:stretch_pru}
   For any $f(\secp) = \omega(\log\secp)$, $k\in\bit^{f(\secp)}$, define $G_k^U = U (X^k\otimes I_{n-f(\secp)}) U$, where $U$ is an $n$-qubit unitary. Then $\set{G_k^U}_{k\in\bit^{f(\secp)}}$ is a PRU in ${\sf iQHROM}$, where $U$ is the Haar oracle. Formally, for any $t$-query two-oracle adversary $\Adversary$, for any polynomial $t$ in $\secp$, we have: 
   $$\TD\left( \underset{\substack{U \leftarrow \haarunitaries_n\\ k \leftarrow \{0,1\}^{f(\secp)}}}{\expect}\left[\ketbra{\Adversary_{t}^{G_k^U,U}}{\Adversary_{t}^{G_k^U,U}}\right],\ \underset{\substack{U \leftarrow \haarunitaries_n\\ V \leftarrow \haarunitaries_n}}{\expect}\left[\ketbra{\Adversary_{t}^{V,U}}{\Adversary_{t}^{V,U}}\right]\right) \leq \negl(\secp).$$
\end{corollary}

\begin{proof}[Proof sketch]
    The proof above goes exactly the same except that the bound in~\Cref{eq:pru:small_key} changes to $$\sqrt{\frac{2^{f(\secp)}-|\left(\set{\Vec{x}}\oplus\set{\Vec{y}}\right)\cup\left(\set{\Vec{x}}\oplus\set{\Vec{z}}\right)|}{2^{f(\secp)}}}\geq\sqrt{\frac{2^{f(\secp)}-(t^2+t\ell)}{2^{f(\secp)}}},$$
    since $f(\secp) = \omega(\log\secp)$, the above is still lower-bounded by $1-\negl(\secp)$. Hence, the construction is secure.
\end{proof}

\section{Key-Stretched PRU in the Plain Model}
\label{sec:stretch_pru}

In this section, applying our result from~\Cref{sec:pru:shortkeys} to a pseudorandom unitary in the plain model, we show that we can stretch the output length of \emph{any} PRU, relative to its key size.

As an immediate consequence of~\Cref{cor:stretch_pru}, we can actually get arbitrary polynomial-stretch PRU in the \emph{plain} model. At a high level, the idea is to sample a single PRU key, and use this to computationally instantiate a Haar random oracle. Then use the construction of PRUs with short keys in the Haar oracle model to get more PRUs while only using $O(\log^{1+\epsilon} \secp)$ more bits of randomness (for any constant $\epsilon$, although we will set this to be $\log^2 \secp$ in the remainder of this section).

\par Plugging these seemingly fresh pseudorandom unitaries into the construction of \cite{schuster2024random}, we can stretch the output size of the pseudorandom unitary. For a graphical depiction of the construction, see~\Cref{fig:stretch_pru}.

\begin{figure}[H]
    \resizebox{\textwidth}{!}{
        \input{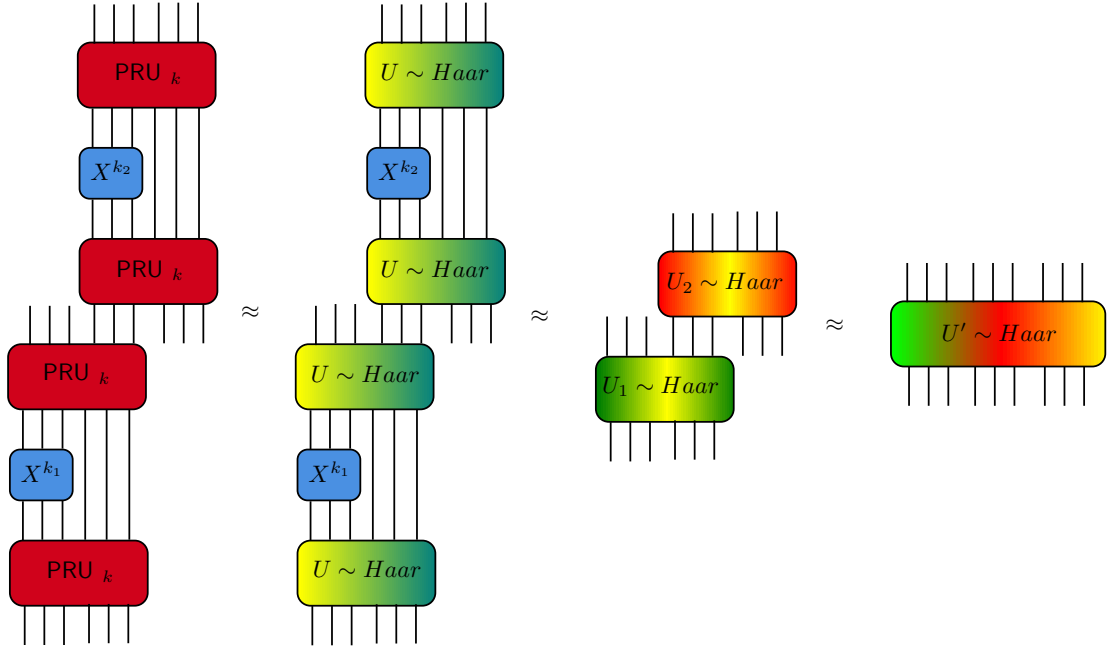}
    }
    \caption{Implementation of key-stretched PRU from any PRU.  Going from left to right, the first approximation uses the definition of the PRU, the next one uses \Cref{thm:intro:prus:iqhrom}, and the final one uses the result from \cite{schuster2024random}.}
    \label{fig:stretch_pru}
\end{figure}

\noindent We recall the main lemma of \cite{schuster2024random}.
\begin{lemma}[Gluing two random unitaries~\cite{schuster2024random}]
    \label{lem:gluing_random_unitaries}
    Let $\reg{A}, \reg{B}, \reg{C}$ be three disjoint subsystems.  Consider a random unitary given by $V_{\reg{ABC}} = U_{\reg{AB}}U'_{\reg{BC}}$, where $U$ and $U'$ are drawn from $\epsilon_1$ and $\epsilon_2$-approximate $k$-designs, respectively.
    Then $V_{\reg{ABC}}$ is an $\epsilon$-approximate $k$ design, with
    \begin{equation*}
        1 + \epsilon \leq (1+ \epsilon_1)(1 + \epsilon_2)\left(1 + \frac{5k^2}{2^{|\reg{B}|}}\right)\,.
    \end{equation*}
    As long as the number of qubits in $\reg{B}$ satisfies $|\reg{B}| \geq \log(5k^2)$.
\end{lemma}

\begin{theorem}[Stretching a PRU]
\label{thm:plain_model:stretching_pru}
    Let $\{\mathsf{PRU}_{\lambda, k}\}_{\lambda \in \mathbb{N}, k \in \{0, 1\}^{\lambda}}$ be a PRU family with keys of size $\lambda$, where $U_{k}$ acts on $t(\lambda)$ many qubits.
    Then there exists a pseudorandom unitary family $\{\mathsf{SPRU}_{k}\}_{\lambda \in \mathbb{N}, k \in \{0, 1\}^{\lambda + 2\log^2(\lambda)}}$ ($\mathsf{S}$ for stretched) with keys of size $\lambda + 2\log^2(\lambda)$ that acts on $2t(\lambda) - \log^2(\lambda)$ qubits.
\end{theorem}

\begin{proof}
    Let $\{\mathsf{PRU}_{k}\}_{k}$ be a family of pseudorandom unitaries for a fixed security parameter $\secp$.  Let $\mathcal{A}^{(\cdot)}$ be a quantum polynomial time adversary that makes queries to an oracle and outputs either $\top$ or $\bot$.   We first define $\mathsf{SPRU}_{k, k_1, k_2}$ as follows.  Let $\reg{ABC}$ be three quantum registers, with $\reg{A}, \reg{C}$ being $t(\secp) - \log^2(\secp)$ and $\reg{B}$ being $\log^2(\secp)$.  
    Then for a key $k$ of size $\secp$ and keys $k_1, k_2$ of size $\log^2(\secp)$, 
    \begin{equation*}
        \mathsf{SPRU}_{k, k_1, k_2} = (\mathsf{PRU}_k)_{\reg{BC}} X^{k_2}_{\reg{BC}}  (\mathsf{PRU}_{k})_{\reg{BC}}(\mathsf{PRU}_k)_{\reg{AB}} X^{k_1}_{\reg{AB}}  (\mathsf{PRU}_{k})_{\reg{AB}}\,.
    \end{equation*}
    We claim the following holds
    \begin{equation*}
        \Bigg| \Pr_{\substack{k \random \{0, 1\}^{\secp} \\ k_1, k_2 \random \{0, 1\}^{\log^2(\secp)}}}\left[\top \leftarrow \Adversary^{\mathsf{SPRU}_{k, k_1, k_2}} \right] - 
        \Pr_{\substack{U \gets \mu_{t(\secp)} \\ k_1, k_2 \random \{0, 1\}^{\log^2(\secp)}}} \left[\top \leftarrow \Adversary^{U_{\reg{BC}} X^{k_2}_{\reg{BC}} U_{\reg{BC}} U_{\reg{AB}} X^{k_1}_{\reg{AB}} U_{\reg{AB}}}\right] \Bigg| 
        \leq \negl(\secp)\,.
    \end{equation*}
    This holds by the pseudorandomness of the original PRU.  If, for the sake of contradiction, there exists an adversary that has a non-negligible advantage above, then the adversary can be turned into an adversary for $\mathsf{PRU}_k$ and $U$ just by simulating $\mathsf{SPRU}$ by sampling the additional keys. \\
    
    \noindent Then, we have the following
    \begin{equation*}
        \Bigg| \Pr_{\substack{U \gets \mu_{t(\secp)} \\ k_1, k_2 \random \{0, 1\}^{\log^2(\secp)}}} \left[\top \leftarrow \Adversary^{U_{\reg{BC}} X^{k_2}_{\reg{BC}} U_{\reg{BC}} U_{\reg{AB}} X^{k_1}_{\reg{AB}} U_{\reg{AB}}}\right] - 
        \Pr_{\substack{U, V \gets \mu_{t(\secp)} \\ k_1 \random \{0, 1\}^{\log^2(\secp)}}} \left[\top \leftarrow \Adversary^{V_{\reg{BC}} U_{\reg{AB}} X^{k_1}_{\reg{AB}} U_{\reg{AB}}}\right] \Bigg| 
        \leq \negl(\secp)\,.
    \end{equation*}
    This holds because of the construction of~\Cref{cor:stretch_pru}.  In particular, we proved that $U X^{k} U$ is indistinguishable from an independently sampled Haar random unitary $V$, even to adversaries who also get query access to $U$.  \\
    
    \noindent Again by~\Cref{cor:stretch_pru}, $UX^{k_1}U$ is indistinguishable from a Haar random unitary, so the following holds
    \begin{equation*}    
        \Bigg|\Pr_{\substack{U, V  \gets \mu_{t(\secp)} \\ k_1 \random \{0, 1\}^{\log^2(\secp)}}} \left[\top \leftarrow \mathcal{A}^{V_{\reg{BC}} U_{\reg{AB}} X^{k_1}_{\reg{AB}} U_{\reg{AB}}}\right] -  \Pr_{U, V  \gets \mu_{t(\secp)}} \left[\top \leftarrow \mathcal{A}^{V_{\reg{BC}} U_{\reg{AB}}}\right]\Bigg|\leq \negl(\secp)\,.
    \end{equation*}    
    Then we apply \Cref{lem:gluing_random_unitaries} with $k = 2^{\log^{1.5}(\secp)}$, and $|B| = \log^2(\secp)$. Since Haar random unitaries are exact $k$-designs for all $k$, $\epsilon_1 = \epsilon_2 = 0$, and we get that $V_{\reg{BC}} U_{\reg{AB}}$ is an $\eps$-approximate $k$-design, where
    \begin{equation*}
        \eps(\secp) = \frac{5 \cdot 2^{2\log^{1.5}(\secp)}}{2^{\log^2(\secp)}} = \negl(\lambda)\,.
    \end{equation*}
    
    \noindent Finally, it is known that $\delta$-approximate $q$-designs are PRUs if $\delta = \negl(\secp)$ and $q = \secp^{\omega(1)}$ (see \eg \cite{AMR20,Kretschmer21,MPSY24} or \cite[Lemma~5]{schuster2024random}). \\
    
    \noindent Since $k = 2^{\log^{1.5}(\secp)} = \secp^{\omega(1)}$, for any adversary that makes only a polynomial number of queries, the following holds
    \begin{equation*}
        \left| \Pr_{U, V  \gets \mu_{t(\secp)}} \left[\top \leftarrow \mathcal{A}^{V_{\reg{BC}} U_{\reg{AB}}}\right] - \Pr_{U \gets \mu_{2t(\lambda) - \log^2(\secp)}}\left[\top \leftarrow \mathcal{A}^{U_{\reg{ABC}}}\right]\right| \leq \negl(\lambda)\,.
    \end{equation*}
    Thus, by the triangle inequality, the original construction of $\mathsf{SPRU}_{k, k_1, k_2}$ is indistinguishable from a large Haar random unitary on $2t(\secp) - \log^2(\secp)$ qubits.
\end{proof}

\noindent We can repeat this reduction $O(\log(\secp))$ many times (since the size of the pseudorandom unitary doubles every time we apply it) to get the following lemma.
\begin{corollary}[Pseudorandom unitaries with small keys from any pseudorandom unitary]
    Let $\{\mathsf{PRU}_{\lambda, k}\}_{\lambda, k}$ be a family of pseudorandom unitaries that has keys of length $\secp$ and acts on $t(\secp)$ qubits.  Then for every constant $c$, there exists a pseudorandom unitary family such that
    \begin{enumerate}
        \item The key length of the family is $\secp + 2c\log^3(\secp)$.
        \item The pseudorandom unitary acts on $\secp^c \left(t(\secp) - \log^2(\secp)\right) + \log^2(\secp)$ many qubits.
    \end{enumerate}
\end{corollary}
\begin{proof}
    Applying the reduction recursively $c\log(\secp)$ many times, we add $2\log^2(\secp)$ many bits to the key each time, and double (minus $\log^2(\secp)$) the output length of the pseudorandom unitary. Hence, after doing this $n$ times, the output length becomes 
    \begin{equation*}
    2^n t(\secp) - 2^{n-1}\log^2{\secp} - 2^{n-2}\log^2 \secp - \ldots - 2^{0}\log^2 \secp\,.
    \end{equation*}
    Hence, the final output length is 
    \begin{equation*}
    2^n \left(t(\secp) - \log^2{\secp}\right)+\log^2 \secp\,.
    \end{equation*}
    Thus, for $n = c\log(\secp)$, we get the desired key length and output length for the pseudorandom unitary.
    Additionally note that doing this requires running the original pseudorandom unitary $2^{c\log(\secp)} = \secp^c$ times, but since this is a polynomial the entire construction runs in polynomial time.
\end{proof}

Rescaling so that $\secp + 2c\log^3(\secp) = \secp'$, we see that for every choice of $c$, there is a pseudorandom unitary whose output length is roughly $\secp^c$, for keys of length $\secp$.  
We also note that our construction did not require any extra conditions on the pseudorandom unitary family, simply that any construction is a pseudorandom unitary.

Applying a brickwork, instead of staircase, layout, we can take any pseudorandom unitary family in depth $d$, and output a pseudorandom unitary family in depth $4d+2$ that has arbitrarily small keys.
\begin{corollary}[Low depth pseudorandom unitaries with short keys]
    Let $\{\mathsf{PRU}_{\secp, k}\}_{\secp, k}$ be a family of pseudorandom unitaries with keys of length $\secp$, such that every $\mathsf{PRU}_{k}$ is depth at most $d(\secp)$ and acts on $t(\secp)$ qubits.
    Then for every constant $c$, there is a pseudorandom unitary family with keys of length $O\left(\secp\right)$, output length $\secp^c t(\secp)$, and depth $4d(\secp) + 2$. 
\end{corollary}
\begin{proof}
    We apply \Cref{thm:stretch_pru} to the brickwork architecture shown in \Cref{fig:low_depth_stretch_pru}, we get a pseudorandom unitary whose depth is $4d(\secp)+2$, and at a cost of sampling $c\log^3 \secp$ additional bits of randomness, whose output length is scaled up by a factor of $\secp^c$.
\end{proof}

\begin{figure}[H]
    \centering
    \input{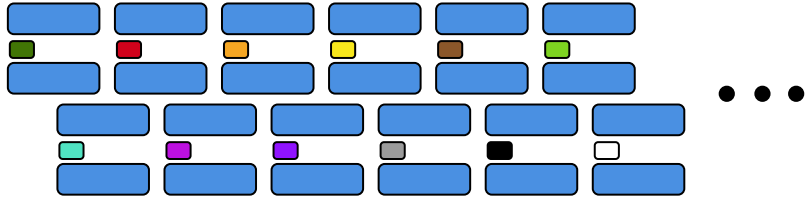}
    \caption{Implementation of low depth, short key pseudorandom unitaries from any pseudorandom unitary family.  Long blue boxes are a single sample of the original pseudorandom unitary family, and short colored boxes are additional $\omega(\log(\secp))$ sized Pauli $X$ strings.}
    \label{fig:low_depth_stretch_pru}
\end{figure}

We note that if one-way functions exist (as in \cite{schuster2024random}, which assumes the subexponential hardness of LWE), then this key shrinkage can be achieved by first shrinking the keys used in the pseudorandom function construction, however our reduction applies to \emph{all} pseudorandom unitary families, even those that do not arise from a classical pseudorandom function.
\section{Bounded-Query Pseudorandom Unitaries with Short Keys}     \label{sec:unitarydesign:shortkeys}

Now, we present an even simpler construction of $O(\secp / \log(\secp)^{1+\epsilon})$-query secure pseudorandom unitaries with keys of length $\secp$ in the iQHROM that only makes \emph{a single query} to the common Haar random unitary.

\begin{theorem}     \label{thm:ell_design}
   For $k\in\bit^{\secp}$, define $G_k^U = (Z^k\otimes I) U$, where $U$ is an $n$-qubit unitary such that $\secp \leq n$. Then $\set{G_k^U}_{k\in\bit^{\secp}}$ is an $\ell$-query secure pseudorandom unitary for $\ell=O(\secp/\log(\secp)^{1+\epsilon})$ for all $\epsilon>0$ in ${\sf iQHROM}$, where $U$ is the Haar oracle. Formally, for any $(\ell,t-\ell)$-query two-oracle adversary $\Adversary$, for any $t$ polynomial in $n$, we have: 
   $$\TD\left( \underset{\substack{U \leftarrow \haarunitaries_n\\ k \leftarrow \{0,1\}^{\secp}}}{\expect}\left[\ketbra{\Adversary_{t}^{G_k^U,U}}{\Adversary_{t}^{G_k^U,U}}\right],\ \underset{\substack{U \leftarrow \haarunitaries_n\\ V \leftarrow \haarunitaries_n}}{\expect}\left[\ketbra{\Adversary_{t}^{V,U}}{\Adversary_{t}^{V,U}}\right]\right) \leq \negl(\secp).$$
Note that $\Adversary$ makes $\ell$ queries to $G_k^U$ (or $V$) and $t-\ell$ queries to $U$. 
\end{theorem}

\ifllncs
\noindent The proof of~\Cref{thm:ell_design} can be found in~\Cref{app:unitarydesign:shortkeys}.
\else
\begin{proof}[Proof of~\Cref{thm:ell_design}]
Consider the following hybrids. \\

\noindent $\hybrid_1$: Output $\rho_1=\underset{\substack{U \leftarrow \haarunitaries_n\\ k \leftarrow \{0,1\}^{\secp}}}{\expect}\left[\ketbra{\Adversary_{t}^{G_k^U,U}}{\Adversary_{t}^{G_k^U,U}}\right]$.

\noindent $\hybrid_2$: Output 
\begin{equation*}
    \rho_2 := \underset{\substack{k \leftarrow \{0,1\}^{\secp}}}{\expect}\left[\Tr_{\reg{E}_1}\left( \ketbra{\Adversary_t^{Z^k{\pcfpr{\ell}{n}}_{\reg{A}\reg{E}_1}^{(\reg{E}_1)},{{\pcfpr{\ell}{n}}_{\reg{A}\reg{E}_1}^{(\reg{E}_1)}} }}{\Adversary_t^{Z^k{\pcfpr{\ell}{n}}_{\reg{A}\reg{E}_1}^{(\reg{E}_1)},{{\pcfpr{\ell}{n}}_{\reg{A}\reg{E}_1}^{(\reg{E}_1)}} }}_{\reg{A}\reg{B}\reg{E}_1}  \right)\right].
\end{equation*}

\begin{myclaim}
$\TD(\rho_1,\rho_2) \leq O\left(\frac{\sqrt{\ell}t^{\ell+1}}{2^{\secp/2}}\right) + \frac{4t(t-1)}{N+1}$.
\end{myclaim}
\begin{proof}
Follows from~\Cref{cor:ind:twopcfpr}. 
\end{proof}

\noindent $\hybrid_3$: Output 
\begin{equation*}
    \rho_3 := \Tr_{\reg{E}_1\reg{E}_2}\left( \ketbra{\Adversary_t^{{\pcfpr{\ell}{n}}_{\reg{A}\reg{E}}^{(\reg{E}_1)},{{\pcfpr{\ell}{n}}_{\reg{A}\reg{E}}^{(\reg{E}_2)}} }}{\Adversary_t^{{\pcfpr{\ell}{n}}_{\reg{A}\reg{E}}^{(\reg{E}_1)},{{\pcfpr{\ell}{n}}_{\reg{A}\reg{E}}^{(\reg{E}_2)}} }}_{\reg{A}\reg{B}\reg{E}_1\reg{E}_2}  \right).
\end{equation*}
\begin{myclaim}
$\rho_2 = \rho_3$.
\end{myclaim}
\begin{proof}
    We note that 
    \begin{multline*}
        \rho_2 = 
        \Tr_{\reg{E}_1\reg{K}} \Bigg[ \left( \frac{1}{\sqrt{2^{\secp}}}\sum_{k}\ket{k}_{\reg{K}}\ket{\Adversary_t^{Z^k{\pcfpr{\ell}{n}}_{\reg{A}\reg{E}_1}^{(\reg{E}_1)},{{\pcfpr{\ell}{n}}_{\reg{A}\reg{E}_1}^{(\reg{E}_1)}} }}_{\reg{A}\reg{B}\reg{E}_1} \right) \\
        \left(\frac{1}{\sqrt{2^{\secp}}}\sum_{k'}\bra{k'}_{\reg{K}}\bra{\Adversary_t^{Z^{k'}{\pcfpr{\ell}{n}}_{\reg{A}\reg{E}_1}^{(\reg{E}_1)},{{\pcfpr{\ell}{n}}_{\reg{A}\reg{E}_1}^{(\reg{E}_1)}} }}_{\reg{A}\reg{B}\reg{E}_1} \right) \Bigg]
    \end{multline*}
    and
    \begin{equation*}
        \rho_3 = \Tr_{\reg{E}_1\reg{E}_2}\left( \ketbra{\Adversary_t^{{\pcfpr{\ell}{n}}_{\reg{A}\reg{E}}^{(\reg{E}_1)},{{\pcfpr{\ell}{n}}_{\reg{A}\reg{E}}^{(\reg{E}_2)}} }}{\Adversary_t^{{\pcfpr{\ell}{n}}_{\reg{A}\reg{E}}^{(\reg{E}_1)},{{\pcfpr{\ell}{n}}_{\reg{A}\reg{E}}^{(\reg{E}_2)}} }}_{\reg{A}\reg{B}\reg{E}_1\reg{E}_2} \right).
    \end{equation*}

    \noindent We show this by showing that there exists an isometry on the register $\reg{E}_1\reg{K}$ that maps $$\ket{\psi_1} = \frac{1}{\sqrt{2^{\secp}}}\sum_{k}\ket{k}_{\reg{K}}\ket{\Adversary_t^{Z^k{\pcfpr{\ell}{n}}_{\reg{A}\reg{E}_1}^{(\reg{E}_1)},{{\pcfpr{\ell}{n}}_{\reg{A}\reg{E}_1}^{(\reg{E}_1)}} }}_{\reg{A}\reg{B}\reg{E}_1}$$ to $$\ket{\psi_2} = \ket{\Adversary_t^{{\pcfpr{\ell}{n}}_{\reg{A}\reg{E}}^{(\reg{E}_1)},{{\pcfpr{\ell}{n}}_{\reg{A}\reg{E}}^{(\reg{E}_2)}} }}_{\reg{A}\reg{B}\reg{E}_1\reg{E}_2}.$$

    \noindent Without loss of generality, assume $\Adversary$ queries the first oracle on $\bfa = \set{a_1,\ldots,a_{\ell}}$ indices and the second oracle on $\bfb = [t]\setminus\set{a_1,\ldots,a_{\ell}}$ indices. Hence we expand $\ket{\psi_2}$ using the definition of ${\pcfpr{\ell}{n}}_{\reg{A}\reg{E}}^{(\reg{E}_1)}$ and ${{\pcfpr{\ell}{n}}_{\reg{A}\reg{E}}^{(\reg{E}_2)}}$,
    \begin{equation*}
    \begin{split}
        \ket{\psi_2} = \sum_{\substack{x_1,\ldots,x_t\in [N]\\ y_1\in\cfreeset_{\ell,\secp}(\set{})\\ y_2\in\cfreeset_{\ell,\secp}(\set{y_1})\\ \vdots\\ y_t\in\cfreeset_{\ell,\secp}(\set{y_1,\ldots,y_{t-1}})}}\prod_{i=1}^{t}\frac{1}{\sqrt{|\cfreeset_{\ell,\secp}(\set{y_1,\ldots,y_{i-1}})|}}\left( \ketbra{y_i}{x_i}_{\reg{A}} \cdot A_{\reg{A} \reg{B}}^{(i)} \right) \ket{0}_{\reg{A}\reg{B}} \\
        \otimes \ket{\set{(x_i,y_i)}_{i\in\bfa}}_{\reg{E_1}}\ket{\set{(x_i,y_i)}_{i\in \bfb}}_{\reg{E_2}}\,.
    \end{split}
    \end{equation*}
    
    \noindent For conciseness, we use the subnormalized vector $\ket{\phi_{\Vec{x},\Vec{y}}}$ to denote $$\ket{\phi_{\Vec{x},\Vec{y}}}_{\reg{A} \reg{B}}:=\prod_{i=1}^{t}\frac{1}{\sqrt{|\cfreeset_{\ell,\secp}(\set{y_1,\ldots,y_{i-1}})|}}\left( \ketbra{y_i}{x_i}_{\reg{A}} \cdot A_{\reg{A} \reg{B}}^{(i)} \right) \ket{0}_{\reg{A}\reg{B}}.$$
    Hence, 
    $$\ket{\psi_2} = \sum_{\substack{x_1,\ldots,x_t\in [N]\\ y_1\in\cfreeset_{\ell,\secp}(\set{})\\ y_2\in\cfreeset_{\ell,\secp}(\set{y_1})\\ \vdots\\ y_t\in\cfreeset_{\ell,\secp}(\set{y_1,\ldots,y_{t-1}})}}\ket{\phi_{\Vec{x},\Vec{y}}}_{\reg{A} \reg{B}}\ket{\set{(x_i,y_i)}_{i\in\bfa}}_{\reg{E_1}}\ket{\set{(x_i,y_i)}_{i\in \bfb}}_{\reg{E_2}}.$$
    
    \noindent Similarly, we can we expand $\ket{\psi_1}$ as $$\ket{\psi_1} = \frac{1}{\sqrt{2^\secp}} \sum_{\substack{k\in\bit^\secp\\ x_1,\ldots,x_t\in [N]\\ y_1\in\cfreeset_{\ell,\secp}(\set{})\\ y_2\in\cfreeset_{\ell,\secp}(\set{y_1})\\ \vdots\\ y_t\in\cfreeset_{\ell,\secp}(\set{y_1,\ldots,y_{t-1}})}}\prod_{i=1}^{\ell} (-1)^{\langle y_{a_i},k||0\rangle}\ket{\phi_{\Vec{x},\Vec{y}}}_{\reg{A} \reg{B}} \ket{\set{(x_i,y_i)}_{i\in [t]}}_{\reg{E_1}}\ket{k}_{\reg{K}}.$$
    
    \noindent Moving the sum over $k$ to the purification, we get
    \begin{multline*}
        \ket{\psi_1} = 
        \sum_{\substack{x_1,\ldots,x_t\in [N]\\ y_1\in\cfreeset_{\ell,\secp}(\set{})\\ y_2\in\cfreeset_{\ell,\secp}(\set{y_1})\\ \vdots\\ y_t\in\cfreeset_{\ell,\secp}(\set{y_1,\ldots,y_{t-1}})}}\ket{\phi_{\Vec{x},\Vec{y}}}_{\reg{A} \reg{B}} \ket{\set{(x_i,y_i)}_{i\in [t]}}_{\reg{E_1}} 
        \otimes \frac{1}{\sqrt{2^\secp}} \sum_{k\in\bit^\secp} (-1)^{\langle \oplus_{i=1}^{\ell}y_{a_i},k||0\rangle}\ket{k}_{\reg{K}}.
    \end{multline*}

    \noindent Next, we show that there exists an isometry $W$ such that for all $\set{(x_i,y_i)}_{i\in [t]}\in\setsofrels^{\cfree(\ell,\secp)}$, 
    \begin{equation*}
    \begin{split}
        & W\ket{\set{(x_i,y_i)}_{i\in [t]}}_{\reg{E_1}} \otimes \frac{1}{\sqrt{2^\secp}} \sum_{k\in\bit^\secp} (-1)^{\langle \oplus_{i=1}^{\ell}y_{a_i},k||0\rangle}\ket{k}_{\reg{K}} \\
        =\ & \ket{\set{(x_i,y_i)}_{i\in\bfa}}_{\reg{E_1}}\ket{\set{(x_i,y_i)}_{i\in \bfb}}_{\reg{E_2}}.
    \end{split}
    \end{equation*}

    \noindent We describe $W$ as the following procedure. We start by applying Hadamard on $\reg{K}$ to get $$\ket{\set{(x_i,y_i)}_{i\in [t]}}_{\reg{E_1}}\ket{\oplus_{i=1}^{\ell}y_{a_i}[1:\secp]}_{\reg{K}}.$$ Next, using $\oplus_{i=1}^{\ell}y_{a_i}[1:\secp]$ as the key we can define a partition of $\set{(x_i,y_i)}_{i\in [t]}$ the XOR of $\ell$-sized subset of $y$'s is $\ket{\oplus_{i=1}^{\ell}y_{a_i}[1:\secp]}$. Since $\set{(x_i,y_i)}_{i\in [t]}\in\setsofrels^{\cfree(\ell,\secp)}$, there is a unique partition with this property, namely $\set{y_i}_{i\in\bfa}$ and $\set{(x_i,y_i)}_{i\in \bfb}$. Hence, applying $\Vpart$, we get the state as $$\ket{\set{(x_i,y_i)}_{i\in \bfa}}_{\reg{E_1}}\ket{\set{(x_i,y_i)}_{i\in\bfb}}_{\reg{E_2}}\ket{\oplus_{i=1}^{\ell}y_{a_i}[1:\secp]}_{\reg{K}}.$$ Next, we XOR the bits $y$'s from $\reg{E}_1$ with $\reg{K}$ to get $\ket{0}_{\reg{K}}$,tracing out $\ket{0}_{\reg{K}}$, we get $$\ket{\set{(x_i,y_i)}_{i\in\bfa}}_{\reg{E_1}}\ket{\set{(x_i,y_i)}_{i\in \bfb}}_{\reg{E_2}}.$$ Hence, $W$ maps $\ket{\psi_1}$ to $\ket{\psi_2}$, and $\rho_2=\rho_3$.
\end{proof}

\noindent $\hybrid_4$: Output $\rho_4= \underset{\substack{U \leftarrow \haarunitaries_n\\ V \leftarrow \haarunitaries_n}}{\expect}\left[\ketbra{\Adversary_{t}^{V,U}}{\Adversary_{t}^{V,U}}\right]$. 

\begin{myclaim}
$\TD(\rho_3,\rho_4) \leq O\left(\frac{\sqrt{\ell}t^{\ell+1}}{2^{\secp/2}}\right) + \frac{4t(t-1)}{N+1}$
\end{myclaim}
\begin{proof}
Follows from~\Cref{cor:ind:twopcfpr}. 
\end{proof}
\noindent This completes the proof of~\Cref{thm:ell_design}.
\end{proof}
\fi

\ifllncs

\else
    \section{Unbounded-Query Secure Non-Adaptive PRUs: Barriers}
Here we show that the construction of $o(\secp/\log(\secp))$-copy secure PRUs are almost tight for constructions that make arbitrary depth-$1$ calls to the PRU.  The main observation is that we can apply the quantum OR technique from \cite{chen2024power} to the Choi states of any PRU that only makes non-adaptive calls.  In particular, consider the following:
\begin{lemma}\label{lem:modify_choi_state}
    Let $U$ be a unitary and $C^{U}$ be any algorithm that makes $t$ non-adaptive calls to $U$.  Then there is an algorithm that prepares $\ket{\Phi_{C^U}}$ given $t$-copies of $\ket{\Phi_{U}}$, where $\ket{\Phi_{U}}$ is the Choi state of $U$.
\end{lemma}
\begin{proof}
    We can assume without loss of generality that $C^U$ applies a unitary $A$, then $U^{\otimes t}$, and then another unitary $B$.  Thus, we can prepare the Choi state $\ket{\Phi_{C^U}}$ as follows:
    \begin{align*}
        (\id \otimes B)(A^{\intercal} \otimes \id) \ket{\Phi_U}^{\otimes t} &= (\id \otimes B)(A^{\intercal} \otimes \id)\ket{\Phi_{U^{\otimes t}}}\\
        &= (\id \otimes B)(A^{\intercal} \otimes \id)(\id \otimes U^{\otimes t}) \ket{\Omega}\\
        &= (\id \otimes B)(\id \otimes U^{\otimes t}) (A^{\intercal} \otimes \id) \ket{\Omega}\\
        &= (\id \otimes B)(\id \otimes U^{\otimes t}) (\id \otimes A) \ket{\Omega}\\
        &= \ket{\Phi_{B U^{\otimes t} A}}\\
        &= \ket{\Phi_{C^U}}\,.
    \end{align*}
    Here the first line uses the definition of the Choi state.  Then we use the fact that $U^{\otimes t}$ and $A^{\intercal}$ act on different registers and therefore commute with each other.  Finally, we use the ricochet property of EPR pairs, and the definition of the Choi state and $C^U$.  
\end{proof}

\subsection{Optimality of $\ell$-Query PRUs in~\Cref{sec:unitarydesign:shortkeys}}

\noindent We argue that the construction presented in~\Cref{sec:unitarydesign:shortkeys} is optimal in terms of its query bound. That is, we show that for any $\ell$-query PRU, where the algorithm makes a {\em single} query to the Haar random oracle, it has to be the case that $\ell = O(\secparam/\log(\secparam)^{1+\varepsilon})$, where $\varepsilon > 0$. Our proof is inspired by the dimension counting argument from~\cite{AGQY22,AGL24}. 

\begin{theorem}
Let $\{G_k^U\}_{k \in \{0,1\}^{\secparam}}$ be a set of oracle algorithms, where $G_k$ runs in polynomial time. Moreover, for every $k \in \{0,1\}^{\secparam}$, there exists efficiently implementable unitaries $A_k$ and $B_k$ such that $G_k=B_k U A_k$. 
\par Then, $\{G_k^U\}_{k \in \{0,1\}^{\secparam}}$ is not an $\ell$-query PRU for $\ell=\omega(\secparam/\log(\secparam))$ in ${\sf iQHROM}$, where $U$ is the Haar unitary. 
\end{theorem}
\begin{proof}
Before we describe the distinguisher that violates the security of $G_k^U$, we first state some observations. 

\paragraph{Observations.} From~\Cref{lem:modify_choi_state}, note that $\ket{\Phi_{G_k^U}} = (A_k^{\intercal} \otimes B_k)\ket{\Phi_U}$. We define $\rho_0^{(t,\ell)}$ below.

\begin{eqnarray*}
\rho_0^{(t,\ell)} & = & \underset{\substack{k \leftarrow \{0,1\}^{\secp}\\ U \leftarrow \haarunitaries_n}}{\expect}\left[ \ketbra{\Phi_U}{\Phi_U}^{\otimes t} \otimes \ketbra{\Phi_{G_k^U}}{\Phi_{G_k^U}}^{\otimes \ell} \right] \\
& = & \underset{\substack{k \leftarrow \{0,1\}^{\secp}\\ U \leftarrow \haarunitaries_n}}{\expect}\left[ \ketbra{\Phi_U}{\Phi_U}^{\otimes t} \otimes \left(\left(A_k^{\intercal} \otimes B_k \right)\ketbra{\Phi_{U}}{\Phi_{U}}\left((A_k^{\intercal})^{\dagger} \otimes (B_k)^{\dagger} \right) \right)^{\otimes \ell} \right] \\
& = &  \underset{\substack{k \leftarrow \{0,1\}^{\secp}}}{\expect}\left[ \left(\id \otimes \left(A_k^{\intercal} \otimes B_k\right)^{\otimes \ell} \right) \underset{\substack{U \leftarrow \haarunitaries_n}}{\expect}\left[ \ketbra{\Phi_U}{\Phi_U}^{\otimes t} \otimes \ketbra{\Phi_{U}}{\Phi_{U}}^{\otimes \ell} \right] \left(\left(\id \otimes \left(A_k^{\intercal} \otimes B_k\right)^{\dagger} \right)^{\otimes \ell} \right) \right].
\end{eqnarray*}

\noindent We determine ${\sf rank}\left(\rho_{0}^{(t,\ell)} \right)$. In order to do that, note that, from~\cite{Har23},\footnote{Specifically,~\cite{Har23} gives a closed form expression for $\sigma_0 = \underset{\substack{U \leftarrow \haarunitaries_n}}{\expect}\left[ \ketbra{\Phi_U}{\Phi_U}^{\otimes t} \otimes \ketbra{\Phi_{U}}{\Phi_{U}}^{\otimes \ell} \right]$ (see equation (23) in~\cite{Har23}) and also, $\sigma_1=\underset{\substack{U \leftarrow \haarunitaries_n}}{\expect}\left[ U^{\otimes (t+\ell)} \ketbra{0^{2m}}{0^{2m}} U^{\otimes (t+\ell)} \right]$ (see equation (22) in~\cite{Har23}). From the closed form expressions, it follows that ${\sf rank}(\sigma_0)={\sf rank}(\sigma_1)$. But $\sigma_1$ is precisely, the normalized projector on the symmetric subspace $\Pi_{{\sf Sym}}^{(t+\ell),2m}$.} the following holds: 
$${\sf rank}\left(\underset{\substack{U \leftarrow \haarunitaries_n}}{\expect}\left[ \ketbra{\Phi_U}{\Phi_U}^{\otimes t} \otimes \ketbra{\Phi_{U}}{\Phi_{U}}^{\otimes \ell} \right] \right) = \dim(\Pi_{{\sf Sym}}^{(t+\ell),2m}),$$ where $\Pi_{{\sf Sym}}^{(t+\ell),2m}$ is the projector onto the symmetric subspace spanned by $(t+\ell)$-copy $m$-qubit states. 

\noindent Thus, we have the following:
$${\sf rank}\left( \rho_0^{(t,\ell)} \right) \leq 2^{\secp} \cdot {2^{2m} + \ell + t - 1 \choose \ell + t}.$$
\noindent We define another state $\rho_1^{(t,\ell)}$ as 

$$\rho_1^{(t,\ell)} = \underset{\substack{U \leftarrow \haarunitaries_n\\ V \leftarrow \haarunitaries_n}}{\expect}\left[ \ketbra{\Phi_U}{\Phi_U}^{\otimes t} \otimes \ketbra{\Phi_{V}}{\Phi_{V}}^{\otimes \ell} \right].$$

\noindent We determine ${\sf rank}\left(\rho_{1}^{(t,\ell)} \right)$ and have
$${\sf rank}\left( \rho_1^{(t,\ell)} \right) = {2^{2m} + \ell - 1 \choose \ell} \cdot {2^{2m} + t - 1 \choose t}.$$

\noindent Finally, we use the following fact to upper bound $\frac{\rank(\rho_0^{(t,\ell)})}{\rank(\rho_1^{(t,\ell)})}$. 

\begin{myclaim}[\cite{AGL24}]
\label{clm:agl24:bound}
Suppose $\ell = \omega(\secp/\log(\secp))$, $t=\secp^3$. Then: 
$$\frac{2^{\secp} \cdot {2^{2m} + \ell + t - 1 \choose \ell + t}}{{2^{2m} + \ell - 1 \choose \ell} \cdot {2^{2m} + t - 1 \choose t}} \leq \frac{1}{2^{\secp}}$$
\end{myclaim}
\noindent Using the above claim, we have (for the parameters stated in the claim) that $\frac{\rank(\rho_0^{(t,\ell)})}{\rank(\rho_1^{(t,\ell)})} \leq \frac{1}{2^{\secp}}$.

\paragraph{Distinguisher.} Let $\Pi$ be the projector that projects onto the eigenspace of $\rho_0^{(t,\ell)}$. The distinguisher $\Adversary$, with oracle access to the Haar random unitary $U$ and another unitary $W$ (which is either the unitary design that depends on $U$ or a Haar random unitary that is independently sampled), does the following: 
\begin{itemize}
    \item $\Adversary$ creates $t$ copies of $\ketbra{\Phi_{U}}{\Phi_U}$ by making $t$ queries to $U$.
    \item $\Adversary$ creates $\ell$ copies of $\ketbra{\Phi_{W}}{\Phi_W}$ by making $\ell$ queries to $W$.
    \item It then measures $\ketbra{\Phi_{U}}{\Phi_U}^{\otimes t} \otimes \ketbra{\Phi_W}{\Phi_W}^{\otimes \ell}$ using the measurement basis $\{\Pi,\id-\Pi\}$. 
    \item If the measurement succeeds, output 1. Otherwise, output 0. 
\end{itemize}

\noindent Let us consider the following cases. \\

\noindent {\bf Case 1. $W=G_k^U$}: In this case, the probability that the distinguisher always outputs 1. \\

\noindent {\bf Case 2. $W=V$, where $V$ is an i.i.d Haar random unitary}: We have the following:
\begin{align*}
\Tr\left(\Pi \rho_1^{\ell,t} \right)
& \leq \frac{\Tr(\Pi)}{\rank(\rho_1^{(\ell,t)})} \\
& = \frac{\rank(\rho_0^{(\ell,t)})}{\rank(\rho_1^{(\ell,t)})} \\
& \leq \frac{1}{2^{\secparam}}\ (\text{\Cref{clm:agl24:bound}}) \hfill \qedhere
\end{align*}
\end{proof}

\subsection{Almost Optimality among Parallel-Query Pseudorandom Unitaries} 
In the previous section, we showed that a PRU construction that makes only a single call to the Haar random unitary cannot be more secure than our construction.  Here we show that PRU constructions that make arbitrarily many queries in \emph{parallel} cannot be much more secure than our construction.

We begin as in \cite{chen2024power}, showing that many copies of the Choi state of a Haar random unitary are far from the Choi state of any fixed unitary with high probability.  
To prove this, we follow the proof of Lemma 5.4 from \cite{chen2024power}, with a slight modification to handle Choi states of Haar random unitaries.
\begin{lemma}[Lemma 5.4 from \cite{chen2024power}]
    Let $\ket{\psi_0}$ be a $2^n$-dimensional state, then
    \begin{equation*}
        \Pr_{\ket{\psi} \gets \Haar_n} \left[\abs{\braket{\psi}{\psi_0}}^2 \geq \frac{1}{2}\right] \leq 8 \exp\left(\frac{- 2^n}{600}\right)\,. 
    \end{equation*}
\end{lemma}
We prove a slightly modified version of the lemma as it pertains to Choi states.
\begin{lemma}
    \label{lem:Haar_choi_state_low_overlap}
    Let $U_0$ be a $2^n \times 2^n$ unitary matrix, then the following holds
    \begin{equation*}
        \Pr_{U \gets \mu_n} \left[\abs{\braket{\Phi_{U}}{\Phi_{U_0}}}^2 \geq \frac{1}{2}\right] \leq 2 \exp\left(\frac{-2^n}{96}\right)\,.
    \end{equation*}
\end{lemma}
\begin{proof}
    Our goal is to apply Levy's lemma to the squared inner product of the Choi states.  First, consider the following function of two unitaries:
    \begin{equation*}
        f_1(U) = \Re \left(\bra{\Omega} U U_0^{\dagger} \ket{\Omega}\right)\,.
    \end{equation*}
    Then the following holds
    \begin{align*}
        \abs{f_1(U) - f_1(V)} &= \abs{\Re \left(\bra{\Omega} U U_0^{\dagger} \ket{\Omega}\right) - \Re \left(\bra{\Omega} V U_0^{\dagger} \ket{\Omega}\right)}\\
        &\leq \abs{ \bra{\Omega} (U - V) U_0^{\dagger} \ket{\Omega} }\\
        &= \abs{\frac{1}{2^n}\Tr[(U - V) U_0^{\dagger} ]}\\
        &\leq \frac{1}{2^n} \norm{U - V}_{F} \norm{U_0}_{F} \\
        &= \norm{U - V}_{F}\,.
    \end{align*}
    The same inequalities apply for $f_2(U) = \Im\left(\bra{\Omega} U U_0^{\dagger} \ket{\Omega}\right)$.  Applying Levy's lemma, we have the following:
    \begin{equation*}
        \Pr_{U \gets \mu_n} \left[f_1(U) \geq 1/2\right] \leq \exp(-\frac{2^n}{96})\,.
    \end{equation*}
    The same holds for $f_2$.  Thus by a union bound, we have that
    \begin{align*}
        \Pr_{U \gets \Haar_n} \left[\abs{\braket{\Phi_{U}}{\Phi_{U_0}}}^2 \geq \frac{1}{2}\right] &= \Pr_{U \gets \mu_n}\left[f_1^2(U) + f_2^2(U) \geq \frac{1}{2}\right]\\
        &\leq \Pr_{U \gets \mu_n}\left[f_1(U) \geq \frac{1}{2}\right] + \Pr_{U \gets \mu_n}\left[f_2(U) \geq \frac{1}{2}\right]\\
        &\leq 2\exp(-\frac{2^n}{96})\,,
    \end{align*}
    as desired.  This completes the proof of the lemma.
\end{proof}

The final technical tool we need for this section is the quantum OR lemma.
\begin{lemma}[Quantum OR lemma~\cite{WB24,harrow2017sequential}]\label{lem:quantum_or}
    Let $\Pi_1, \ldots, \Pi_m$ be a collection of projectors, and let $\epsilon \leq \frac{1}{2}$ and $\delta$.  
    Let $\rho$ be a state with the promise that either there exists $\Pi_i$ such that $\Tr[\Pi_i \rho] \geq 1 - \epsilon$ (case $1$), or for all $i$, $\Tr[\Pi_i \rho] \leq \delta$ (case $2$).
    Then there is a polynomial space quantum algorithm $A$ such that 
    \begin{enumerate}
        \item In case $1$, $A$ accepts with probability at least $(1 - \epsilon)^2 / 4.5$.
        \item In case $2$, $A$ accepts with probability at most $2m \delta$.  
    \end{enumerate}
    $A$ only requires black box access to the projector valued measurements $\{\Pi_i, \id - \Pi_i\}$ and $O(\log m)$ additional space.
\end{lemma}

Using these results, we can prove that no construction of PRU in the $\mathsf{iQHROM}$ can only make a single parallel query to the common Haar random unitary.

\begin{theorem}[Tightness of PRU] 
For any construction of PRU making non-adaptive calls to the Haar random oracle, with output size equal to $\secp$, there exists an adversary that breaks its security by making $O(\secp)$ non-adaptive queries to the PRU and $p(\secp)$ non-adaptive queries to the Haar random oracle for some polynomial $p$.
\end{theorem}
\begin{proof}
    We construct an adversary against any PRU that makes non-adaptive queries using the quantum OR attack from \cite{chen2024power}.  
    Assume that the PRU construct on key $k$ first calls unitary $A_k$, then calls the Haar random unitary $U^{\otimes t}$, and finally runs $B_k$.  We describe the $k^{th}$ measurement that we use as input to the quantum OR algorithm.  Starting from $O(t\secp)$ copies of $\ket{\Phi_{U}}$ and $O(\secp)$ copies of $\ket{\Phi_{\oracle}}$, for every key $k$, the adversary performs the algorithm described by \Cref{lem:modify_choi_state} on $\left(\ket{\Phi_{U}}^{\otimes t}\right)^{\otimes O(\secp)}$ to get $O(\secp)$ copies of $\ket{\Phi_{U_k}}$, the Choi state corresponding to the PRU with key $k$.  
    Then the measurement performs a SWAP test with all $O(\secp)$ copies of $\ket{\Phi_{\oracle}}$.  
    Finally, the measurement uncomputes the computation.  
    Formally, the $k$'th measurement operator $\Pi_k$ has the following description, where the registers $\reg{AB}$ contain copies of the Choi state $\ket{\Phi_U}$ and the registers $\reg{CD}$ contain copies of the Choi state $\ket{\Phi_\oracle}$.
    \begin{equation*}
        \Pi_k = \left(\left((A^{*}_k \otimes B^{\dagger}_k)_{\reg{AB}} \otimes \id_{\reg{BC}} \right)\left(\Pi^{\mathrm{sym}}_{\reg{ABCD}}\right)\left((A^{\intercal}_k \otimes B_k)_{\reg{AB}} \otimes \id_{\reg{BC}}\right)\right)^{\otimes O(\secp)}\,.
    \end{equation*}
    Here $A^*$ is the complex conjugate of $A$, as opposed to $A^{\dagger}$, which is the conjugate transpose.  
    From \Cref{lem:modify_choi_state}, when the oracle $\oracle$ is equal to $U_{k}$ for some $k$, there is a choice (namely the one corresponding to $k$) of measurement that accepts with probability $1$.  
    Thus, in \Cref{lem:quantum_or}, we can set $\epsilon = 0$.   
    
    Furthermore, from \Cref{lem:Haar_choi_state_low_overlap} if the oracle is a Haar random unitary, then with probability $1 - 2\exp(-\frac{d}{96})$, the Choi states have less than $1/2$ squared fidelity.  The probability that the swap test succeeds on all $O(\secp)$ copies of the state is thus upper bounded by $\left(\frac{3}{4}\right)^{O(\secp)}$. 
    Setting the number of copies of the state to be $c\secp$ for any constant such that $(3/4)^c < 1/2$, we can apply a union bound over all $2^{\secp}$ measurements to get that the probability that the quantum OR accepts for a Haar random state is negligible in $\secp$.
    Applying \Cref{lem:quantum_or}, we see that our adversary breaks PRU security with constant probability. 
\end{proof}

    \newcommand{\adversary}{{\cal A}}
\newcommand{\reduction}{{\cal R}}
\newcommand{\subspace}{{\cal S}}
\section{Unbounded-Copy PRS in the iQHROM} 

\noindent In this section the main result is a simple construction of pseudorandom states in the $\mathsf{iQHROM}$.  The construction of a PRS generator $G$, with oracle access to an $n(\secp)$-qubit Haar random unitary $U$ for $n \geq \secp$, is as follows: on input $k \in \bit^\secp$, output $U\ket{k||0^{n-\secp}}$.  
We note that this result is probably provable using the lower bound on unstructured search. However, as a warm-up for the next section, we present a nice argument using the path-recording framework.

\begin{theorem}\label{thm:multi_copy_prs}
$G$ is a multi-copy secure PRS generator in the $\mathsf{iQHROM}$. Formally, an adversary that receives $t$ copies of a state that is either the output of $G$ on a random $k \in\bit^\lambda$ or a Haar random state that is independent of $U$ and makes $s$ adaptive queries to $U$ has distinguishing advantage at most $O\qty(\sqrt{\frac{s}{2^\secp}} + \frac{(t+s)^2}{\sqrt{2^n}})$.
\end{theorem}
\begin{proof}
    Fix $n,\lambda,t,s$ and let $N := 2^n$ for the rest of the proof. We complete by hybrid arguments. \\
    
    \noindent $\hybrid~1$: The challenger samples a uniformly random $k\in\bit^\lambda$ and sends $t$ copies of $U\ket{k||0^{n-\secp}}$ to the adversary, where $U$ is the Haar random oracle. Then, the adversary makes $s$ queries to $U$. \\
    
    \noindent $\hybrid~2$: Using the path-recording framework, we can write the initial state of an adversary that receives $t$ copies of $\mathsf{PR}\ket{k||0^{n-\secp}}$ for a uniformly random $k \in \bit^\lambda$ as follows
    \ifllncs
    \begin{equation*}
    \begin{split}
        & \frac{1}{\sqrt{2^{\lambda}}}\sum_{k \in \bit^\lambda} \left( \bigotimes_{i = 1}^t  \mathsf{PR} \ket{k||0^{n-\secp}}_{\reg{C_i}} \right) \ket{k}_{\reg{K}} \\
        = & \sqrt{\frac{1}{2^{\lambda} N^{\downarrow t}}} \sum_{\substack{k \in \bit^\lambda, \\ (y_1, \ldots, y_t) \in [N]^{t}_{\mathrm{dist}}}} \ket{y_1, \ldots, y_t}_{\reg{C_1}\dots\reg{C_t}} \ket{\{(k||0^{n-\secp}, y_i)\}_{i = 1}^{t}}_{\reg{E}} \ket{k}_{\reg{K}}.
    \end{split}
    \end{equation*}
    \else
    \begin{equation*}
        \frac{1}{\sqrt{2^{\lambda}}}\sum_{k \in \bit^\lambda} \left( \bigotimes_{i = 1}^t  \mathsf{PR} \ket{k||0^{n-\secp}}_{\reg{C_i}} \right) \ket{k}_{\reg{K}}
        = \sqrt{\frac{1}{2^{\lambda} N^{\downarrow t}}} \sum_{\substack{k \in \bit^\lambda, \\ (y_1, \ldots, y_t) \in [N]^{t}_{\mathrm{dist}}}} \ket{y_1, \ldots, y_t}_{\reg{C_1}\dots\reg{C_t}} \ket{\{(k||0^{n-\secp}, y_i)\}_{i = 1}^{t}}_{\reg{E}} \ket{k}_{\reg{K}}.
    \end{equation*}
    \fi
    Then we can write the state of any adversary that given registers $(\reg{C_1},\dots,\reg{C_t})$ and makes $s$ many queries to $\mathsf{PR}$ as follows
    \ifllncs
    \begin{equation*}
    \begin{split}
        \ket{\psi_{\mathrm{real}}}_{\reg{A}\reg{B}\reg{E}\reg{K}} := 
        \sqrt{\frac{1}{2^{\lambda}N^{\downarrow t+s}}} \sum_{\substack{k \in \bit^\lambda, \\ (x_1, \ldots, x_s) \in [N]^s, \\ (y_1, \ldots, y_t, z_1, \ldots, z_s) \in [N]^{t+s}_{\mathrm{dist}}}} \left(\prod_{i = 1}^{s} \ket{z_i}\!\bra{x_i}_{\reg{A}} A^{(i)}_{\reg{AB}} \right) \ket{\psi_\mathrm{Init}(\Vec{y})}_{\reg{AB}} \\
        \otimes \ket{\{(k||0^n, y_i)\}_{i = 1}^{t} \cup \{(x_j, z_j)\}_{j = 1}^{s}}_{\reg{E}} \ket{k}_{\reg{K}},
    \end{split}
    \end{equation*}
    \else
    \begin{multline*}
        \ket{\psi_{\mathrm{real}}}_{\reg{A}\reg{B}\reg{E}\reg{K}} := \\
        \sqrt{\frac{1}{2^{\lambda}N^{\downarrow t+s}}} \sum_{\substack{k \in \bit^\lambda, \\ (x_1, \ldots, x_s) \in [N]^s, \\ (y_1, \ldots, y_t, z_1, \ldots, z_s) \in [N]^{t+s}_{\mathrm{dist}}}} \left(\prod_{i = 1}^{s} \ket{z_i}\!\bra{x_i}_{\reg{A}} A^{(i)}_{\reg{AB}} \right) \ket{\psi_\mathrm{Init}(\Vec{y})}_{\reg{AB}} \ket{\{(k||0^{n-\secp}, y_i)\}_{i = 1}^{t} \cup \{(x_j, z_j)\}_{j = 1}^{s}}_{\reg{E}} \ket{k}_{\reg{K}},
    \end{multline*}
    \fi
    where $\ket{\psi_\mathrm{Init}(\Vec{y})}_{\reg{A}\reg{B}} := \ket{y_1, \ldots, y_t}_{\reg{C_1}\dots\reg{C_t}} \ket{0}_{\reg{B'}}$ for some ancilla register $\reg{B'}$. From~\Cref{thm:MH24}, the trace distance between $\Tr_{\reg{EK}}(\ketbra{\psi_{\mathrm{real}}}{\psi_{\mathrm{real}}}_{\reg{A}\reg{B}\reg{E}\reg{K}})$ and the adversary's density matrix in~$\hybrid~1$ is $O\qty((t+s)^2/2^n)$. \\
    
    Here the key fact is that the $z_j$ and $y_i$ are distinct from each other, because the path-recording oracle only appends $z_j$ that are not already in the image of the relation.
    We can then apply an isometry on the purifying register that takes all elements of the relation state that have $k$ as the input and puts them in a new set. 
    To do so, we say that a pair $(R,k)$ of an injective relation $R$ of size $t+s$ and a key $k$ is \emph{good} if and only if 
    \[
    |\set{y\in\Im(R):\ (k||0^n,y)\in R}| = t.
    \]
    We define the projection $\Pi_{\mathrm{good}}$ onto the the subspace spanned by $\set{\ket{R}\ket{k}: (R,k) \text{ is good}}$. Explicitly,
    \[
    \Pi_{\mathrm{good}} := \sum_{(R,k) \text{ is good}} \ketbra{R}{R}_{\reg{E}} \otimes \ketbra{k}{k}_{\reg{K}}.
    \]
    
    \noindent $\hybrid~3$: Now, consider applying $\Pi_{\mathrm{good}}$ on $\ket{\psi_{\mathrm{real}}}_{\reg{A}\reg{B}\reg{E}\reg{K}}$. Observe that $\ket{\{(k||0^{n-\secp}, y_i)\}_{i = 1}^{t} \cup \{(x_j, z_j)\}_{j = 1}^{s}}_{\reg{E}} \ket{k}_{\reg{K}}$ vanishes if and only if $k \in \set{\Vec{x}}$, where $\set{\Vec{x}} :=  \bigcup_{i=1}^s \{x_i\}$. Therefore, we get the following subnormalized state
    \ifllncs
    \begin{equation*}
    \begin{split}
        \sqrt{\frac{1}{2^\lambda N^{\downarrow t+s}}} \sum_{\substack{(x_1, \ldots, x_s) \in [N]^s, \\ (y_1, \ldots, y_t, z_1, \ldots, z_s) \in [N]^{t+s}_{\mathrm{dist}}, \\ k \not\in \set{\Vec{x}}}} \left(\prod_{i = 1}^{s} \ket{z_i}\!\bra{x_i}_{\reg{A}} A^{(i)}_{\reg{AB}} \right) \ket{\psi_\mathrm{Init}(\Vec{y})}_{\reg{AB}} \\
        \otimes \ket{\{(k||0^{n-\secp}, y_i)\}_{i = 1}^{t} \cup \{(x_j, z_j)\}_{j = 1}^{s}}_{\reg{E}} \ket{k}_{\reg{K}}.
    \end{split}
    \end{equation*}
    \else
    \begin{equation*}
        \sqrt{\frac{1}{2^\lambda N^{\downarrow t+s}}} \sum_{\substack{(x_1, \ldots, x_s) \in [N]^s, \\ (y_1, \ldots, y_t, z_1, \ldots, z_s) \in [N]^{t+s}_{\mathrm{dist}}, \\ k \not\in \set{\Vec{x}}}} \left(\prod_{i = 1}^{s} \ket{z_i}\!\bra{x_i}_{\reg{A}} A^{(i)}_{\reg{AB}} \right) \ket{\psi_\mathrm{Init}(\Vec{y})}_{\reg{AB}} \ket{\{(k||0^{n-\secp}, y_i)\}_{i = 1}^{t} \cup \{(x_j, z_j)\}_{j = 1}^{s}}_{\reg{E}} \ket{k}_{\reg{K}}.
    \end{equation*}
    \fi
    Then applying the partition isometry, we obtain
    \ifllncs
    \begin{equation*}
    \begin{split}
        \sqrt{\frac{1}{2^\lambda N^{\downarrow t+s}}} \sum_{\substack{(x_1, \ldots, x_s) \in [N]^s, \\ (y_1, \ldots, y_t, z_1, \ldots, z_s) \in [N]^{t+s}_{\mathrm{dist}}, \\ k \not\in \set{\Vec{x}}}} \left(\prod_{i = 1}^{s} \ket{z_i}\!\bra{x_i}_{\reg{A}} A^{(i)}_{\reg{AB}} \right) \ket{\psi_\mathrm{Init}(\Vec{y})}_{\reg{AB}} \\
        \otimes \ket{\{(k||0^{n-\secp}, y_i)\}_{i = 1}^{t}}_{\reg{E_1}} \ket{\{(x_j, z_j)\}_{j = 1}^{s}}_{\reg{E_2}} \ket{k}_{\reg{K}}.
    \end{split}
    \end{equation*}
    \else
    \begin{equation*}
        \sqrt{\frac{1}{2^\lambda N^{\downarrow t+s}}} \sum_{\substack{(x_1, \ldots, x_s) \in [N]^s, \\ (y_1, \ldots, y_t, z_1, \ldots, z_s) \in [N]^{t+s}_{\mathrm{dist}}, \\ k \not\in \set{\Vec{x}}}} \left(\prod_{i = 1}^{s} \ket{z_i}\!\bra{x_i}_{\reg{A}} A^{(i)}_{\reg{AB}} \right) \ket{\psi_\mathrm{Init}(\Vec{y})}_{\reg{AB}} \ket{\{(k||0^{n-\secp}, y_i)\}_{i = 1}^{t}}_{\reg{E_1}} \ket{\{(x_j, z_j)\}_{j = 1}^{s}}_{\reg{E_2}} \ket{k}_{\reg{K}}.
    \end{equation*}
    \fi
    Then we apply an isometry to uncompute $k$ in register $\reg{E_1}$, controlled by register $\reg{K}$, and obtain the sub-normalized state 
    \ifllncs
    \begin{equation} \label{eq:PRS_GOOD}
    \begin{split}
        \ket{\psi_{\mathrm{good}}}_{\reg{ABE_1E_2K}} := \sqrt{\frac{1}{2^\lambda N^{\downarrow t+s}}} \sum_{\substack{(x_1, \ldots, x_s) \in [N]^s, \\ (y_1, \ldots, y_t, z_1, \ldots, z_s) \in [N]^{t+s}_{\mathrm{dist}}, \\ k \not\in \set{\Vec{x}}}} \left(\prod_{i = 1}^{s} \ket{z_i}\!\bra{x_i}_{\reg{A}} A^{(i)}_{\reg{AB}} \right) \ket{\psi_\mathrm{Init}(\Vec{y})}_{\reg{AB}} \\
        \otimes \ket{\{(0^n, y_i)\}_{i = 1}^{t}}_{\reg{E_1}} \ket{\{(x_j, z_j)\}_{j = 1}^{s}}_{\reg{E_2}} \ket{k}_{\reg{K}}.
    \end{split}
    \end{equation}
    \else
    $\ket{\psi_{\mathrm{good}}}_{\reg{ABE_1E_2K}}$:
    \begin{equation} \label{eq:PRS_GOOD}
        \sqrt{\frac{1}{2^\lambda N^{\downarrow t+s}}} \sum_{\substack{(x_1, \ldots, x_s) \in [N]^s, \\ (y_1, \ldots, y_t, z_1, \ldots, z_s) \in [N]^{t+s}_{\mathrm{dist}}, \\ k \not\in \set{\Vec{x}}}} \left(\prod_{i = 1}^{s} \ket{z_i}\!\bra{x_i}_{\reg{A}} A^{(i)}_{\reg{AB}} \right) \ket{\psi_\mathrm{Init}(\Vec{y})}_{\reg{AB}} \ket{\{(0^n, y_i)\}_{i = 1}^{t}}_{\reg{E_1}} \ket{\{(x_j, z_j)\}_{j = 1}^{s}}_{\reg{E_2}} \ket{k}_{\reg{K}}.
    \end{equation}
    \fi
    
    \noindent The trace distance between $\Tr_{\reg{EK}}(\ketbra{\psi_{\mathrm{real}}}{\psi_{\mathrm{real}}})$ in $\hybrid~2$ and $\Tr_{\reg{E_1E_2K}}(\ketbra{\psi_{\mathrm{good}}}{\psi_{\mathrm{good}}})$ in $\hybrid~3$ satisfies
    \begin{align}
        & \TD(\Tr_{\reg{EK}}(\ketbra{\psi_{\mathrm{real}}}{\psi_{\mathrm{real}}}_{\reg{ABEK}}), 
        \Tr_{\reg{E_1E_2K}}(\ketbra{\psi_{\mathrm{good}}}{\psi_{\mathrm{good}}}_{\reg{ABE_1E_2K}})) \nonumber \\
        = & \TD(\Tr_{\reg{EK}}(\ketbra{\psi_{\mathrm{real}}}{\psi_{\mathrm{real}}}_{\reg{ABEK}}), 
        \Tr_{\reg{EK}}(\Pi_{\mathrm{good}}\ketbra{\psi_{\mathrm{real}}}{\psi_{\mathrm{real}}}_{\reg{ABEK}}\Pi_{\mathrm{good}})) \nonumber \\
        \leq & \TD(\ketbra{\psi_{\mathrm{real}}}{\psi_{\mathrm{real}}}_{\reg{ABE}\reg{K}},
        \Pi_{\mathrm{good}}\ketbra{\psi_{\mathrm{real}}}{\psi_{\mathrm{real}}}_{\reg{ABE}\reg{K}}\Pi_{\mathrm{good}}) \nonumber \\
        \leq & \sqrt{1 - \norm{\Pi_\mathrm{good}\ket{\psi_{\mathrm{real}}}_{\reg{ABEK}}}^2} \nonumber \\
        = & \sqrt{1 - \norm{\ket{\psi_{\mathrm{good}}}_{\reg{ABE_1E_2K}}}^2}, \label{eq:PRS_hybrid}
    \end{align}
    where the last inequality follows from~\Cref{lem:norm:proj}. Looking ahead, instead of bounding $\norm{\ket{\psi_{\mathrm{good}}}}$ directly, we will show that the inner product between $\ket{\psi_{\mathrm{good}}}$ and some normalized state is negligibly close to $1$. Hence, we can conclude that $\norm{\ket{\psi_{\mathrm{good}}}}$ is also negligibly close to $1$ from~\Cref{lem:norm:cs}. \\
    
    For readability, we define $\hybrid~6$ to $\hybrid~4$ in reverse order. \\
    
    \noindent $\hybrid~6$: Next, we can examine the state of an adversary who got copies of an independently sampled Haar random state, which we can model as getting copies of $V\ket{0}$ for an independently sampled unitary $V$. \\

    \noindent $\hybrid~5$: Similar to $\hybrid~2$, using the path-recording framework, we replace $(U,V)$ with $(\mathsf{PR}_1,\mathsf{PR}_2)$ and define the normalized state 
    \ifllncs
    \begin{equation*}
    \begin{split}
        \ket{\psi^{\mathsf{PR}_1,\mathsf{PR}_2}_{\mathrm{ideal}}}_{\reg{ABE_1E_2}} :=
        \sqrt{\frac{1}{N^{\downarrow t}N^{\downarrow s}}} \sum_{\substack{(x_1, \ldots, x_s) \in [N]^s, \\  (y_1, \ldots, y_t) \in [N]^{t}_{\mathrm{dist}}, \\ (z_1, \ldots, z_s) \in [N]^s_{\mathrm{dist}}}} \left(\prod_{i = 1}^{s} \ket{z_i}\!\bra{x_i}_{\reg{A}} A^{(i)}_{\reg{AB}} \right) \ket{\psi_\mathrm{Init}(\Vec{y})}_{\reg{AB}} \\
        \otimes \ket{\{(0^n, y_i)\}_{i = 1}^{t}}_{\reg{E_1}} \ket{\{(x_j, z_j)\}_{j = 1}^{s}}_{\reg{E_2}}.
    \end{split}
    \end{equation*}
    \else
    \begin{equation*}
        \ket{\psi^{\mathsf{PR}_1,\mathsf{PR}_2}_{\mathrm{ideal}}}_{\reg{ABE_1E_2}} :=
        \sqrt{\frac{1}{N^{\downarrow t}N^{\downarrow s}}} \sum_{\substack{(x_1, \ldots, x_s) \in [N]^s, \\  (y_1, \ldots, y_t) \in [N]^{t}_{\mathrm{dist}}, \\ (z_1, \ldots, z_s) \in [N]^s_{\mathrm{dist}}}} \left(\prod_{i = 1}^{s} \ket{z_i}\!\bra{x_i}_{\reg{A}} A^{(i)}_{\reg{AB}} \right) \ket{\psi_\mathrm{Init}(\Vec{y})}_{\reg{AB}} \ket{\{(0^n, y_i)\}_{i = 1}^{t}}_{\reg{E_1}} \ket{\{(x_j, z_j)\}_{j = 1}^{s}}_{\reg{E_2}}.
    \end{equation*}
    \fi

    From~\Cref{thm:MH24}, the trace distance between $\Tr_{\reg{E_1E_2}}(\ketbra{\psi^{\mathsf{PR}_1,\mathsf{PR}_2}_{\mathrm{ideal}}}{\psi^{\mathsf{PR}_1,\mathsf{PR}_2}_{\mathrm{ideal}}})$ in~$\hybrid~5$ and the adversary's density matrix in~$\hybrid~6$ is $O((t^2+s^2)/2^n)$. \\
    
    \noindent $\hybrid~4$:
    We replace $(\mathsf{PR}_1,\mathsf{PR}_2)$ with $(\pcfpr{1}{n}^{(\reg{E_1})},\pcfpr{1}{n}^{(\reg{E_2})})$ as defined in~\Cref{sec:cfPR} and define the normalized state 
    \ifllncs
    \begin{equation*}
    \begin{split}
        \ket{\psi^{\mathsf{cfPR}_1,\mathsf{cfPR}_2}_{\mathrm{ideal}}}_{\reg{ABE_1E_2}} := 
        \sqrt{\frac{1}{N^{\downarrow t+s}}} \sum_{\substack{(x_1, \ldots, x_s) \in [N]^s, \\  (y_1, \ldots, y_t, z_1, \ldots, z_s) \in [N]^{t+s}_{\mathrm{dist}}}} \left(\prod_{i = 1}^{s} \ket{z_i}\!\bra{x_i}_{\reg{A}} A^{(i)}_{\reg{AB}} \right) \ket{\psi_\mathrm{Init}(\Vec{y})}_{\reg{AB}} \\
        \otimes \ket{\{(0^n, y_i)\}_{i = 1}^{t}}_{\reg{E_1}} \ket{\{(x_j, z_j)\}_{j = 1}^{s}}_{\reg{E_2}}.
    \end{split}
    \end{equation*}
    \else
    \begin{multline*}
        \ket{\psi^{\mathsf{cfPR}_1,\mathsf{cfPR}_2}_{\mathrm{ideal}}}_{\reg{ABE_1E_2}} := \\
        \sqrt{\frac{1}{N^{\downarrow t+s}}} \sum_{\substack{(x_1, \ldots, x_s) \in [N]^s, \\  (y_1, \ldots, y_t, z_1, \ldots, z_s) \in [N]^{t+s}_{\mathrm{dist}}}} \left(\prod_{i = 1}^{s} \ket{z_i}\!\bra{x_i}_{\reg{A}} A^{(i)}_{\reg{AB}} \right) \ket{\psi_\mathrm{Init}(\Vec{y})}_{\reg{AB}} \ket{\{(0^n, y_i)\}_{i = 1}^{t}}_{\reg{E_1}} \ket{\{(x_j, z_j)\}_{j = 1}^{s}}_{\reg{E_2}}.
    \end{multline*}
    \fi
    We apply an isometry to append $\frac{1}{\sqrt{2^\lambda - |\set{\Vec{x}}|}} \sum_{k\notin\set{\Vec{x}}} \ket{k}_\reg{K}$ on the above state, controlled by register $\reg{E}$, and result in the normalized state
    \ifllncs
    \begin{equation}
    \begin{split}
        \ket{\phi^{\mathsf{cfPR}_1,\mathsf{cfPR}_2}_{\mathrm{ideal}}}_{\reg{ABE_1E_2K}} := 
        \sqrt{\frac{1}{N^{\downarrow t+s}}} \sum_{\substack{(x_1, \ldots, x_s) \in [N]^s, \\ (y_1, \ldots, y_t, z_1, \ldots, z_s) \in [N]^{t+s}_{\mathrm{dist}}}} \left(\prod_{i = 1}^{s} \ket{z_i}\!\bra{x_i}_{\reg{A}} A^{(i)}_{\reg{AB}} \right) \ket{\psi_\mathrm{Init}(\Vec{y})}_{\reg{AB}} \\
        \otimes \ket{\{(0^n, y_i)\}_{i = 1}^{t}}_{\reg{E_1}} \ket{\{(x_j, z_j)\}_{j = 1}^{s}}_{\reg{E_2}} 
        \otimes \frac{1}{\sqrt{2^\lambda - |\set{\Vec{x}}|}} \sum_{k\notin\set{\Vec{x}}} \ket{k}_\reg{K}. \label{eq:PRS_IDEAL}
    \end{split}
    \end{equation}
    \else
    \begin{equation}
    \begin{split}
        \ket{\phi^{\mathsf{cfPR}_1,\mathsf{cfPR}_2}_{\mathrm{ideal}}}_{\reg{ABE_1E_2K}} := 
        \sqrt{\frac{1}{N^{\downarrow t+s}}} \sum_{\substack{(x_1, \ldots, x_s) \in [N]^s, \\ (y_1, \ldots, y_t, z_1, \ldots, z_s) \in [N]^{t+s}_{\mathrm{dist}}}} \left(\prod_{i = 1}^{s} \ket{z_i}\!\bra{x_i}_{\reg{A}} A^{(i)}_{\reg{AB}} \right) \ket{\psi_\mathrm{Init}(\Vec{y})}_{\reg{AB}} \\
        \otimes \ket{\{(0^n, y_i)\}_{i = 1}^{t}}_{\reg{E_1}} \ket{\{(x_j, z_j)\}_{j = 1}^{s}}_{\reg{E_2}} \frac{1}{\sqrt{2^\lambda - |\set{\Vec{x}}|}} \sum_{k\notin\set{\Vec{x}}} \ket{k}_\reg{K}. \label{eq:PRS_IDEAL}
    \end{split}
    \end{equation}
    \fi

    By~\Cref{thm:ind:pcfpr:pr}, the trace distance between $\ket{\psi^{\mathsf{cfPR}_1,\mathsf{cfPR}_2}_{\mathrm{ideal}}}$ in $\hybrid~4$ and $\ket{\psi^{\mathsf{PR}_1,\mathsf{PR}_2}_{\mathrm{ideal}}}$ in $\hybrid~5$ is $O((t+s)^2/\sqrt{2^n})$. \\
    
    Finally, we show that the inner product between the sub-normalized state $\ket{\psi_{\mathrm{good}}}$ in $\hybrid~3$ (\Cref{eq:PRS_GOOD}) and the normalized state $\ket{\phi^{\mathsf{cfPR}_1,\mathsf{cfPR}_2}_{\mathrm{ideal}}}$ in $\hybrid~4$ (\Cref{eq:PRS_IDEAL}) is negligibly close to $1$. 
    Notice that 
    \begin{equation*}
    \begin{split}
        \ket{\phi^{\mathsf{cfPR}_1,\mathsf{cfPR}_2}_{\mathrm{ideal}}}_{\reg{ABE_1E_2K}} := 
        \sqrt{\frac{1}{N^{\downarrow t+s}}} \sum_{\substack{(x_1, \ldots, x_s) \in [N]^s, \\ (y_1, \ldots, y_t, z_1, \ldots, z_s) \in [N]^{t+s}_{\mathrm{dist}}}} \left(\prod_{i = 1}^{s} \ket{z_i}\!\bra{x_i}_{\reg{A}} A^{(i)}_{\reg{AB}} \right) \ket{\psi_\mathrm{Init}(\Vec{y})}_{\reg{AB}} \\
        \otimes \ket{\{(0^{n}, y_i)\}_{i = 1}^{t}}_{\reg{E_1}} \ket{\{(x_j, z_j)\}_{j = 1}^{s}}_{\reg{E_2}} \frac{1}{\sqrt{2^\lambda - |\set{\Vec{x}}|}} \sum_{k\notin\set{\Vec{x}}} \ket{k}_\reg{K},
    \end{split}
    \end{equation*}
    and 
    \begin{equation*}
    \begin{split}
        \ket{\psi_{\mathrm{good}}}_{\reg{ABE_1E_2K}} := 
        \sqrt{\frac{1}{N^{\downarrow t+s}}} \sum_{\substack{(x_1, \ldots, x_s) \in [N]^s, \\ (y_1, \ldots, y_t, z_1, \ldots, z_s) \in [N]^{t+s}_{\mathrm{dist}}}}
        \sqrt{\frac{2^\secp - |\set{\Vec{x}}|}{2^{\secp}}}\left(\prod_{i = 1}^{s} \ket{z_i}\!\bra{x_i}_{\reg{A}} A^{(i)}_{\reg{AB}} \right) \ket{\psi_\mathrm{Init}(\Vec{y})}_{\reg{AB}} \\
        \otimes \ket{\{(0^{n}, y_i)\}_{i = 1}^{t}}_{\reg{E_1}} \ket{\{(x_j, z_j)\}_{j = 1}^{s}}_{\reg{E_2}} \frac{1}{\sqrt{2^\lambda - |\set{\Vec{x}}|}} \sum_{k\notin\set{\Vec{x}}} \ket{k}_\reg{K}.
    \end{split}
    \end{equation*}

    \noindent Since for all $(x_1, \ldots, x_s) \in [N]^s$, $$\sqrt{\frac{2^\secp - |\set{\Vec{x}}|}{2^{\secp}}}\geq \sqrt{\frac{2^\secp - s}{2^{\secp}}}.$$
    Hence, by~\Cref{lem:norm:sub-state}, the inner product between $\ket{\psi_{\mathrm{good}}}$ and $\ket{\phi^{\mathsf{cfPR}_1,\mathsf{cfPR}_2}_{\mathrm{ideal}}}$ is $\sqrt{1 - \frac{s}{2^\lambda}}$ and the trace distance between $\Tr_{\reg{EK}}(\ketbra{\psi_{\mathrm{real}}}{\psi_{\mathrm{real}}})$ in $\hybrid~2$ and $\Tr_{\reg{E_1E_2K}}(\ketbra{\psi_{\mathrm{good}}}{\psi_{\mathrm{good}}}))$ in $\hybrid~3$ (\Cref{eq:PRS_hybrid}) is $O(\sqrt{s/2^\lambda})$. \\

    Collecting the bounds, the trace distance between the adversary's density matrix in $\hybrid~1$ and that in $\hybrid~6$ is upper bounded by $O\qty(\sqrt{\frac{s}{2^\secp}} + \frac{(t+s)^2}{\sqrt{2^n}})$ as desired.
\end{proof}

    \section{Adaptively Secure PRFS in the iQHROM}

In this section, we extend the construction of PRS to PRFS. Our construction of a PRFS generator with key length $\secp$, input length $m = m(\secp)$, and output length $n = n(\secp)$ such that $n \geq \lambda + m$ is intuitive and simple. Given oracle access to an $n$-qubit Haar random oracle $U$, on input $w\in\bit^m$ and key $k\in\bit^\lambda$, the PRFS generator $G$ outputs $U \ket{k||w||0^{n-\lambda-m}}$. Namely,
\begin{equation*}
    \oracle_{\mathsf{PRFS}}(k,\cdot): \ket{w} \ket{0^n} \mapsto \ket{w} \otimes U \ket{k||w||0^{n-\lambda-m}}.
\end{equation*}

\noindent We will prove that $G$ is secure against any adversary that makes arbitrary polynomial adaptive quantum queries to $U$ and arbitrary polynomial adaptive \emph{classical} queries to $\oracle$, where $\oracle$ is either the PRFS construction $\oracle_{\mathsf{PRFS}}$ or a Haar random state generator $\oracle_{\mathsf{Haar}}$ defined in~\Cref{def:prfs}.

\begin{theorem} \label{thm:prfs}
$G$ is an APRFS generator in the $\mathsf{iQHROM}$ if $n$ and $m$ satisfy $n \geq \lambda + m$. Formally, for any $t\in\N$ and adversary that makes $t$ adaptive quantum queries to $U$ and $t$ adaptive classical queries to $\oracle$, the distinguishing advantage of $\Adversary$ is at most $O\qty(\frac{t^2}{2^{n-m}} + \frac{t^2}{\sqrt{2^n}} + \sqrt{\frac{t}{2^{\secp}}})$.
\end{theorem}
\begin{proof}
Fix $\secp,m,n,t\in\N$ and let $M:=2^m$ and $N:=2^n$. We assume that $\Adversary$ alternatively makes queries to $U$ and $\oracle$. That is, the adversary makes all odd-indexed queries to $U$ and all even-indexed queries to $\oracle$. This is without loss of generality, as we can pad dummy queries and will later prove that the asymptotic bound remains unchanged. 

From~\Cref{def:prfs}, we need to show that classical query access to $\oracle_{\mathsf{PRFS}}(k,\cdot)$ is indistinguishable from classical query access to $\oracle_{\mathsf{Haar}}$. Here ${\cal O}_{\sf Haar}(\cdot)$, on input $x \in \{0,1\}^{m(\secparam)}$, outputs $\ket{\vartheta_x}$, where, for every $w \in \{0,1\}^{m(\secparam)}$, $\ket{\vartheta_w} := U_w\ket{0^n}$ and $\set{U_w}_{w\in\bit^m}$ is a set of i.i.d.\,$n$-qubit Haar unitaries. Notice that by the unitary invariance of the Haar measure, we can equivalently define $\ket{\vartheta_w} := U_w\ket{\overline{w}}$, where $\overline{w} := 0^{\secp}||w||0^{n-\secp-m}$. 

To model that adversary $\Adversary$ makes classical queries, we introduce the transcript register $\reg{Q} = (\reg{Q}_1,\dots,\reg{Q}_t)$ where $\reg{Q}_i$ stores $\Adversary$'s $i$-th classical query. We allow $\Adversary$ to send the input register $\reg{A}$, but then register $\reg{A}$ is immediately measured in the computational basis. Moreover, the adversary $\Adversary$ does not have access to register $\reg{Q}$. By the deferred measurement principle, equivalently, we can assume that right after $\Adversary$ makes a classical query, the content in $\reg{A}$ is copied to $\reg{Q}$. At the end, register $\reg{Q}$ is measured in the computational basis. Further note that whenever $\Adversary$ queries on an input, the oracle appends an $n$-qubit answer register $\reg{D}$ to the system. See~\Cref{fig:PRFS} for an exposition.\footnote{Register $\reg{A}$ contains the first $m$ qubits of $\Adversary$, register $\reg{B}$ contains the $(m+1)$-th qubit to the $n$-th qubit of $\Adversary$, and register $\reg{C}$ contains all the other qubits of $\Adversary$. Recall that $\oracle$ is an isometry with an input length $m$ and an output length $m+n$. Hence, the number of $\Adversary$'s qubits increases by $n$ after the $i$-th query to $\oracle$ for all $i \in [t]$, and the next internal unitary $A_{2i-1}$ is defined over the extended space. Although this is not reflected in~\Cref{fig:PRFS}, it should not cause confusion.} \\

\begin{figure}[h!]
\centering
\resizebox{0.8\textwidth}{!}{
    \begin{tikzpicture}

    \tikzset{meter/.append style={draw, inner sep=10, rectangle, font=\vphantom{A}, minimum width=30, line width=.8,
    path picture={\draw[black] ([shift={(.1,.3)}]path picture bounding box.south west) to[bend left=50] ([shift={(-.1,.3)}]path picture bounding box.south east);\draw[black,-latex] ([shift={(0,.1)}]path picture bounding box.south) -- ([shift={(.3,-.1)}]path picture bounding box.north);}}}

    \draw (0,4) node[left] {$\ket{\bot}_{\reg{Q}_t}$} -- (6,4) ;
    \draw (0,2) node[left] {$\ket{\bot}_{\reg{Q}_1}$} -- (6,2) ;
    \draw (0,0) node[left] {$\ket{0}_{\reg{A}}$} -- (6,0) ;
    \draw (0,-1) node[left] {$\ket{0}_{\reg{B}}$} -- (6,-1) ;
    \draw (0,-2) node[left] {$\ket{0}_{\reg{C}}$} -- (6,-2) ;

    \draw[fill=white] (0.5,-2.5) rectangle (1.5,0.5) node[midway] {$A_1$};

    \draw[fill=white] (2,-1.5) rectangle (3,0.5) node[midway] {$U$};

    \draw[fill=white] (3.5,-2.5) rectangle (4.5,0.5) node[midway] {$A_2$};

    \draw (5,2) -- (5,0);
    \filldraw (5,0) circle (2pt); 
    \draw (5,2) circle (5pt); 
    \draw (4.8,2) -- (5.2,2); 
    \draw (5,2.2) -- (5,1.8); 

    \draw[fill=white] (5.5,-0.5) rectangle (6.5,0.5) node[midway] {$\oracle$};

    \node at (7,1) {$\cdots$};

    \draw[dotted] (6.5,4) -- (8,4);
    \draw[dotted] (6.5,2) -- (8,2);
    \draw[dotted] (6.5,0) -- (8,0);
    \draw[dotted] (6.5,-1) -- (8,-1);
    \draw[dotted] (6.5,-2) -- (8,-2);

    \node at (1.75, 3) {$\vdots$};

    \draw (8.5,4) -- (16,4);
    \draw (8.5,2) -- (16,2);
    \draw (8.5,0) -- (16,0);
    \draw (8.5,-1) -- (16,-1);
    \draw (8.5,-2) -- (16,-2);

    \draw[fill=white] (9,-2.5) rectangle (10,0.5) node[midway] {$A_{2t-1}$};

    \draw[fill=white] (10.5,-1.5) rectangle (11.5,0.5) node[midway] {$U$};

    \draw[fill=white] (12,-2.5) rectangle (13,0.5) node[midway] {$A_{2t}$};

    \draw (13.5,4) -- (13.5,0);
    \filldraw (13.5,0) circle (2pt); 
    \draw (13.5,4) circle (5pt); 
    \draw (13.3,4) -- (13.7,4); 
    \draw (13.5,4.2) -- (13.5,3.8); 

    \draw[fill=white] (14,-0.5) rectangle (15,0.5) node[midway] {$\oracle$};

    \draw[fill=white] (14.5,1.5) rectangle (15.5,2.5) node[midway] {};
    \node[meter] at (15,2) {};

    \draw[fill=white] (14.5,3.5) rectangle (15.5,4.5) node[midway] {};
    \node[meter] at (15,4) {};
    
    \end{tikzpicture}
}
\caption{Modeling classical queries to $\oracle$.} 
\label{fig:PRFS}
\end{figure}
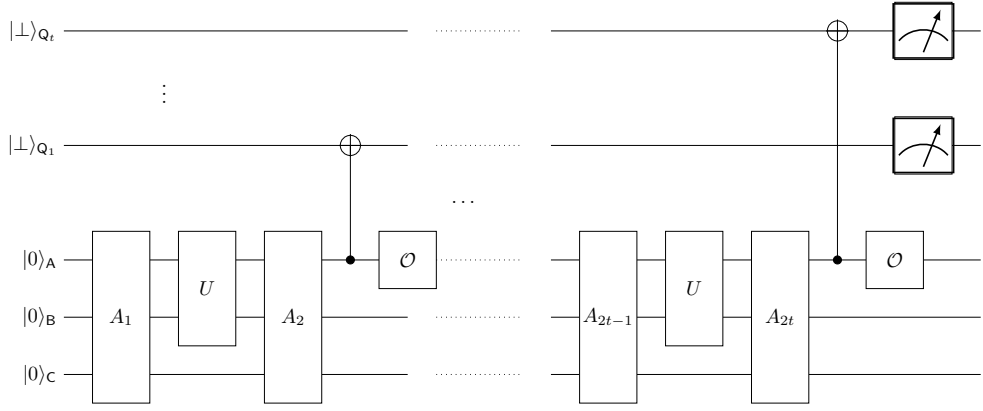

\noindent For ease of notation, going forward we define $$\oracle_{\mathsf{PRFS}}(k,\cdot):\ket{w}_{\reg{A}}\mapsto \ket{w}_{\reg{A}} \otimes U\ket{k||w||0^{n-\secp-m}}_{\reg{D}}, $$
and 
$$\oracle_{\mathsf{Haar}}: \ket{w}_{\reg{A}}\mapsto \ket{w}_{\reg{A}} \otimes U_w\ket{\overline{w}}_{\reg{D}}.$$
Here for the $i$-th query, the register appended is called $\reg{D}_i$ and is appended to the system.

\noindent We refer to the experiment as $\mathsf{Ideal}$ if $\oracle = \oracle_{\mathsf{Haar}}$ and as $\mathsf{Real}$ if $\oracle = \oracle_{\mathsf{PRFS}}$, respectively. \\

\noindent We complete the proof by hybrid arguments. \\

\noindent $\bullet~\mathsf{Ideal}$: The adversary gets access to the Quantum Haar Random Oracle $U$ and $\oracle_{\mathsf{Haar}}$. \\

\noindent $\bullet~\hybrid~1$: The adversary gets access to a path recording isometry $\mathsf{PR}_{\reg{ABE}_1}$ and $\oracle_{\mathsf{Haar}}$.

\begin{lemma} \label{lem:PRFS_Ideal_Hyb1}
    $\qty|\Pr_{\hybrid\,1}[1 \gets \Adversary^{\mathsf{PR}_{\reg{ABE}_1},\oracle_{\mathsf{Haar}}}] - \Pr_{\mathsf{Ideal}}[1 \gets \Adversary^{U,\oracle_{\mathsf{Haar}}}]| = O\qty(\frac{t^2}{2^n})$.
\end{lemma}
\begin{proof}[Proof of~\Cref{lem:PRFS_Ideal_Hyb1}]
    This is true by~\Cref{thm:MH24}.
\end{proof}

\noindent Before defining $\hybrid~2.i$, we define the path-recording isometries indexed by $w\in\bit^m$, $$\mathsf{PR}^w_{\reg{DE}^w_2}:\ket{x}_{\reg{D}}\ket{R_1}_{\reg{E}_1}\bigotimes_{w'\in\bit^m}\ket{R^{w'}_2}_{\reg{E}^{w'}_2}\mapsto
\frac{1}{\sqrt{N-|R^w_2|}}\sum_{z\in[N]\setminus\Image(R^{w}_2)}\ket{z}_{\reg{D}}\ket{R_1}_{\reg{E}_1}\ket{R^{w}_2\cup (\overline{x},z)}_{\reg{E}^{w}_2}\bigotimes_{\substack{w'\in\bit^m\\ w'\neq w}}\ket{R^{w'}_2}_{\reg{E}^{w'}_2}.\footnote{Here we assume that each $\reg{E}^{w}_2$ to be of the size $2n2^n$, with the first $|R^{w'}_2|$ qubits to hold $R^{w'}_2$ and the rest $2n2^n - |R^{w'}_2|$ to be $\bot$. Whenever we add something to $R^{w'}_2$, we update the number of $\bot$ such that the total size of $\reg{E}^{w}_2$ remains $2n2^n$.}$$

\noindent Hence, for $i\in\bit^m$, we define $\oracle_{\mathsf{Hyb}2.i}$ as 
\begin{equation*}
    \oracle_{\mathsf{Hyb}2.i}:
    \ket{w}_A\ket{R_1}_{\reg{E}_1}\bigotimes_{w'\in\bit^m}\ket{R^{w'}_2}_{\reg{E}^{w'}_2}\mapsto 
    \begin{cases}
        \ket{w}_{\reg{A}}\mathsf{PR}^w_{\reg{DE}^w_2}\ket{w}_{\reg{D}}\ket{R_1}_{\reg{E}_1}\bigotimes_{w'\in\bit^m}\ket{R^{w'}_2}_{\reg{E}^{w'}_2} &\text{~, if }w\leq i\\
        
        \ket{w}_{\reg{A}}\ket{\vartheta_w}_{\reg{D}}\ket{R_1}_{\reg{E}_1}\bigotimes_{w'\in\bit^m}\ket{R^{w'}_2}_{\reg{E}^{w'}_2} &\text{~, otherwise,}\\
        
    \end{cases}
\end{equation*}
where, for every $w \in \{0,1\}^{m(\secparam)}$, $\ket{\vartheta_w}  = U_w\ket{\overline{w}}$ and $U_w$ is an i.i.d.\,$n$-qubit Haar unitary.

\noindent $\bullet~\hybrid~2.i$ for $0 \leq i \leq M$: The adversary gets access to a path recording isometry $\mathsf{PR}_{\reg{ABE}_1}$ and $\oracle_{\mathsf{Hyb}2.i}$. We assume that the initial state is $$\ket{0}_{\reg{ABC}} 
    \ket{\bot,\dots,\bot}_{\reg{Q}} 
    \ket{\set{}}_{\reg{E}_1} \bigotimes_{w\in\bit^m}\ket{\set{}}_{\reg{E}^w_2},$$
    where registers $\reg{E}_1,\reg{E}_2$ are purifying registers.

\begin{lemma} \label{lem:PRFS_Hyb1_Hyb2.0}
    $\Pr_{\hybrid\,2.0}[1 \gets \Adversary^{\mathsf{PR}_{\reg{ABE}_1},\oracle_{\mathsf{Hyb}2.0}}] = \Pr_{\hybrid\,1}[1 \gets \Adversary^{\mathsf{PR}_{\reg{ABE}_1},\oracle_{\mathsf{Haar}}}].$
\end{lemma}

\begin{proof}[Proof of~\Cref{lem:PRFS_Hyb1_Hyb2.0}]
    Note that $\oracle_{\mathsf{Hyb}2.0}$ and $\oracle_{\mathsf{Haar}}$ are identically distributed. Hence the above lemma holds.
\end{proof}
\begin{lemma} \label{lem:PRFS_Hyb2i_i+1}
    For all $0 \leq i < M$, $\qty|\Pr_{\hybrid\,2.i}[1 \gets \Adversary^{\mathsf{PR}_{\reg{ABE}_1},\oracle_{\mathsf{Hyb}2.i}}] - \Pr_{\hybrid\,2.i+1}[1 \gets \Adversary^{\mathsf{PR}_{\reg{ABE}_1},\oracle_{\mathsf{Hyb}2.i+1}}]| = O\qty(\frac{t^2}{2^n})$.
\end{lemma}

\begin{proof}[Proof of~\Cref{lem:PRFS_Hyb2i_i+1}]
    Let there is an adversary $\Adversary$ such that $$\qty|\Pr_{\hybrid\,2.i}[1 \gets \Adversary^{\mathsf{PR}_{\reg{ABE}_1},\oracle_{\mathsf{Hyb}2.i}}] - \Pr_{\hybrid\,2.i+1}[1 \gets \Adversary^{\mathsf{PR}_{\reg{ABE}_1},\oracle_{\mathsf{Hyb}2.i+1}}]| = \eps,$$ for some $\eps\in [0,1]$. We give the following (information-theoretic) reduction $\cR$ given access to an oracle $\cP$ that is either an $n$-qubit Haar random oracle $U$ or a path-recoding oracle $\mathsf{PR}$ such that $$\qty|\Pr[1 \gets \cR^{\mathsf{PR},\Adversary}] - \Pr[1 \gets \cR^{U,\Adversary}]| = \eps.$$ 
    
    \noindent First, $\cR$ initializes $\reg{Q}$, $\reg{E}_1$ and $\reg{E}^w_2$ for all $w\in\bit^m, w\neq i+1$. Next, $\cR$ samples i.i.d. $n$-qubit Haar random unitaries $U_w$ for $i+2 \leq w \leq M$. 
    \begin{itemize}
        \item Whenever $\Adversary$ queries $\mathsf{PR}_{\reg{ABE}_1}$, $\cR$ locally simulates an independent path-recording oracle to respond to $\Adversary$'s quantum queries to $\mathsf{PR}$ using $\reg{E}_1$.
        \item To respond to $\Adversary$'s classical queries, the reduction $\cR$ locally simulates $\mathsf{PR}^w_{\reg{AE}_2^w}$ for $1\leq w \leq i$, embeds $\cP$ as the $(i+1)$-th oracle, and returns with $U_w\ket{\overline{w}}$, for $i+2\leq w\leq M$.
    \end{itemize}  
    By the above, $\cR$ perfectly simulates $\Adversary$, hence $$\qty|\Pr[1 \gets \cR^{\mathsf{PR},\Adversary}] - \Pr[1 \gets \cR^{U,\Adversary}]| = \eps.$$ 
    By~\Cref{thm:MH24}, it holds that $\eps = O\qty(\frac{t^2}{2^n})$. Hence, $\qty|\Pr_{\hybrid\,2.i}[1 \gets \Adversary^{\mathsf{PR},\oracle_{\mathsf{Hyb}2.i}}] - \Pr_{\hybrid\,2.i+1}[1 \gets \Adversary^{\mathsf{PR},\oracle_{\mathsf{Hyb}2.i+1}}]| = O\qty(\frac{t^2}{2^n})$ as desired.
\end{proof}

\noindent Before defining $\hybrid~3$, we define 
$$\oracle_{\mathsf{Hyb}3}:\ket{w}_{\reg{A}}\ket{R_1}_{\reg{E}_1}\bigotimes_{w'\in\bit^m}\ket{R^{w'}_2}_{\reg{E}^{w'}_2}\mapsto 
\ket{w}_{\reg{A}}\mathsf{PR}^w_{\reg{DE}^w_2}\ket{w}_{\reg{D}}\ket{R_1}_{\reg{E}_1}\bigotimes_{w'\in\bit^m}\ket{R^{w'}_2}_{\reg{E}^{w'}_2}.$$

\noindent $\bullet~\hybrid~3$: The adversary gets access to a path recording isometry $\mathsf{PR}_{\reg{ABE}_1}$ and $\oracle_{\mathsf{Hyb}3}$.

\begin{lemma} \label{lem:PRFS_Hyb2M_Hyb3}
    $\Pr_{\hybrid\,3}[1 \gets \Adversary^{\mathsf{PR}_{\reg{ABE}_1},\oracle_{\mathsf{Hyb}3}}] = \Pr_{\hybrid\,2.M}[1 \gets \Adversary^{\mathsf{PR}_{\reg{ABE}_1},\oracle_{\mathsf{Hyb}2.M}}].$
\end{lemma}

\begin{proof}[Proof of~\Cref{lem:PRFS_Hyb2M_Hyb3}]
    Note that $\oracle_{\mathsf{Hyb}2.M}$ and $\oracle_{\mathsf{Hyb}3}$ are identical. Hence the above lemma holds.
\end{proof}

\noindent Before defining $\hybrid~4$, we define oracles $\mathsf{cfPR}$ and $\oracle_{\mathsf{Hyb}4}$ that maintain the ``global'' injectivity of all purifying registers:
\begin{equation*}
\begin{split}
& \mathsf{cfPR} \ket{x}_{\reg{AB}} \ket{R_1}_{\reg{E}_1} \bigotimes_{w'\in\bit^m}\ket{R^{w'}_2}_{\reg{E}^{w'}_2} := \\ 
& \frac{1}{\sqrt{N - \qty|\Im\left( R_1 \cup \bigcup_{w'\in\bit^m}R^{w'}_2 \right)|}} 
\sum_{\substack{ y \in [N]: \\ y \notin \Im\left( R_1 \cup \bigcup_{w'\in\bit^m}R^{w'}_2 \right) }} 
\ket{y}_{\reg{AB}} 
\otimes \ket{R_1 \cup \set{(x,y)}}_{\reg{E}_1} \bigotimes_{w'\in\bit^m}\ket{R^{w'}_2}_{\reg{E}^{w'}_2}
\end{split}
\end{equation*}
and
\begin{equation*}
\begin{split}
& \oracle_{\mathsf{Hyb}4} \ket{w}_{\reg{A}} \ket{R_1}_{\reg{E}_1} \bigotimes_{w'\in\bit^m}\ket{R^{w'}_2}_{\reg{E}^{w'}_2} := \\ 
& \frac{1}{\sqrt{N - \qty|\Im\left( R_1 \cup \bigcup_{w'\in\bit^m}R^{w'}_2 \right)|}} 
\sum_{\substack{ y \in [N]: \\ y \notin \Im\left( R_1 \cup \bigcup_{w'\in\bit^m}R^{w'}_2 \right) }} 
\ket{w}_{\reg{A}}\ket{y}_{\reg{D}} \ket{R_1}_{\reg{E}_1} \ket{R^{w}_2\cup (\overline{w},y)}_{\reg{E}^{w}_2}\bigotimes_{\substack{w'\in\bit^m\\ w'\neq w}}\ket{R^{w'}_2}_{\reg{E}^{w'}_2}.
\end{split}
\end{equation*}

\noindent $\bullet~\hybrid~4$: The adversary gets access to a path recording isometry $\mathsf{cfPR}$ and $\oracle_{\mathsf{Hyb}4}$.

\begin{lemma} \label{lem:PRFS_Hyb3_Hyb4}
    $\qty|\Pr_{\hybrid\,4}[1 \gets \Adversary^{\mathsf{cfPR}_{\reg{ABE}_1},\oracle_{\mathsf{Hyb}4}}] - \Pr_{\hybrid\,3}[1 \gets \Adversary^{\mathsf{PR}_{\reg{ABE}_1},\oracle_{\mathsf{Hyb}3}}]| = O\qty(\frac{t^2}{\sqrt{2^n}})$.
\end{lemma}

\begin{proof}[Proof of~\Cref{lem:PRFS_Hyb3_Hyb4}]
    By a similar hybrid as in~\Cref{thm:ind:pcfpr:pr} and calculation of inner product as in~\Cref{lem:pcfpr:hybrids}, we have $\qty|\Pr_{\hybrid\,4}[1 \gets \Adversary^{\mathsf{cfPR}_{\reg{ABE}_1},\oracle_{\mathsf{Hyb}4}}] - \Pr_{\hybrid\,3}[1 \gets \Adversary^{\mathsf{PR}_{\reg{ABE}_1},\oracle_{\mathsf{Hyb}3}}]| = O\qty(\frac{t^2}{\sqrt{2^n}})$.
\end{proof}

For ease of notation, we assume that $A_{2i+1}$ maps $\reg{ABCD}_{i}$ to $\reg{ABC}$ where $\reg{C}$ is updated to be increased by the size of $\reg{D}_i$. \\

\noindent Hence, the final state in $\hybrid~4$ is
\begin{multline} \label{eq:PRFS_PRw}
\ket{\Adversary^{\mathsf{cfPR},\oracle_{\mathsf{Hyb}4}}_{t,t}} 
:= \sum_{\substack{\vec{x}\in[N]^t, \vec{w}\in[M]^t, \\ (y_1,\dots,y_t,z_1,\dots,z_t) \in [N]^{2t}_\dist}}
\frac{1}{\sqrt{N^{\downarrow 2t}}} 
\qty(\prod_{i=1}^t \ketbra{w_i}{w_i}_{\reg{A}} \otimes \ket{z_i}_{\reg{D}_i} \cdot A_{2i} \ketbra{y_i}{x_i}_{\reg{AB}} A_{2i-1})
\ket{0}_{\reg{ABC}} \\
\otimes \ket{\vec{w}}_{\reg{Q}} \ket{\set{(x_j,y_j)}_{j\in[t]}}_{\reg{E_1}} \bigotimes_{w'\in\bit^m} \ket{\set{(\overline{w}_i,z_i):w_i=w'}}_{\reg{E}_2^{w'}}.
\end{multline}

\noindent Finally, $\hybrid~5$ is the same as $\mathsf{Real}$ except that $U$ is replaced with $\mathsf{PR}$. \\

\noindent $\bullet~\hybrid~5$: $\Adversary$ makes $t$ quantum queries to $\mathsf{PR}_{\reg{ABE}}$ and $t$ classical queries to $\oracle_{\mathsf{Hyb}5}$ defined as follows.
\begin{itemize}
    \item Initial state: $\ket{0}_{\reg{ABC}} \ket{\bot,\dots,\bot}_{\reg{Q}} \ket{\set{}}_{\reg{E}} \otimes \frac{1}{\sqrt{2^\lambda}}\sum_{k \in \bit^\lambda} \ket{k}_{\reg{K}}$.
    \item Oracle $\mathsf{PR}_{\reg{ABE}}$: 
    \begin{equation*}
        \ket{x}_{\reg{AB}} \ket{R}_{\reg{E}} 
        \mapsto \frac{1}{\sqrt{N - |\Im(R)|}}\sum_{\substack{y\in[N]: \\ y \notin \Im(R) }}
        \ket{y}_{\reg{AB}} \ket{R\cup\set{(x,y)}}_{\reg{E}}.
    \end{equation*}
    \item Oracle $\oracle_{\mathsf{Hyb}5}$:
    \begin{equation*}
    \ket{w}_{\reg{A}}\ket{R}_{\reg{E}} \ket{k}_{\reg{K}}\mapsto
    \ket{w}_{\reg{A}}\otimes \frac{1}{\sqrt{N - |\Im(R)|}}\sum_{\substack{z\in[N]: \\ z \notin \Im(R) }}
        \ket{z}_{\reg{D}_i} \ket{R\cup\set{(k||w||0^{n-\lambda-m},z)}}_{\reg{E}} \ket{k}_{\reg{K}}.
    \end{equation*}
\end{itemize}

\noindent Hence, the final state in $\hybrid~5$ is
\begin{multline} \label{eq:PRFS_PR}
\ket{\Adversary^{\mathsf{PR},\oracle_{\mathsf{Hyb}5}}_{t,t}} 
= \sum_{
\substack{
k\in\bit^\lambda, \\ \vec{x}\in[N]^t, \vec{w}\in[M]^t, \\ 
(y_1,\dots,y_t,z_1,\dots,z_t) \in [N]^{2t}_\dist
}
}
\frac{1}{\sqrt{2^\lambda \cdot N^{\downarrow 2t}}} 
\qty( \prod_{i=1}^t \ketbra{w_i}{w_i}_\reg{A} \otimes \ket{z_i}_{\reg{D}_i} \cdot 
A_{2i} \ketbra{y_i}{x_i}_{\reg{AB}} A_{2i-1} )
\ket{0}_{\reg{ABC}} \\
\otimes \ket{\vec{w}}_{\reg{Q}} \ket{\set{(x_j,y_j),(k||w_j||0^{n-\lambda-m},z_j)}_{j\in[t]}}_{\reg{E}} \ket{k}_{\reg{K}}.
\end{multline}

\begin{lemma} \label{lem:PRFS_Hyb4_Hyb5}
    $\qty|\Pr_{\hybrid\,5}[1 \gets \Adversary^{\mathsf{PR},\oracle_{\mathsf{Hyb}5}}] - \Pr_{\hybrid\,4}[1 \gets \Adversary^{\mathsf{cfPR}_{\reg{ABE}_1},\oracle_{\mathsf{Hyb}4}}]|
    = O\qty(\sqrt{\frac{t}{2^{\secp}}})$.
\end{lemma}
\begin{proof}[Proof of~\Cref{lem:PRFS_Hyb4_Hyb5}]

\noindent The structure of the proof is similar to that of~\Cref{thm:multi_copy_prs}. We will show that $\ket{\Adversary^{\mathsf{cfPR},\oracle_{\mathsf{Hyb}4}}_{t,t}}$ (\Cref{eq:PRFS_PRw}) and $\ket{\Adversary^{\mathsf{PR},\oracle_{\mathsf{Hyb}5}}_{t,t}}$ (\Cref{eq:PRFS_PR}) are negligibly close after partially tracing out their purifying registers. Specifically, we will define a projector $\Pi^{\mathsf{good}}$ and isometries $V^{\mathsf{part}}, V^{\mathsf{func},\oplus}, V^{\mathsf{split}}$ and $W^{\mathsf{AppK}}$ such that 
\begin{equation*}
    W^{\mathsf{AppK}} \ket{\Adversary^{\mathsf{cfPR},\oracle_{\mathsf{Hyb}4}}_{t,t}} 
    \approx V^{\mathsf{split}} \cdot V^{\mathsf{func},\oplus} \cdot V^{\mathsf{part}} \cdot \Pi^{\mathsf{good}} \ket{\Adversary^{\mathsf{PR},\oracle_{\mathsf{Hyb}5}}_{t,t}}.
\end{equation*}

\noindent We say that a pair of an injective relation and a key $(R,k) \in \rel_{2t}^{\inj} \times \bit^\secp$ is \emph{good} if
\begin{equation*}
    |\set{(x,y)\in R: x_{[1:\lambda]} = k}| = t,
\end{equation*}
where $x_{[1:\lambda]}$ denotes the first $\lambda$ bits of $x\in\bit^n$. For any good $(R,k)$, we denote $R_k := \set{(x,y)\in R: x_{[1:\lambda]} = k}$. Define the projector $\Pi^{\mathsf{good}}_{\reg{EK}}$ onto the subspace spanned by $\set{\ket{R}_{\reg{E}}\ket{k}_{\reg{K}}}$ for all good $(R,k)$. Then we have 

\begin{multline} \label{eq:PRFS_PR_PROJ}
\Pi^{\mathsf{good}} \ket{\Adversary^{\mathsf{PR},\oracle_{\mathsf{Hyb}5}}_{t,t}}
= \sum_{
\substack{
\vec{x}\in[N]^t, \vec{w}\in[M]^t, \\ 
(y_1,\dots,y_t,z_1,\dots,z_t) \in [N]^{2t}_\dist
}
}
\frac{1}{\sqrt{N^{\downarrow 2t}}} 
\qty( \prod_{i=1}^t \ketbra{w_i}{w_i}_{\reg{A}} \otimes \ket{z_i}_{\reg{D}_i} A_{2i} \ketbra{y_i}{x_i}_{\reg{AB}} A_{2i-1} )
\ket{0}_{\reg{ABC}}
\ket{\vec{w}}_{\reg{Q}} \\
\otimes \frac{1}{\sqrt{2^\lambda}} \sum_{k \notin \set{\vec{x}_{[1:\lambda]}} }
\ket{\set{(x_j,y_j),(k||w_j||0^{n-\lambda-m},z_j)}_{j\in[t]}}_{\reg{E}} \ket{k}_{\reg{K}},
\end{multline}
where $\set{\vec{x}_{[1:\lambda]}}$ denotes the union of the first $\lambda$ bits of all coordinates of $\vec{x}$, \ie $\set{\vec{x}_{[1:\lambda]}} := \bigcup_{i\in[t]} \set{(x_i)_{[1:\lambda]}} \subseteq \bit^\lambda$ for any $\vec{x} = (x_1,\dots,x_t) \in [N]^t$. \\

\noindent By~\Cref{lem:norm:proj,lem:norm:sub-state}, it is sufficient to show the closeness between \Cref{eq:PRFS_PRw} and \Cref{eq:PRFS_PR_PROJ} up to isometries. First, we define the following partial isometry on all good $(R,k)$
\begin{equation*}
    V^{\mathsf{part}}: \ket{R}_{\reg{E}} \ket{k}_{\reg{K}}
    \mapsto \ket{R \setminus R_k}_{\reg{E}_1} \ket{R_k}_{\reg{E}_2} \ket{k}_{\reg{K}}.
\end{equation*}
Then applying $V^{\mathsf{part}}$ on~\Cref{eq:PRFS_PR_PROJ}, we have
\begin{multline}    \label{eq:part_good_hyb5}
V^{\mathsf{part}} \Pi^{\mathsf{good}} \ket{\Adversary^{\mathsf{PR},\oracle_{\mathsf{Hyb}5}}_{t,t}} = 
\sum_{
\substack{
\vec{x}\in[N]^t, 
\vec{w}\in[M]^t, \\ 
(y_1,\dots,y_t,z_1,\dots,z_t) \in [N]^{2t}_\dist
}
}
\frac{1}{\sqrt{N^{\downarrow 2t}}} 
\qty( \prod_{i=1}^t \ketbra{w_i}{w_i}_{\reg{A}} \otimes \ket{z_i}_{\reg{D}_i} \cdot A_{2i} \ketbra{y_i}{x_i}_{\reg{AB}} A_{2i-1} )
\ket{0}_{\reg{ABC}} 
\ket{\vec{w}}_{\reg{Q}} \\
\otimes 
\frac{1}{\sqrt{2^\lambda}} 
\sum_{k \notin \set{\vec{x}_{[1:\lambda]}} }
\ket{\set{(x_j,y_j)}_{j\in[t]}}_{\reg{E}_1}
\otimes
\ket{\set{(k||w_\ell||0^{n-\lambda-m},z_\ell)}_{\ell\in[t]}}_{\reg{E}_2} 
\ket{k}_{\reg{K}}.
\end{multline}

\noindent Next, define a partial isometry isometry $V^{\mathsf{func},\oplus}$ on all $\vec{w}\in[M]^t, \vec{z}\in[N]^t_\dist$ and $k\in\bit^\lambda$ such that
\begin{equation*}
    V^{\mathsf{func},\oplus}: 
    \ket{\set{(k||w_\ell||0^{n-\lambda-m},z_\ell)}_{\ell\in[t]}}_{\reg{E}_2} \ket{k}_{\reg{K}}
    \mapsto \ket{\set{(\overline{w_\ell},z_\ell)}_{\ell\in[t]}}_{\reg{E}_2} \ket{k}_{\reg{K}}.
\end{equation*}

\noindent Applying $V^{\mathsf{func},\oplus}$ to~\Cref{eq:part_good_hyb5} to uncompute $k$ from each $k||w_\ell||0^{n-\lambda-m}$ on register $\reg{E}_2$, we have
\begin{multline} \label{eq:PRFS_SPLIT_1}
V^{\mathsf{func},\oplus} V^{\mathsf{part}} \Pi^{\mathsf{good}} \ket{\Adversary^{\mathsf{PR},\oracle_{\mathsf{Hyb}5}}_{t,t}} = 
\sum_{
\substack{
\vec{x}\in[N]^t, 
\vec{w}\in[M]^t, \\ 
(y_1,\dots,y_t,z_1,\dots,z_t) \in [N]^{2t}_\dist
}
}
\frac{1}{\sqrt{N^{\downarrow 2t}}} 
\qty( \prod_{i=1}^t \ketbra{w_i}{w_i}_{\reg{A}} \otimes \ket{z_i}_{\reg{D}_i} \cdot
A_{2i} \ketbra{y_i}{x_i}_{\reg{AB}} A_{2i-1} )
\ket{0}_{\reg{ABC}} \\
\otimes 
\ket{\vec{w}}_{\reg{Q}} 
\ket{\set{(x_j,y_j)}_{j\in[t]}}_{\reg{E}_1}
\ket{\set{(\overline{w_\ell},z_\ell)}_{\ell\in[t]}}_{\reg{E}_2}
\otimes 
\frac{1}{\sqrt{2^\lambda}} 
\sum_{k \notin \set{\vec{x}_{[1:\lambda]}} }
\ket{k}_{\reg{K}}.
\end{multline}

\noindent Define a partial isometry $V^{\mathsf{split}}$ defined on all $\vec{w}\in[M]^t$ and $\vec{z}\in[N]^t_\dist$ such that
\begin{equation*}
\begin{split}
V^{\mathsf{split}}: \ket{\set{(\overline{w_i},z_i)}_{i\in[t]}}_{\reg{E}_2} 
\mapsto \bigotimes_{w'\in\bit^m} \ket{\set{(\overline{w}_i,z_i):w_i=w'}}_{\reg{E}_2^{w'}}.
\end{split}
\end{equation*}

\noindent Applying $V^{\mathsf{split}}$ to~\Cref{eq:PRFS_SPLIT_1}, we have the following subnormalized state $\ket{\psi_5}$,
\begin{multline}
\ket{\psi_5} =
\sum_{
\substack{
\vec{x}\in[N]^t, 
\vec{w}\in[M]^t, \\ 
(y_1,\dots,y_t,z_1,\dots,z_t) \in [N]^{2t}_\dist}
}
\frac{1}{\sqrt{N^{\downarrow 2t}}}
\qty( \prod_{i=1}^t \ketbra{w_i}{w_i}_{\reg{A}} \otimes \ket{z_i}_{\reg{D}_i}
\cdot A_{2i} \ketbra{y_i}{x_i}_{\reg{AB}} A_{2i-1} )
\ket{0}_{\reg{ABC}} \\
\otimes 
\ket{\vec{w}}_{\reg{Q}} 
\ket{\set{(x_j,y_j)}_{j\in[t]}}_{\reg{E}_1}
\bigotimes_{w'\in\bit^m} \ket{\set{(\overline{w}_i,z_i):w_i=w'}}_{\reg{E}_2^{w'}}
\otimes 
\frac{1}{\sqrt{2^\lambda}} 
\sum_{k \notin \set{\vec{x}_{[1:\lambda]}} }
\ket{k}_{\reg{K}}.
\end{multline}
\noindent Now, we define the following isometry
\begin{equation*}
    W^{\mathsf{AppK}}: \ket{R}_{\reg{E}_1} 
    \mapsto \ket{R}_{\reg{E}_1} \otimes \frac{1}{\sqrt{2^\lambda - |\set{\vec{x}_{[1:\lambda]}}|}} \sum_{\substack{k \in \bit^\lambda: \\ k \notin \set{\vec{x}_{[1:\lambda]}}}} \ket{k}_{\reg{K}}.
\end{equation*}

\noindent Applying $W^{\mathsf{AppK}}$ to \Cref{eq:PRFS_PRw}, we get $\ket{\psi_4}$ as

\begin{multline} 
\ket{\psi_4} =
\sum_{
\substack{
\vec{x}\in[N]^t, 
\vec{w}\in[M]^t, \\ 
(y_1,\dots,y_t,z_1,\dots,z_t) \in [N]^{2t}_\dist}
}
\frac{1}{\sqrt{N^{\downarrow 2t}}}
\qty( \prod_{i=1}^t \ketbra{w_i}{w_i}_{\reg{A}} \otimes \ket{z_i}_{\reg{D}_i}
\cdot A_{2i} \ketbra{y_i}{x_i}_{\reg{AB}} A_{2i-1} )
\ket{0}_{\reg{ABC}} \\
\otimes 
\ket{\vec{w}}_{\reg{Q}} 
\ket{\set{(x_j,y_j)}_{j\in[t]}}_{\reg{E}_1}
\bigotimes_{w'\in\bit^m} \ket{\set{(\overline{w}_i,z_i):w_i=w'}}_{\reg{E}_2^{w'}}
\otimes 
\frac{1}{\sqrt{2^\lambda - |\set{\vec{x}_{[1:\lambda]}}|}} 
\sum_{k \notin \set{\vec{x}_{[1:\lambda]}} }
\ket{k}_{\reg{K}}.
\end{multline}
\noindent Notice that 

\begin{multline}
\ket{\psi_5} = 
\sum_{
\substack{
\vec{x}\in[N]^t, 
\vec{w}\in[M]^t, \\ 
(y_1,\dots,y_t,z_1,\dots,z_t) \in [N]^{2t}_\dist}
}\sqrt{\frac{2^\secp - |\set{\vec{x}_{[1:\lambda]}}|}{2^{\secp}}}
\frac{1}{\sqrt{N^{\downarrow 2t}}}
\qty( \prod_{i=1}^t \ketbra{w_i}{w_i}_{\reg{A}} \otimes \ket{z_i}_{\reg{D}_i}
\cdot A_{2i} \ketbra{y_i}{x_i}_{\reg{AB}} A_{2i-1} )
\ket{0}_{\reg{ABC}} \\
\otimes 
\ket{\vec{w}}_{\reg{Q}} 
\ket{\set{(x_j,y_j)}_{j\in[t]}}_{\reg{E}_1}
\bigotimes_{w'\in\bit^m} \ket{\set{(\overline{w}_i,z_i):w_i=w'}}_{\reg{E}_2^{w'}}
\otimes 
\frac{1}{\sqrt{2^\lambda - |\set{\vec{x}_{[1:\lambda]}}|}} 
\sum_{k \notin \set{\vec{x}_{[1:\lambda]}} }
\ket{k}_{\reg{K}}.
\end{multline}

\noindent Since for all $\vec{x}\in[N]^t$, $$\sqrt{\frac{2^\secp - |\set{\vec{x}_{[1:\lambda]}}|}{2^{\secp}}}\geq \sqrt{\frac{2^\secp - t}{2^{\secp}}},$$
by~\Cref{lem:norm:sub-state}, the inner product between $\ket{\psi_4}$ and $\ket{\psi_5}$ is at most $\sqrt{1- t/2^{\secp}}$. Hence the required trace distance is at most $O\qty(\sqrt{\frac{t}{2^\secp}})$.
\end{proof}

\noindent $\bullet~\mathsf{Real}$: The adversary gets access to the Quantum Haar Random Oracle $U$ and $\oracle_{\mathsf{PRFS}}$. \\

\begin{lemma} \label{lem:PRFS_Hyb5_Real}
    $\qty|\Pr_{\hybrid\,5}[1 \gets \Adversary^{\mathsf{PR},\oracle_{\mathsf{Hyb}5}}] - \Pr_{\mathsf{Real}}[1 \gets \Adversary^{U,\oracle_{\mathsf{PRFS}}}]| = O\qty(\frac{t^2}{2^n})$.
\end{lemma}
\begin{proof}[Proof of~\Cref{lem:PRFS_Hyb5_Real}]
    This is true by~\Cref{thm:MH24}.
\end{proof}

\noindent Collecting the probabilities, the distinguishing advantage of $\Adversary$ is at most $O\qty(\frac{t^2}{2^{n-m}} + \frac{t^2}{\sqrt{2^n}} + \sqrt{\frac{t}{2^{\secp}}})$ as desired. This completes the proof of~\Cref{thm:prfs}.
\end{proof}

\fi

\section*{Acknowledgements}
\noindent This work is supported by the National Science Foundation under Grant No. 2329938 and Grant No.
2341004.
\noindent JB is supported by Henry Yuen's AFORS (award FA9550-21-1-036) and NSF CAREER (award CCF2144219). This work was done in part when JB was visiting the Simons Institute for the Theory of Computing, supported by NSF QLCI Grant No. 2016245.

\printbibliography

@misc{MH24,
      title={How to Construct Random Unitaries}, 
      author={Fermi Ma and Hsin-Yuan Huang},
      year={2024},
      eprint={2410.10116},
      archivePrefix={arXiv},
      primaryClass={quant-ph},
      url={https://arxiv.org/abs/2410.10116}, 
}

@article{WildeBook,
  title={From classical to quantum Shannon theory},
  author={Wilde, Mark M},
  journal={arXiv preprint arXiv:1106.1445},
  year={2011}
}

@article{CM24,
  title={Quantum merkle trees},
  author={Chen, Lijie and Movassagh, Ramis},
  journal={Quantum},
  volume={8},
  pages={1380},
  year={2024},
  publisher={Verein zur F{\"o}rderung des Open Access Publizierens in den Quantenwissenschaften}
}

@misc{BBOSS24,
      title={Signatures From Pseudorandom States via $\bot$-PRFs}, 
      author={Mohammed Barhoush and Amit Behera and Lior Ozer and Louis Salvail and Or Sattath},
      year={2024},
      eprint={2311.00847},
      archivePrefix={arXiv},
      primaryClass={cs.CR},
      url={https://arxiv.org/abs/2311.00847}, 
}

@article{BJ24,
  title={Commitments are equivalent to one-way state generators},
  author={Batra, Rishabh and Jain, Rahul},
  journal={FOCS 2024 (to appear)},
  URL={https://arxiv.org/abs/2404.03220},
  year={2024}
}

@inproceedings{AMR20,
  title={Efficient simulation of random states and random unitaries},
  author={Alagic, Gorjan and Majenz, Christian and Russell, Alexander},
  booktitle={Advances in Cryptology--EUROCRYPT 2020: 39th Annual International Conference on the Theory and Applications of Cryptographic Techniques, Zagreb, Croatia, May 10--14, 2020, Proceedings, Part III 39},
  pages={759--787},
  year={2020},
  organization={Springer}
}

@article{BHH16,
  title={Local random quantum circuits are approximate polynomial-designs},
  author={Brandao, Fernando GSL and Harrow, Aram W and Horodecki, Micha{\l}},
  journal={Communications in Mathematical Physics},
  volume={346},
  pages={397--434},
  year={2016},
  publisher={Springer}
}

@misc{Green20,
title={What is the random oracle model and why should you care?},
author={Matthew Green},
note={\url{https://blog.cryptographyengineering.com/2020/01/05/what-is-the-random-oracle-model-and-why-should-you-care-part-5/}},
year={2020}
}

@inproceedings{BR93,
  title={Random oracles are practical: A paradigm for designing efficient protocols},
  author={Bellare, Mihir and Rogaway, Phillip},
  booktitle={Proceedings of the 1st ACM Conference on Computer and Communications Security},
  pages={62--73},
  year={1993}
}

@inproceedings{KQST23,
  title={Quantum cryptography in algorithmica},
  author={Kretschmer, William and Qian, Luowen and Sinha, Makrand and Tal, Avishay},
  booktitle={Proceedings of the 55th Annual ACM Symposium on Theory of Computing},
  pages={1589--1602},
  year={2023}
}

@article{chia2024quantum,
  title={Quantum State Learning Implies Circuit Lower Bounds},
  author={Chia, Nai-Hui and Liang, Daniel and Song, Fang},
  journal={arXiv preprint arXiv:2405.10242},
  year={2024}
}

@inproceedings{ABFGVZZ24,
  author       = {Scott Aaronson and
                  Adam Bouland and
                  Bill Fefferman and
                  Soumik Ghosh and
                  Umesh V. Vazirani and
                  Chenyi Zhang and
                  Zixin Zhou},
  editor       = {Venkatesan Guruswami},
  title        = {Quantum Pseudoentanglement},
  booktitle    = {15th Innovations in Theoretical Computer Science Conference, {ITCS}
                  2024, January 30 to February 2, 2024, Berkeley, CA, {USA}},
  series       = {LIPIcs},
  volume       = {287},
  pages        = {2:1--2:21},
  publisher    = {Schloss Dagstuhl - Leibniz-Zentrum f{\"{u}}r Informatik},
  year         = {2024},
  url          = {https://doi.org/10.4230/LIPIcs.ITCS.2024.2},
  doi          = {10.4230/LIPICS.ITCS.2024.2},
  bibsource    = {dblp computer science bibliography, https://dblp.org}
}

@article{huang2022quantum,
  title={Quantum advantage in learning from experiments},
  author={Huang, Hsin-Yuan and Broughton, Michael and Cotler, Jordan and Chen, Sitan and Li, Jerry and Mohseni, Masoud and Neven, Hartmut and Babbush, Ryan and Kueng, Richard and Preskill, John and others},
  journal={Science},
  volume={376},
  number={6598},
  pages={1182--1186},
  year={2022},
  publisher={American Association for the Advancement of Science}
}

@article{MPSY24,
  title={Pseudorandom unitaries with non-adaptive security},
  author={Metger, Tony and Poremba, Alexander and Sinha, Makrand and Yuen, Henry},
  journal={FOCS 2024 (to appear)},
  URL={https://arxiv.org/abs/2402.14803},
  year={2024}
}

@misc{AGL24,
  author = {Prabhanjan Ananth and Aditya Gulati and Yao-Ting Lin},
  title = {Cryptography in the Common Haar State Model: Feasibility Results and Separations},
  howpublished = {Cryptology ePrint Archive, Paper 2024/1043},
  year = {2024},
  note = {\textit{To appear in TCC 2024}},
  url = {https://eprint.iacr.org/2024/1043}
}

@inproceedings{WB24,
  title={Quantum Event Learning and Gentle Random Measurements},
  author={Watts, Adam Bene and Bostanci, John},
  booktitle={15th Innovations in Theoretical Computer Science Conference (ITCS 2024)},
  year={2024},
  organization={Schloss-Dagstuhl-Leibniz Zentrum f{\"u}r Informatik}
}

@article{Har23,
  title={Approximate orthogonality of permutation operators, with application to quantum information},
  author={Harrow, Aram W},
  journal={Letters in Mathematical Physics},
  volume={114},
  number={1},
  pages={1},
  year={2023},
  publisher={Springer}
}

@inproceedings{KT24,
  title={Commitments from quantum one-wayness},
  author={Khurana, Dakshita and Tomer, Kabir},
  booktitle={Proceedings of the 56th Annual ACM Symposium on Theory of Computing},
  pages={968--978},
  year={2024}
}

@inproceedings{Yan22,
  title={General properties of quantum bit commitments},
  author={Yan, Jun},
  booktitle={International Conference on the Theory and Application of Cryptology and Information Security},
  pages={628--657},
  year={2022},
  organization={Springer}
}

@inproceedings{BCQ23,
  title={On the Computational Hardness Needed for Quantum Cryptography},
  author={Brakerski, Zvika and Canetti, Ran and Qian, Luowen},
  booktitle={14th Innovations in Theoretical Computer Science Conference, ITCS 2023},
  pages={24},
  year={2023},
  organization={Schloss Dagstuhl-Leibniz-Zentrum fur Informatik GmbH, Dagstuhl Publishing}
}

@misc{HMY23,
      title={A Note on Output Length of One-Way State Generators}, 
      author={Minki Hhan and Tomoyuki Morimae and Takashi Yamakawa},
      year={2023},
      eprint={2312.16025},
      archivePrefix={arXiv},
      primaryClass={quant-ph}
}

@inproceedings{BBSS23,
  author       = {Amit Behera and
                  Zvika Brakerski and
                  Or Sattath and
                  Omri Shmueli},
  editor       = {Guy N. Rothblum and
                  Hoeteck Wee},
  title        = {Pseudorandomness with Proof of Destruction and Applications},
  booktitle    = {Theory of Cryptography - 21st International Conference, {TCC} 2023,
                  Taipei, Taiwan, November 29 - December 2, 2023, Proceedings, Part
                  {IV}},
  series       = {Lecture Notes in Computer Science},
  volume       = {14372},
  pages        = {125--154},
  publisher    = {Springer},
  year         = {2023},
  url          = {https://doi.org/10.1007/978-3-031-48624-1\_5},
  doi          = {10.1007/978-3-031-48624-1\_5},
  bibsource    = {dblp computer science bibliography, https://dblp.org}
}

@inproceedings{GJMZ23,
  title={Commitments to quantum states},
  author={Gunn, Sam and Ju, Nathan and Ma, Fermi and Zhandry, Mark},
  booktitle={Proceedings of the 55th Annual ACM Symposium on Theory of Computing},
  pages={1579--1588},
  year={2023}
}

@article{schuster2024random,
  title={Random unitaries in extremely low depth},
  author={Schuster, Thomas and Haferkamp, Jonas and Huang, Hsin-Yuan},
  journal={arXiv preprint arXiv:2407.07754},
  year={2024}
}

@article{BEMPQY23,
  title={Unitary complexity and the uhlmann transformation problem},
  author={Bostanci, John and Efron, Yuval and Metger, Tony and Poremba, Alexander and Qian, Luowen and Yuen, Henry},
  journal={arXiv preprint arXiv:2306.13073},
  year={2023}
}

@misc{KT24b,
      author = {Dakshita Khurana and Kabir Tomer},
      title = {Founding Quantum Cryptography on Quantum Advantage, or, Towards Cryptography from $\#\mathsf{P}$-Hardness},
      howpublished = {Cryptology {ePrint} Archive, Paper 2024/1490},
      year = {2024},
      url = {https://eprint.iacr.org/2024/1490}
}

@article{arute2019quantum,
  title={Quantum supremacy using a programmable superconducting processor},
  author={Arute, Frank and Arya, Kunal and Babbush, Ryan and Bacon, Dave and Bardin, Joseph C and Barends, Rami and Biswas, Rupak and Boixo, Sergio and Brandao, Fernando GSL and Buell, David A and others},
  journal={Nature},
  volume={574},
  number={7779},
  pages={505--510},
  year={2019},
  publisher={Nature Publishing Group}
}

@inproceedings{MY21,
  title={Quantum commitments and signatures without one-way functions},
  author={Morimae, Tomoyuki and Yamakawa, Takashi},
  booktitle={Annual International Cryptology Conference},
  pages={269--295},
  year={2022},
  organization={Springer}
}

@article{CG24,
  title={The random oracle methodology, revisited},
  author={Canetti, Ran and Goldreich, Oded and Halevi, Shai},
  journal={Journal of the ACM (JACM)},
  volume={51},
  number={4},
  pages={557--594},
  year={2004},
  publisher={ACM New York, NY, USA}
}

@article{haferkamp2022random,
  title={Random quantum circuits are approximate unitary $t$-designs in depth $O\left( nt^{5+o(1)} \right)$},
  author={Haferkamp, Jonas},
  journal={Quantum},
  volume={6},
  pages={795},
  year={2022},
  publisher={Verein zur F{\"o}rderung des Open Access Publizierens in den Quantenwissenschaften}
}

@inproceedings{AGQY22,
  title={Pseudorandom (Function-Like) Quantum State Generators: New Definitions and Applications},
  author={Ananth, Prabhanjan and Gulati, Aditya and Qian, Luowen and Yuen, Henry},
  booktitle={Theory of Cryptography Conference},
  pages={237--265},
  year={2022},
  organization={Springer}
}

@inproceedings{JLS18,
  author    = {Zhengfeng Ji and
               Yi{-}Kai Liu and
               Fang Song},
  editor    = {Hovav Shacham and
               Alexandra Boldyreva},
  title     = {Pseudorandom Quantum States},
  booktitle = {Advances in Cryptology - {CRYPTO} 2018 - 38th Annual International
               Cryptology Conference, Santa Barbara, CA, USA, August 19-23, 2018,
               Proceedings, Part {III}},
  series    = {Lecture Notes in Computer Science},
  volume    = {10993},
  pages     = {126--152},
  publisher = {Springer},
  year      = {2018},
  doi       = {10.1007/978-3-319-96878-0_5},
  bibsource = {dblp computer science bibliography, https://dblp.org}
}

@book{nielsen_chuang_2010, place={Cambridge}, title={Quantum Computation and Quantum Information: 10th Anniversary Edition}, DOI={10.1017/CBO9780511976667}, publisher={Cambridge University Press}, author={Nielsen, Michael A. and Chuang, Isaac L.}, year={2010}}

@inproceedings{BS19,
  author    = {Zvika Brakerski and
               Omri Shmueli},
  editor    = {Dennis Hofheinz and
               Alon Rosen},
  title     = {(Pseudo) Random Quantum States with Binary Phase},
  booktitle = {Theory of Cryptography - 17th International Conference, {TCC} 2019,
               Nuremberg, Germany, December 1-5, 2019, Proceedings, Part {I}},
  series    = {Lecture Notes in Computer Science},
  volume    = {11891},
  pages     = {229--250},
  publisher = {Springer},
  year      = {2019},
  doi       = {10.1007/978-3-030-36030-6_10},
  bibsource = {dblp computer science bibliography, https://dblp.org}
}

@inproceedings{AQY21,
      title={Cryptography from Pseudorandom Quantum States.}, 
      author={Ananth, Prabhanjan and Qian, Luowen and Yuen, Henry},
      booktitle={CRYPTO},
      year={2022}
}

@inproceedings{BrakerskiS20,
  author    = {Zvika Brakerski and
               Omri Shmueli},
  editor    = {Daniele Micciancio and
               Thomas Ristenpart},
  title     = {Scalable Pseudorandom Quantum States},
  booktitle = {Advances in Cryptology - {CRYPTO} 2020 - 40th Annual International
               Cryptology Conference, {CRYPTO} 2020, Santa Barbara, CA, USA, August
               17-21, 2020, Proceedings, Part {II}},
  series    = {Lecture Notes in Computer Science},
  volume    = {12171},
  pages     = {417--440},
  publisher = {Springer},
  year      = {2020},
  doi       = {10.1007/978-3-030-56880-1_15},
  bibsource = {dblp computer science bibliography, https://dblp.org}
}

@inproceedings{Kretschmer21,
  author    = {William Kretschmer},
  editor    = {Min{-}Hsiu Hsieh},
  title     = {Quantum Pseudorandomness and Classical Complexity},
  booktitle = {16th Conference on the Theory of Quantum Computation, Communication
               and Cryptography, {TQC} 2021, July 5-8, 2021, Virtual Conference},
  series    = {LIPIcs},
  volume    = {197},
  pages     = {2:1--2:20},
  publisher = {Schloss Dagstuhl - Leibniz-Zentrum f{\"{u}}r Informatik},
  year      = {2021},
  doi       = {10.4230/LIPIcs.TQC.2021.2},
  bibsource = {dblp computer science bibliography, https://dblp.org}
}

@article{chen2024efficient,
  title={Efficient unitary designs and pseudorandom unitaries from permutations},
  author={Chen, Chi-Fang and Bouland, Adam and Brand{\~a}o, Fernando GSL and Docter, Jordan and Hayden, Patrick and Xu, Michelle},
  journal={arXiv preprint arXiv:2404.16751},
  year={2024}
}

@article{brakerski2024real,
  title={Real-Valued Somewhat-Pseudorandom Unitaries},
  author={Brakerski, Zvika and Magrafta, Nir},
  journal={TCC 2024 (to appear)},
  URL={https://arxiv.org/abs/2403.16704},
  year={2024}
}

@article{LQSYZ23,
  title={Quantum pseudorandom scramblers},
  author={Lu, Chuhan and Qin, Minglong and Song, Fang and Yao, Penghui and Zhao, Mingnan},
  journal={TCC 2024 (to appear)},
  URL={https://arxiv.org/abs/2309.08941},
  year={2023}
}

@inproceedings{BFV19,
  author    = {Adam Bouland and
               Bill Fefferman and
               Umesh V. Vazirani},
  editor    = {Thomas Vidick},
  title     = {Computational Pseudorandomness, the Wormhole Growth Paradox, and Constraints
               on the AdS/CFT Duality (Abstract)},
  booktitle = {11th Innovations in Theoretical Computer Science Conference, {ITCS}
               2020, January 12-14, 2020, Seattle, Washington, {USA}},
  series    = {LIPIcs},
  volume    = {151},
  pages     = {63:1--63:2},
  publisher = {Schloss Dagstuhl - Leibniz-Zentrum f{\"{u}}r Informatik},
  year      = {2020},
  doi       = {10.4230/LIPIcs.ITCS.2020.63},
  bibsource = {dblp computer science bibliography, https://dblp.org}
}

@article{chen2024power,
  title={The power of a single Haar random state: constructing and separating quantum pseudorandomness},
  author={Chen, Boyang and Coladangelo, Andrea and Sattath, Or},
  journal={arXiv preprint arXiv:2404.03295},
  year={2024}
}

@inproceedings{ananth2023pseudorandom,
  title={Pseudorandom strings from pseudorandom quantum states},
  author={Ananth, Prabhanjan and Lin, Yao-Ting and Yuen, Henry},
  booktitle={ITCS},
  year={2024}
}

@inproceedings{AGKL,
  title={Pseudorandom isometries},
  author={Ananth, Prabhanjan and Gulati, Aditya and Kaleoglu, Fatih and Lin, Yao-Ting},
  booktitle={Annual International Conference on the Theory and Applications of Cryptographic Techniques},
  pages={226--254},
  year={2024},
  organization={Springer}
}

@inproceedings{harrow2017sequential,
  title={Sequential measurements, disturbance and property testing},
  author={Harrow, Aram W and Lin, Cedric Yen-Yu and Montanaro, Ashley},
  booktitle={Proceedings of the Twenty-Eighth Annual ACM-SIAM Symposium on Discrete Algorithms},
  pages={1598--1611},
  year={2017},
  organization={SIAM}
}

@inproceedings{zhandry2019record,
  title={How to record quantum queries, and applications to quantum indifferentiability},
  author={Zhandry, Mark},
  booktitle={Advances in Cryptology--CRYPTO 2019: 39th Annual International Cryptology Conference, Santa Barbara, CA, USA, August 18--22, 2019, Proceedings, Part II 39},
  pages={239--268},
  year={2019},
  organization={Springer}
}

@article{bostanci2024efficient,
  title={Efficient Quantum Pseudorandomness from Hamiltonian Phase States},
  author={Bostanci, John and Haferkamp, Jonas and Hangleiter, Dominik and Poremba, Alexander},
  journal={arXiv preprint arXiv:2410.08073},
  year={2024}
}
\appendix

\ifllncs

    \section{Omitted Proofs}

\ifllncs
\subsection{Omitted Proofs in~\Cref{sec:prelim}}    \label{app:prelim}
\begin{proof}[Proof of~\Cref{lem:norm:cs}]
    By Cauchy-Schwarz inequality, $$|\braket{\psi}{\phi}|^2\leq \|\ket{\psi}\|^2\|\ket{\phi}\|^2$$
    Since, $\ket{\psi}$ is a quantum state, $\|\ket{\psi}\|=1$, hence, $$|\braket{\psi}{\phi}|^2\leq \|\ket{\psi}\|^2$$
    Hence, $$\|\ket{\psi}\|^2\geq |\braket{\psi}{\phi}|^2. $$
\end{proof}

\begin{proof}[Proof of~\Cref{lem:norm:sub-state}]
    Notice that $$\braket{\psi}{\phi} = \sum_i \alpha_i\braket{\psi_i}{\psi_i}.$$ Since $\alpha_i\geq \beta$ and $\braket{\psi_i}{\psi_i} \geq 0$, hence $$\braket{\psi}{\phi} \geq \beta\sum_i\braket{\psi_i}{\psi_i}.$$ Since $\ket{\psi}$ is a quantum state, $\braket{\psi}{\psi} = 1$, hence $\sum_i \braket{\psi_i}{\psi_i} = 1$. Hence $$\braket{\psi}{\phi} \geq \beta.$$
\end{proof}
\else
\fi

\subsection{Omitted Proofs in~\Cref{sec:PR_New}}    \label{app:PR_New}

\begin{proof}[Proof of~\Cref{lem:cf_set_size}]
    The total number of constraints impossed by $i$-fold collision-freeness is ${|S| \choose i}{|S| \choose i-1}$. Each constraint removes at most $2^{n-\secp}$ elements from $\cfreeset_{\ell,\secp}(S)$. Hence the total number of elements not in $\cfreeset_{\ell,\secp}(S)$ is less than
    \begin{equation*}
    \begin{split}
        & 2^{n-\secp} \cdot \left(\sum_{i=1}^{\ell}{|S| \choose i}{|S| \choose i-1}\right)
        \leq 2^{n-\secp} \cdot \left(\sum_{i=1}^{\ell}\left(\frac{e|S|}{i}\right)^{2i}\right) \\
        & \leq 2^{n-\secp} \cdot \left(\sum_{i=1}^{\ell}\left(\frac{e|S|}{\ell}\right)^{2\ell}\right)\leq 2^{n-\secp}\cdot\ell|S|^{2\ell}.
    \end{split}
    \end{equation*}
\end{proof}

\ifllncs

\else

\fi

\subsection{Omitted Proofs in~\Cref{sec:pru:shortkeys}}     \label{app:pru:shortkeys}

\begin{proof}[Proof of~\Cref{cor:proj:key}]
Fix $\Vec{x}\in [N]^t$, $\Vec{y}\in [N]^t_{\dist}$ and $ \Vec{z}\in ([N]\setminus\set{\Vec{y}})^\ell_{\dist}$.
Let for any $k\in\bit^n$, 
$$R_k := 
\set{(x_i,z_i)}_{i\in\bfa} \uplus 
\set{(z_i\oplus k,y_i)}_{i\in\bfa} \uplus 
\set{(x_i,y_i)}_{i\in\bfb}.$$
Then, we want to show that 
        \begin{equation*}
        \sum_{\substack{k\in\bit^n}} \Pi^{\mathrm{good}}_{\reg{E}_1\reg{K}} \ket{R_k}_{\reg{E}_1}\ket{k}_{\reg{K}} = \sum_{\substack{k\in\bit^n\setminus \left(\set{\Vec{x}}\oplus\set{\Vec{y}}\right)\cup\left(\set{\Vec{x}}\oplus\set{\Vec{z}}\right)}}\ket{R_k}_{\reg{E}_1}\ket{k}_{\reg{K}}
        \end{equation*}

        \noindent We show this in 3 parts, 
        \begin{enumerate}
            \item If $k\in\set{\Vec{x}}\oplus\set{\Vec{y}}$, then $\Pi^{\mathrm{good}}_{\reg{E}_1\reg{K}} \ket{R_k}_{\reg{E}_1}\ket{k}_{\reg{K}} = 0$.
            \item If $k\in\set{\Vec{x}}\oplus\set{\Vec{z}}$, then $\Pi^{\mathrm{good}}_{\reg{E}_1\reg{K}} \ket{R_k}_{\reg{E}_1}\ket{k}_{\reg{K}} = 0$.
            \item If $k\in\bit^n\setminus \left(\set{\Vec{x}}\oplus\set{\Vec{y}}\right)\cup\left(\set{\Vec{x}}\oplus\set{\Vec{z}}\right)$, then $\Pi^{\mathrm{good}}_{\reg{E}_1\reg{K}} \ket{R_k}_{\reg{E}_1}\ket{k}_{\reg{K}} = \ket{R_k}_{\reg{E}_1}\ket{k}_{\reg{K}}$.
        \end{enumerate}

        \begin{myclaim}
            If $k\in\set{\Vec{x}}\oplus\set{\Vec{y}}$, then $\Pi^{\mathrm{good}}_{\reg{E}_1\reg{K}} \ket{R_k}_{\reg{E}_1}\ket{k}_{\reg{K}} = 0$.
        \end{myclaim}

        \begin{proof}
            Let $k=y_{j_1}\oplus x_{j_2}$ for some $j_1,j_2\in [t]$, then we can divide this into the following cases:
            \begin{itemize}
                \item Let $j_1,j_2\in\bfa$, then $$\set{((x_i,z_i),(z_i\oplus k,y_i))}_{i\in\bfa}\cup\set{((z_{j_1}\oplus k,y_{j_1}), (x_{j_2},z_{j_2}))}\subseteq\CorX(R_k,k),$$ hence $|\CorX(R_k,k)|\geq\ell+1$. 
                \item Let $j_1\in\bfa,j_2\in\bfb$, then $$\set{((x_i,z_i),(z_i\oplus k,y_i))}_{i\in\bfa}\cup\set{((z_{j_1}\oplus k,y_{j_1}), (x_{j_2},y_{j_2}))}\subseteq\CorX(R_k,k),$$ hence $|\CorX(R_k,k)|\geq\ell+1$. 
                \item Let $j_1\in\bfa,j_2\in\bfb$, then symmetric to the above case $|\CorX(R_k,k)|\geq\ell+1$. 
                \item Let $j_1,j_2\in\bfb$, then $$\set{((x_i,z_i),(z_i\oplus k,y_i))}_{i\in\bfa}\cup\set{((x_{j_1},y_{j_1}), (x_{j_2},y_{j_2}))}\subseteq\CorX(R_k,k),$$ hence $|\CorX(R_k,k)|\geq\ell+1$.
            \end{itemize}
        \end{proof}

        \begin{myclaim}
            If $k\in\set{\Vec{x}}\oplus\set{\Vec{z}}$, then $\Pi^{\mathrm{good}}_{\reg{E}_1\reg{K}} \ket{R_k}_{\reg{E}_1}\ket{k}_{\reg{K}} = 0$.
        \end{myclaim}

        \begin{proof}
            Let $k=z_{j_1}\oplus x_{j_2}$ for some $j_2\in [t],j_1\in\bfa$, then we can divide this into the following cases:
            \begin{itemize}
                \item Let $j_1,j_2\in\bfa$, then $$\set{((x_i,z_i),(z_i\oplus k,y_i))}_{i\in\bfa}\cup\set{((x_{j_1},z_{j_1}), (x_{j_2},z_{j_2}))}\subseteq\CorX(R_k,k),$$ hence $|\CorX(R_k,k)|\geq\ell+1$. 
                \item Let $j_1\in\bfa,j_2\in\bfb$, then $$\set{((x_i,z_i),(z_i\oplus k,y_i))}_{i\in\bfa}\cup\set{((x_{j_1}\oplus k,z_{j_1}), (x_{j_2},y_{j_2}))}\subseteq\CorX(R_k,k),$$ hence $|\CorX(R_k,k)|\geq\ell+1$. 
            \end{itemize}
        \end{proof}

        \begin{myclaim}
            If $k\in\bit^n\setminus \left(\set{\Vec{x}}\oplus\set{\Vec{y}}\right)\cup\left(\set{\Vec{x}}\oplus\set{\Vec{z}}\right)$, then 
            \begin{equation*}
                \Pi^{\mathrm{good}}_{\reg{E}_1\reg{K}} \ket{R_k}_{\reg{E}_1}\ket{k}_{\reg{K}} = \ket{R_k}_{\reg{E}_1}\ket{k}_{\reg{K}}.
            \end{equation*}
        \end{myclaim}
        \begin{proof}
        We complete the proof by a case analysis. For any $(p_0,p_1) \in R_k \times R_k$, there are $9$ cases depending on the sets which $p_0$ and $p_1$ respectively belong to (see~\Cref{tab:4x4_table}).
            \begin{table}[H]
            \centering
            \begin{tabular}{|c|c|c|c|}  
            \hline
                 \diagbox{$p_0$ \\ $\rotatebox[origin=c]{270}{$\in$}$}{$p_1\in$} & $\set{(x_j,z_j)}_{j\in\bfa}$ & $\set{(z_j\oplus k,y_j)}_{j\in\bfa}$ & $\set{(x_j,y_j)}_{j\in\bfb}$ \\  
            \hline
                $\set{(x_i,z_i)}_{i\in\bfa}$ & $E_1:= \set{z_i}_{i\in\bfa}\oplus\set{x_j}_{j\in\bfa}$ & $E_2:= \set{z_i}_{i\in\bfa}\oplus\set{z_j\oplus k}_{j\in\bfa}$ & $E_3:= \set{z_i}_{i\in\bfa}\oplus\set{x_j}_{j\in\bfb}$ \\ 
            \hline
                $\set{(z_i\oplus k,y_i)}_{i\in\bfa}$ & $E_4:= \set{y_i}_{i\in\bfa}\oplus\set{x_j}_{j\in\bfa}$ & $E_5:= \set{y_i}_{i\in\bfa}\oplus\set{z_j\oplus k}_{j\in\bfa}$ & $E_6:= \set{y_i}_{i\in\bfa}\oplus\set{x_j}_{j\in\bfb}$ \\
            \hline
                $\set{(x_i,y_i)}_{i\in\bfb}$ & $E_7:= \set{y_i}_{i\in\bfb}\oplus\set{x_j}_{j\in\bfa}$ & $E_8:= \set{y_i}_{i\in\bfb}\oplus\set{z_j\oplus k}_{j\in\bfa}$ & $E_9:= \set{y_i}_{i\in\bfb}\oplus\set{x_j}_{j\in\bfb}$ \\ 
            \hline
        \end{tabular}   \caption{Cases analysis for $(p_0, p_1)$.}    \label{tab:4x4_table}
        \end{table}
        We will show that $|E_2| = \ell$ and $|E_i| = 0$ for $i \neq 2$.
        \begin{itemize}
            \item Since $\Vec{z}$ has pairwise distinct coordinates, $z_i \oplus (z_j \oplus k) = k$ if and only if $i = j$. Hence, $|E_2| = |\bfa| = \ell$.
            \item Since $k\not\in\set{\Vec{x}}\oplus\set{\Vec{y}}$, then $E_4=E_6=E_7=E_9=\set{}$.
            \item Since $k\not\in\set{\Vec{x}}\oplus\set{\Vec{z}}$, then $E_1=E_3=\set{}$.
            \item Since for any $i\in [t], j\in\bfa$, $y_i\neq z_j$, $y_i\oplus z_j\oplus k\neq k$, hence $E_2=E_5=\set{}$.
        \end{itemize}
        As a result, we have $|\CorX(R_k,k)| = \sum_{i=1}^9 |E_i| = \ell$.
        \end{proof}
        Combining the above claims completes the proof of~\Cref{cor:proj:key}.
\end{proof}

\ifllncs
\subsection{Omitted Proofs in~\Cref{sec:unitarydesign:shortkeys}}    \label{app:unitarydesign:shortkeys}

\else

\fi

\else
    
\fi

\end{document}